\documentclass[12pt,reqno]{amsart}
\usepackage[left=1.25in,right=1.25in,top=1.25in,bottom=1.25in]{geometry}
\usepackage{color}
\usepackage[dvipsnames]{xcolor}
\usepackage{amsmath, amsthm, amssymb, amsfonts}
\usepackage{natbib}
\usepackage{har2nat}
\usepackage{tabularx}
\usepackage{parskip}
\usepackage[british]{babel}
\usepackage{mathtools}
\usepackage{xfrac}
\usepackage[colorlinks=true,allcolors=Blue,linkcolor=Blue]{hyperref}
\usepackage{dsfont}
\usepackage[onehalfspacing]{setspace}

\usepackage{tikz}
\usepackage{caption}
\usepackage{pgfplots}
\usepackage{float}
\restylefloat{table}
\usepgfplotslibrary{fillbetween}
\usetikzlibrary{patterns}
\pgfplotsset{compat=newest}
\usetikzlibrary{decorations.pathreplacing,angles,quotes}
\usepackage{graphicx}
\usepackage{subcaption}
\usepackage{threeparttable}
\allowdisplaybreaks
\newtheorem{theorem}{Theorem}
\newtheorem{corollary}{Corollary}
\newtheorem{manualtheoreminner}{Corollary}
\newenvironment{manualcorollary}[1]{%
  \renewcommand\themanualtheoreminner{#1}%
  \manualtheoreminner
}{\endmanualtheoreminner}

\newtheorem{lemma}{Lemma}

\newtheorem{proposition}{Proposition}

\newtheorem{remark}{Remark}

\theoremstyle{definition}

\newtheorem{assumption}{Assumption}
\newtheorem{example}{Example}

\newcommand{\ul}{\underline}
\newcommand{\ol}{\overline}
\newcommand{\df}{\mathrm{d}}

\newcommand{\bdis}{\begin{displaymath}}
\newcommand{\edis}{\end{displaymath}}
\newcommand{\beq}{\begin{equation}}
\newcommand{\eeq}{\end{equation}}
\newcommand{\bea}{\begin{eqnarray*}}
\newcommand{\eea}{\end{eqnarray*}}
\newcommand{\bean}{\begin{eqnarray}}
\newcommand{\eean}{\end{eqnarray}}

\newcommand{\R}{\mathbb{R}}
\newcommand{\N}{\mathbb{N}}
\newcommand{\E}{\mathbb{E}}

\DeclareMathOperator\supp{supp}

\DeclareMathOperator*{\argmax}{arg\,max}

\newcommand{\1}{\mbox{\bf 1}}

\begin{document}
\title{Persuasion with Non-Linear Preferences}
\author[\uppercase{Kolotilin, Corrao, and Wolitzky}]{\larger \textsc{Anton Kolotilin, Roberto Corrao, and Alexander Wolitzky}}
\date{\today}
\thanks{ \ \\
\textit{Kolotilin}: School of Economics, UNSW Business School. \\
\textit{Corrao} and \textit{Wolitzky}: Department of Economics, MIT.
\ \\
This paper was previously circulated with the title ``Persuasion as Matching.'' It supercedes the earlier paper ``Assortative Information Disclosure'' by Kolotilin and Wolitzky. For helpful comments and suggestions, we thank Jak{\v{s}}a Cvitani{\'c}, Jeffrey Ely, Piotr Dworczak, Drew Fudenberg, Emir Kamenica, Elliot Lipnowski, Stephen Morris, Paula Onuchic, Eran Shmaya, and Andriy Zapechelnyuk, as well as many seminar participants. We thank Daniel Clark and Yucheng Shang for excellent research assistance. Anton Kolotilin gratefully acknowledges support from the
Australian Research Council Discovery Early Career Research Award
DE160100964 and from MIT Sloan's Program on Innovation in Markets and
Organizations. Alexander Wolitzky gratefully acknowledges support from NSF CAREER Award 1555071 and Sloan Foundation Fellowship 2017-9633.}

\begin{abstract}
In persuasion problems where the receiver's action is one-dimensional and his utility is single-peaked, optimal signals are characterized by duality, based on a first-order approach to the receiver's problem. A signal is optimal iff the induced joint distribution over states and actions is supported on a compact set (the \emph{contact set}) where the dual constraint binds. A signal that pools at most two states in each realization is always optimal, and such \emph{pairwise} signals are the only solutions under a non-singularity condition on utilities (the \emph{twist condition}). We provide conditions under which higher actions are induced at more or less extreme pairs of states. Finally, we provide conditions for the optimality of either full disclosure or \emph{negative assortative disclosure}, where signal realizations can be ordered from least to most extreme. Optimal negative assortative disclosure is characterized as the solution to a pair of ordinary differential equations.

\ \\

\noindent\emph{JEL\ Classification:}\ C78, D82, D83 \newline


\ 
\end{abstract}

\maketitle
\thispagestyle{empty}

\let\MakeUppercase\relax %

\newpage
\clearpage
\pagenumbering{arabic} 
\section{Introduction}

\label{s:intro}

Following the seminal papers of \citet{RS} and \citet{KG}, the past decade has witnessed an explosion of interest in the design of optimal information disclosure policies, or Bayesian persuasion. However, while significant progress has been made in the special case where the sender's and receiver's utilities are linear in the unknown state (\citealt{GK-RS}, \citealt{KMZL}, \citealt{Kolotilin2017}, \citealt{DM})---so that a distribution over states is effectively summarized by its mean---general results beyond this simple case remain scarce.

This paper reports progress on persuasion with non-linear preferences. We consider a standard persuasion problem with one sender and one receiver, where the receiver's action and the state of the world are both one-dimensional, and the receiver's expected utility is single-peaked in his action for any belief about the state. In this model, the receiver's action is optimal iff his expected marginal utility from increasing his action equals zero: that is, iff the receiver's first-order condition holds. The validity of such a \emph{first-order approach} is key for tractability. We provide four types of results.

First, a signal (i.e., a disclosure policy or Blackwell experiment) is optimal iff the joint distribution over states $\theta$ and actions $a$ that it induces is supported on a compact set $\Gamma$, which we call the \emph{contact set}. The contact set is the set of pairs $(a,\theta)$ that satisfy the dual constraint of the sender's problem with equality. In economic terms, $(a,\theta) \in \Gamma$ iff it is optimal for the sender to induce action $a$ at state $\theta$, where the sender's ``value'' for inducing $a$ at $\theta$ is equal to the sum of the sender's utility when $a$ is taken at $\theta$ and the sender's shadow value of the impact of inducing $a$ at $\theta$ on the receiver's obedience constraint when he is recommended action $a$. This technical result is the foundation for our analysis.

Second, it is always without loss to focus on \emph{pairwise signals}, where each induced posterior distribution has at most binary support. Moreover, when the contact set is pairwise---meaning that for each action $a$ there are at most two states $\theta$ such that $(a,\theta) \in \Gamma$---then every optimal signal is pairwise. We show that the contact set is pairwise under a non-singularity condition on the sender's and receiver's utilities, which we call the \emph{twist condition}. This result implies that, for example, no-disclosure is generically suboptimal whenever the support of the prior contains three or more states. It also implies several previously-known conditions for all optimal signals to be pairwise (\citealt{RS}, \citealt{AC}, \citealt{ZZ}).

Third, we ask when it is optimal for the sender to induce higher actions at more or less extreme states. That is, if the sender pools states $\theta_1, \theta_2$ and also pools states $\theta'_1, \theta'_2$, for $\theta_1 < \theta'_1 < \theta'_2 < \theta_2$, should the induced action be higher at the pair $\{\theta_1, \theta_2\}$---in which case we say that disclosure is \emph{single-dipped}, as more extreme states induce higher actions---or at the pair $\{\theta'_1, \theta'_2\}$---in which case we say that disclosure is \emph{single-peaked}? This seemingly obscure question turns out to unify a large part of the prior literature on persuasion with non-linear preferences. For instance, \citeauthor{FH}'s \citeyear{FH}\textquotedblleft matching extremes\textquotedblright\ gerrymandering
solution, where a gerrymanderer creates electoral districts that pool extreme
supporters with similarly extreme opponents, and wins those districts with
the most extreme supporters and opponents with the highest probability, is an example of single-dipped disclosure. \citeauthor{GL}'s \citeyear{GL} non-monotone stress tests, where a regulator designs a test that pools the weakest banks that it wants to receive funding with the strongest banks, pools slightly less weak banks with slightly less strong banks, and so on, such that the weakest and strongest banks receive the highest funding, is another such example. On the other hand, \citeauthor{GS2018}'s \citeyear{GS2018} \textquotedblleft nested intervals\textquotedblright\ disclosure rule, where a designer pools favorable states with similarly
unfavorable states, and persuades the receiver to take her preferred action
with higher probability at more moderate states, is an example of single-peaked disclosure.

We provide general conditions for the optimality of single-dipped disclosure (and, similarly, single-peaked disclosure), which are all based on a very simple idea. If disclosure is not single-dipped, then there must exist a \emph{single-peaked triple}: a pair of pooled state $\{\theta_1, \theta_3\}$ and an intervening state $\theta_2 \in (\theta_1, \theta_3)$ such that the induced action at $\theta_2$ (say, action $a_2$) is greater than the induced action at $\{\theta_1, \theta_3\}$ (say, action $a_1$). Our conditions ensure that any single-peaked triple can be profitably perturbed in the direction of single-dippedness by shifting weight on $\theta_1$ and $\theta_3$ from $a_1$ to $a_2$, while shifting weight on $\theta_2$ in the opposite direction. The conditions are a bit complicated in the general model, but they are very simple in leading special cases. In particular, if the receiver's optimal action equals the posterior mean state (the \emph{linear receiver case}), then single-dipped disclosure is optimal if the sender's marginal utility is convex in the state; and if the sender's utility is state-independent (the \emph{state-independent sender case}), then single-dipped disclosure is optimal if the cross-partial of the receiver's utility is log-supermodular. These conditions generalize ones in the prior literature, such as \citeauthor{FH}'s gerrymandering model and \citeauthor{BJ}'s ``martingale optimal transport'' model. We also establish a notable theoretical implication of single-dippedness/-peakedness: under some regularity conditions, whenever a strict version of this property holds, the optimal outcome is unique.

Fourth, we provide conditions for the optimality of either full disclosure, where the state is always disclosed, or \emph{negative assortative disclosure}, where the states are paired in a negatively assortative manner, so that signal realizations can be ordered from least to most extreme, and only a single state in the middle is disclosed. Intuitively, full disclosure and negative assortative disclosure represent the extremes of maximum disclosure (disclosing all states) and minimal disclosure (disclosing only one state). There is a unique full disclosure outcome, but there are many negative assortative disclosure outcomes, depending on the weights on the states in each pair. We further characterize the optimal negative assortative disclosure outcome as the solution of a pair of ordinary differential equations, and show that in some cases these equations admit an explicit solution. Notably, negative assortative disclosure is optimal whenever our conditions for the optimality of (strict) single-dipped/-peaked disclosure are satisfied and in addition the sender would rather pool any pair of states (with some non-degenerate weights) rather than separating them.

Our model and results generalize a great deal of prior literature; we give references throughout the paper. Methodologically, we rely on linear programming duality and connections to optimal transport. We build on \citet{Kolotilin2017}, which introduces the first-order approach to persuasion and the corresponding strong duality result. \citet{DM} and \citet{DizdarKovac} prove strong duality with linear preferences under weak assumptions. The linear case is important but non-generic, and the structure of optimal signals is typically very different from that in our model. \citet{DK} prove strong duality for a general persuasion problem and study its implications for multidimensional persuasion, focusing on the multidimensional linear case. \citet{KX} prove duality results for an insider-trading problem that can be shown to be mathematically equivalent to a subcase of our linear receiver case, albeit with a two-dimensional state space. \citet{GLP} use duality to study the value of data. The most related strand of the optimal transport literature is that on martingale optimal transport (e.g., \citealt{BHP}, \citealt{GHT}), which we discuss in Section \ref{s:contact}. A few recent papers apply optimal transport to persuasion, but these works are not very related to ours either methodologically or substantively.\footnote{\citet{PRS} and \citet{LinLiu} consider limited sender commitment; \citet{ABS22} and \citet{SY} consider persuasion with multiple receivers; \citet{MCS} focus on the question of when optimal signals partition a multidimensional state space.}

\section{Environment}

\label{s:model}

\subsection{Model}

We consider a standard persuasion problem, where a sender chooses a signal to reveal information to a receiver, who then takes an action. The sender's utility $V(a,\theta)$ and the receiver's utility $U(a,\theta)$ depend on the receiver's action $a\in A:=[0,1]$ and the state of the world $\theta\in \ol\Theta:=[0,1]$. The sender and receiver share a common prior $\phi \in \Delta(\ol\Theta)$, whose support is denoted by $\Theta:=\supp (\phi)$.\footnote{Throughout, for any compact metric space $X$, $\Delta(X)$ denotes the set of Borel probability measures on $X$, endowed with the weak* topology. By Theorem 12.14 in \citet{aliprantis2006}, any $\eta\in \Delta (X)$ has a well-defined support $\supp(\eta)$, which is the smallest compact set of measure one.} An \emph{outcome} $\pi \in \Delta( A \times \Theta)$ is a joint distribution over actions and states.

We impose three standard assumptions on the utility functions.  First, utilities are smooth.

\begin{assumption}
\label{a:smooth} $V(a,\theta)$ and $U(a,\theta)$ are differentiable in $a$, and the marginal utilities 
\begin{equation*}
v(a,\theta) := \frac{\partial V(a,\theta)}{\partial a}\quad\text{and}%
\quad u(a,\theta) := \frac{\partial U(a,\theta)}{\partial a}
\end{equation*}
are continuous in $(a,\theta)$. Moreover, the receiver's marginal utility $u(a,\theta)$ is differentiable in $a$, and the partial derivative $u_{a}(a,\theta):=\partial u(a,\theta) / \partial a$ is continuous in $(a,\theta)$.%
\end{assumption}

Second, the receiver's expected utility is single-peaked in his action for any posterior
belief. This is our key economic assumption.

\begin{assumption}
\label{a:qc} $U(a,\theta)$ satisfies \emph{strict aggregate quasi-concavity} in $
a$: for all posteriors $\mu\in \Delta(\Theta)$, 
\begin{equation*}
\int u(a,\theta)\df\mu=0 \implies \int u_{a}(a,\theta)\df\mu <0.
\end{equation*}
\end{assumption}

\citeauthor{Quah2012} (\citeyear*{Quah2012}) and \citeauthor{CS} (\citeyear*{CS}) characterized a weak version of aggregate quasi-concavity in terms of primitive conditions on $u$. We provide an analogous characterization of strict aggregate quasi-concavity in Appendix \ref{a:ad}.  A sufficient condition for strict aggregate quasi-concavity is that $u_{a}(a,\theta)<0$ for all $(a,\theta)$, so that $U$ is strictly concave in $a$. This stronger condition is violated in some applications we consider; however, Appendix \ref{a:ad} shows that strict aggregate quasi-concavity is equivalent to strict concavity up to a normalization.

Third, the receiver's optimal action satisfies an interiority condition.

\begin{assumption}
\label{a:int} $\min_{\theta \in \ol\Theta
}u(0,\theta )=\max_{\theta \in \ol\Theta }u(1,\theta )=0$.
\end{assumption}

Assumptions \ref{a:smooth}--\ref{a:int} imply that for any posterior $\mu$, the receiver's optimal action \linebreak $a ^{\star}(\mu ):= \argmax_{a\in [0,1]} \int U(a,\theta)\df \mu$ is unique and is characterized by the first-order condition 
\begin{equation*}
\int u(a ^{\star }(\mu ),\theta )\df\mu=0.    \label{e:rho}
\end{equation*} 
Our assumptions thus allow a ``first-order approach'' to the persuasion problem, similar to the approach of \citet{Mirrlees75} and \citet{Holmstrom79} to the classical moral hazard problem.\footnote{The conditions under which the first-order approach is valid in the persuasion problem (Assumptions \ref{a:smooth}--\ref{a:int}) are much simpler than those in the classical moral hazard problem (e.g., \citealt{rogerson}, \citealt{jewitt}). The first-order approach to persuasion is due to \citet{Kolotilin2017}.}$^,$\footnote{The substance of Assumption \ref{a:int} is that for each $\theta$, there exists $a$ such that $u(a,\theta)=0$. Note that it can never be optimal for the receiver to take any $a$ such that $u(a,\theta)$ has a constant sign for all $\theta$. We can then remove all such $a$ from $A$ and renormalize $A$ to $[0,1]$, so that Assumption \ref{a:int} holds.}

A common interpretation of the receiver's action $a \in [0,1]$ is that the receiver has a private type and makes a binary choice---say, whether to accept or reject a proposal---and $a$ is the receiver's choice of a cutoff type below which he accepts. This interpretation is especially useful for some special cases of the model, as we see next.\footnote{To spell out this interpretation, let $g(t|\theta )$ be the conditional density of the receiver's type $t\in\lbrack 0,1]$ given the state $\theta \in \lbrack 0,1]$. The sender's and receiver's utilities from rejection are normalized to zero. The sender's and receiver's utilities from acceptance are functions $\tilde{v}(t,\theta )$ and $\tilde{u}(t,\theta)$, with $\tilde{u}(t,\theta )g(t|\theta )$ satisfying Assumption \ref{a:qc}. For $a \in [0,1]$ (interpreted as the cutoff such that the receiver accepts iff $t \leq a$), we recover our model with $V(a ,\theta )=\int_{0}^{a }\tilde{v}(t,\theta )g(t|\theta )\mathrm{d}t$ and $U(a ,\theta )=\int_{0}^{a }\tilde{u}(t,\theta)g(t|\theta )\mathrm{d}t$.}

\subsection{Special Cases}

We define some leading special cases of the model, which we return to periodically to illustrate our results.

(1) The \emph{linear case} (\citealt{KG}): $u(a,\theta)=\theta-a$ and $V(a,\theta)=V(a)$. That is, $a^\star (\mu)=\mathbb{E}_{\mu}[\theta]$ and $V$ is state-independent. This is the well-studied case where the sender's indirect utility from inducing posterior $\mu$ is $V(\E_\mu[\theta])$.

(2) The \emph{linear receiver case} (\citealt{BHP}): $u(a,\theta)=\theta-a$ but $V$ is arbitrary (e.g., possibly state-dependent). Here the receiver's preferences are as in the linear case, while the sender's preferences are general. %

(2a) The \emph{separable subcase} (\citealt{RS}): $V(a,\theta)=w(\theta)G(a)$ with $w> 0$, $G> 0$, and $G'>0$, where $G'$ is the derivative of $G$. An interpretation of this subcase is that the receiver has a private type $t$ with distribution $G$ and accepts a proposal iff $\mathbb{E}_{\mu}[\theta]\geq t$, and the sender's utility when the proposal is accepted is $w(\theta)$. \citeauthor{RS} focused on the sub-subcase with the uniform distribution $G(a)=a$.\footnote{\citeasnoun{RS} assume that the state $(\omega ,\theta )$ is two-dimensional, and the sender's and receiver's marginal utilities are ${v}(a,\theta ,\omega )=\omega $ and ${u}(a,\theta )=\theta -a$. They assume
that there are finitely many states $(\theta ,\omega )$, so generically the sender's utility can be written as $v(a,\theta )=w(\theta)$. \citet{Rayo}, \citet{NP}, and \citet{OR} consider the separable subcase where $\theta$ is continuous and $(\omega ,\theta )$ is supported on the graph of $\theta\rightarrow w(\theta)$, albeit \citeauthor{Rayo} and \citeauthor{OR} restrict attention to monotone partitions. \citet{T18}, \citet{KX}, and \citet{DK} allow more general distributions of $(\omega,\theta)\in\R^2$.}

(2b) The \emph{translation-invariant subcase} (\citealt{BJ}): $V(a,\theta)=P(a-\theta)$. An interpretation of this subcase is that the receiver ``values'' the proposal at $\mathbb{E}_{\mu}[\theta]$, and the sender's utility depends on the amount by which the proposal is ``over-valued,'' $\mathbb{E}_{\mu}[\theta]-\theta$. For example, a school may care about the extent to which its students are over- or under-placed. These preferences are similar to those in \citet{GL}'s model of stress tests, discussed in Appendix \ref{s:other}.

(3) The \emph{state-independent sender case} (\citealt{FH}): $V(a,\theta)=V(a)$ with $v(a)>0$, and $u$ satisfies $u_\theta (a,\theta) :=\partial u (a,\theta) / \partial \theta >0$ for all $(a,\theta)$ but is otherwise arbitrary. Here the sender's preferences are as in the linear case, and in addition the sender prefers higher actions and the receiver's utility is strictly supermodular. %

(3a) The \emph{translation-invariant subcase}: $u(a,\theta)=T(\theta-a)$, with $T(0)=0$ and $T'>0$, where $T'$ is the derivative of $T$. An example that fits this subcase is that the sender's utility when the proposal is accepted is $1$, and the proposal corresponds to the receiver undertaking a project that can either succeed or fail, where the receiver's payoff is $1-\kappa$ when the project succeeds and $-\kappa$ when it fails (and $0$ when it is not undertaken), with $\kappa\in (0,1)$. The difficulty of the project is $1-\theta$, the receiver's ability is $1-t$, the receiver's ``bad luck'' $\varepsilon$ has distribution $J$, and the project succeeds iff $1-\theta \leq 1-t -\varepsilon$, or equivalently $\varepsilon \leq \theta-t$. This example fits the current subcase with $V$ equal to the distribution of $t$ and $T(\theta -a)=J(\theta -a)-\kappa$.

(3b) The \emph{quantile sub-subcase}: $u(a,\theta)=\1 \{\theta\geq a\}-\kappa$, with $\kappa\in (0,1)$. This subcase corresponds to the previous example with $J(\theta -a)=\1 \{\theta\geq a\}$, so the project succeeds iff the receiver's ability exceeds the project's difficulty. While $u$ is now discontinuous, we can admit this subcase as a limit of the translation-invariant case. \citet{FH} focused on the translation-invariant case where $T$ is a continuous approximation of the step function $\1 \{\theta\geq a\}-1/2$.

The mapping between our model and \citet{BHP}, \citet{BJ}, or \citet{FH} is not entirely obvious. We explain the connection following Theorem \ref{t:SDPD}, which is the closest point of contact with their results.

\section{Duality}

\label{s:duality}

We set up the sender's problem, and then derive a duality theorem that forms the basis of our analysis.

The sender's (primal) problem is to choose an outcome $\pi \in \Delta( A \times \Theta)$ to
\begin{align}
\text{maximize}\quad & \int_{A\times \Theta }V(a,\theta )\df\pi (a,\theta ) 
\tag{P} \\
\text{subject to}\quad & \int_{A\times \widetilde{\Theta }}\df\pi (a,\theta
)=\int_{\widetilde{\Theta }}\df\phi (\theta ),\quad  \text{for all measurable }\widetilde{\Theta }\subset
\Theta ,  \tag{P1} \\
& -\int_{\tilde{A}\times \Theta }u(a,\theta )\df\pi (a,\theta )=0,\quad 
\text{for all measurable } \widetilde{A}\subset A.  \tag{P2}
\end{align}

(P1) is the \textit{feasibility} constraint that the marginal of $\pi$ on $\Theta$ equals the prior, $\phi$. (P2) is the \textit{obedience} constraint that the receiver's action is $a^\star(\mu)$ at each posterior $\mu$. An outcome $\pi$ that violates (P2) is inconsistent with optimal play by the receiver, as there exists $\tilde{A}\subset A$ such that the receiver's play is suboptimal conditional on the event $\{a \in \tilde{A}\}$. Conversely, for any outcome $\pi$ that satisfies (P1) and (P2), if the sender designs a mechanism that draws $(a,\theta)$ according to $\pi$ and recommends action $a$ to the receiver, it is optimal for the receiver to obey the recommendation. We therefore say that an outcome is \emph{implementable} iff it satisfies (P1) and (P2), and \emph{optimal} iff it solves (P).

We can compare (P) to the standard optimal transport (Monge-Kantorovich) problem (e.g., \citealt{villani}). In optimal transport, two marginal distributions are given (e.g., of men and women, or workers and firms), and the problem is to find an optimal joint distribution with the given marginals. In persuasion, the marginal distribution over states is given (by the prior $\phi$), and the problem is to find an optimal joint distribution with this marginal (so (P1) holds), where for each action the conditional distribution over states satisfies obedience (so (P2) holds).

The dual problem is to find a continuous function $p:\Theta \to \mathbb{R}$ and a bounded, measurable function $q:A \to \mathbb{R}$ to
\begin{align}
\text{minimize}\quad & \int_{\Theta }p(\theta )\df\phi (\theta )  \tag{D} \\
\text{subject to}\quad & p(\theta )-q(a)u(a,\theta )\geq V(a,\theta ),\quad
\text{for all }(a,\theta )\in A\times \Theta .  \tag{D1}
\end{align}

We say that $(p,q)$ is \emph{feasible} iff it satisfies (D1), and \emph{optimal} iff it solves (D).
The interpretation of the dual problem is that $p(\theta)$ is the shadow price of state $\theta$; $q(a)$ is the value of relaxing the obedience constraint at action $a$;
and the dual constraint (D1) says that $p(\theta)$ is no less than the sender's value from assigning state $\theta$ to any action $a$, where this value is the sum of the sender's utility, $V(a,\theta)$, and the product of $q(a)$ and the amount by which the obedience constraint at $a$ is relaxed when state $\theta$ is assigned to action $a$, $u(a,\theta)$.

A first result is that solutions to (P) and (D) exist, and there is no duality gap. Let $C(\Theta)$ denote the set of continuous functions on $\Theta$, and let $B(A)$ denote the set of bounded, measurable functions on $A$. We say that a price function $p \in C(\Theta)$ \emph{solves} (D) iff there exists $q \in B(A)$ such that $(p,q)$ is a solution to (D).

\begin{lemma} \label{l:dual}
Let Assumptions \ref{a:smooth}--\ref{a:int} hold.
\begin{enumerate}
\item There exists $\pi \in \Delta (A \times \Theta)$ that solves (P).
\item There exists $p \in C(\Theta)$ that solves (D).
\item The values of (P) and (D) are the same: for any solutions $\pi$ of (P) and $p$ of (D), we have 
\[\int_{A\times \Theta }V(a,\theta )\df\pi (a,\theta )=\int_{\Theta }p(\theta )\df\phi (\theta ).\]
\end{enumerate}
\end{lemma}

Lemma \ref{l:dual} is similar to Lemmas 1 and 2 of \citet{Kolotilin2017}. We provide a more detailed alternative proof that applies under slightly weaker assumptions.\footnote{The proof in \citet{Kolotilin2017} uses the Banach-Alaoglu theorem, as in linear programming references such as \citet{anderson1987}. Our proof uses the Arzela-Ascoli theorem, as in optimal transport references such as \citet{villani} and \citet{santambrogio}. A key step in the proof (Lemma \ref{l:D=D'})---which was left somewhat implicit in \citeauthor{Kolotilin2017}---is showing that $q$ may be assumed bounded in (D). Our proof also remains valid when $\Theta$ is an arbitrary compact metric space.  \citet{DK} prove a related duality result which allows the receiver's action to be multi-dimensional but requires Lipschitz continuity of the sender's indirect utility.}

\section{Contact Set} \label{s:contact}

In this section, we define a compact set $\Gamma \subset A \times \Theta$ with the properties that an implementable outcome $\pi$ is optimal iff $\supp(\pi)\subset \Gamma$, and the first-order condition of the dual problem (equation \eqref{e:FOC}) holds at any pair $(a,\theta)$ in a full-measure subset $\Gamma^\star \subset \Gamma$. Following the optimal transport literature (e.g., Chapter 3 in \citealt{abs}), we refer to this set $\Gamma$ as the \emph{contact set}. Readers who wish to skip the technical details can just familiarize themselves with Theorem \ref{t:contact} and equation \eqref{e:FOC} before moving on to the next section. In particular, the specification of the multipliers $q(a)$ and the distinction between the sets $\Gamma$ and $\Gamma^\star$ can be elided on a first reading.

We henceforth assume that the receiver prefers higher actions at higher states.

\begin{assumption}
\label{a:sc}
$u(a,\theta)$ satisfies \emph{strict single-crossing} in $\theta$: for all $a$ and $\theta<\theta'$, \[u(a,\theta)\geq 0 \implies u(a,\theta')> 0.\]
\end{assumption}

Together with Assumptions \ref{a:smooth}--\ref{a:int}, Assumption \ref{a:sc} ensures that for each action $a$ there is a unique state $\theta^\star(a)$ such that $u(a,\theta^\star(a))=0$, and that $\theta^\star(a)$ is a strictly increasing, continuous function from $A$ onto $\ol \Theta$.%

Let $p$ be the optimal price function (which we will see is unique under Assumptions \ref{a:smooth}--\ref{a:sc}), and let $I$ be any sufficiently large compact interval (e.g., as defined in Lemma \ref{l:D=D'}). 
Let 
\[Q(a):=\left\{r \in I : \; p(\theta) \geq V(a,\theta) + ru(a,\theta)\text{ for all }\theta \in \Theta \right\}, \quad \text{for all } a \in A.\] 

This is the set of possible values for $q(a)\in I$ that satisfy (D1) for all $\theta$, given the optimal price function $p$. Note that for any measurable selection $q$ from $Q$, the pair $(p,q)$ is a solution to (D).

By part (3) of Lemma \ref{l:dual}, together with (P1) and (P2), any optimal $\pi$ and $(p,q)$ satisfy
\[
\int_{A\times \Theta }(p(\theta)-V(a,\theta) -q(a)u(a,\theta))\df\pi (a,\theta ) =0.
\]
By (D1), the integrand is non-negative, and hence any optimal $\pi$ is concentrated on the set $\Gamma$ of points $(a,\theta)$ that satisfy (D1) with equality. We call any such set $\Gamma$ \emph{a contact set}. Note that $\Gamma$ depends on the selection $q$ from $Q$.

Our first main result (Theorem \ref{t:contact}) shows that $q$ given by
\[
q(a):=
\begin{cases}
-\frac{v(a,\theta^\star(a))}{u_a(a,\theta^\star(a))}, &\text{$\theta^\star (a)\in \Theta$ and $p(\theta^\star(a))=V(a,\theta^\star(a))$,}\\
\frac{\min Q(a) +\max  Q(a)}{2}, &\text{otherwise},
\end{cases}
\]
is a measurable selection from $Q$, and the associated contact set $\Gamma$ given by
\[
\Gamma: = \{(a,\theta)\in A\times \Theta:p(\theta)=V(a,\theta)+q(a)u(a,\theta)\}
\]
has the desired properties. We call this set $\Gamma$ \emph{the} contact set, to distinguish it from contact sets that result from different choices of $q$.  We explain the role of our chosen multipliers $q$ after stating our result.

Some notation is in order. First, for each $a$, the $a$-section of $\Gamma$ is defined as
\[
\Gamma_a:=\{\theta\in \Theta:(a,\theta)\in \Gamma\}.
\]
Intuitively, $\Gamma_a$ is the set of states that it is optimal to pool together to induce action $a$. Next, the projection of $\Gamma$ on $A$ is defined as
\[
A_{\Gamma}:=\{a\in A:(a,\theta)\in \Gamma\text{ for some }\theta\in \Theta\}.
\]
Intuitively, $A_{\Gamma}$ is the set of actions that it is ever optimal to induce. Finally, the set $\Gamma^\star\subset \Gamma$ is defined by letting its $a$-section be given by
\[
\Gamma^\star_a:=
\begin{cases}
\{\theta^\star(a)\}, & a \in A_{\Gamma} \text{ and } \theta^\star (a)\in \{\min \Gamma_a,\max \Gamma_a\},\\
\Gamma_a, &\text{otherwise},
\end{cases}
\; \; \; \; \text{for all } a\in A.
\]
As we will explain, $\Gamma^\star$ is a subset of $\Gamma$ that removes ``redundant'' states from each $a$-section.
\begin{theorem}\label{t:contact} Let Assumptions \ref{a:smooth}--\ref{a:sc} hold.
\begin{enumerate}
	\item The set $\Gamma$ is compact and satisfies $\min \Gamma_a\leq \theta^\star (a)\leq \max \Gamma_a$ for all $a\in A_\Gamma$. Moreover, $(p,q)$ solves (D). Consequently, an implementable outcome $\pi$ solves (P) iff $\supp(\pi)\subset \Gamma$.
	\item The set $\Gamma^\star$ is a Borel subset of $\Gamma$, and
\begin{equation}\label{e:FOC}
v(a,\theta)+ q(a)u_a(a,\theta)+q'(a)u(a,\theta)=0, \quad \text{for all $(a,\theta)\in \Gamma^\star$},
\end{equation}
with the convention that $q'(a)\cdot 0=0$, even if $q$ is not differentiable at $a$. 	
Moreover, an implementable outcome $\pi$ solves (P) iff there exists a conditional probability $\pi_a$ of $\pi$ given $a$ such that $\supp (\pi_a)\subset \Gamma_a^\star$ and $\int_\Theta u(a,\theta) \df \pi_a(\theta)=0$ for all $a\in \supp (\alpha_\pi)$, where $\alpha_\pi$ denotes the marginal distribution of $\pi$ on $A$. 
\end{enumerate}
\end{theorem}

Equation \eqref{e:FOC} is the first-order condition of the dual problem: by (D1), the sender chooses an action $a$ to induce at state $\theta$ so as to maximize $V(a,\theta)+q(a)u(a,\theta)$, and taking the FOC with respect to $a$ yields \eqref{e:FOC}. Thus, Theorem \ref{t:contact} says that there is a compact contact set $\Gamma$ such that an implementable outcome is optimal iff it is supported on $\Gamma$; and there is a measure-1 subset $\Gamma^\star \subset \Gamma$ such that the sender's FOC holds on $\Gamma^\star$. Theorem \ref{t:contact} is our key tool for characterizing optimal outcomes: by showing that points $(a,\theta)$ violate \eqref{e:FOC}, we can exclude them from $\Gamma ^\star$, and hence from the support of any optimal outcome. 

Taking the expectation of \eqref{e:FOC} with respect to an optimal conditional probability $\pi_a$ yields a useful formula for $q(a)$:
\begin{equation} \label{eqn:q} 
q(a)=-\frac{\E_{\pi_{a}}[v(a,\theta)] }{\E_{\pi_{a}}[u_a(a,\theta)]}, \quad \text{for all } a \in \supp (\alpha_{\pi}).
\end{equation}
This says that $q(a)$ equals the product of the sender's expected marginal utility at $a$ and the rate at which $a$ increases as the obedience constraint is relaxed, where the latter term equals $-1/\E_{\pi_{a}}[u_a(a,\theta)]$ by the implicit function theorem applied to the obedience constraint. Note that we defined $q$ so that \eqref{eqn:q} holds for $a$ where $\pi_{a}=\delta_{\theta^\star (a)}$ (i.e., for actions induced at disclosed states); here we see that this equation also holds for $a$ where $\pi_{a}$ is non-degenerate (i.e., for actions induced at pooled states).

The technical aspects of Theorem \ref{t:contact}---the particular choice of $q$ and the distinction between $\Gamma$ and $\Gamma^\star$---are specified so that $\Gamma$ and $\Gamma^\star$ have the desired properties of compactness and satisfaction of \eqref{e:FOC}, respectively. Intuitively, by selecting $q(a)$ from the interior of $Q(a)$ (when $p(\theta^\star(a))>V(a,\theta^\star(a))$ and $Q(a)$ is multivalued), we ensure that $A_\Gamma$ does not contain any actions $a$ that are ``redundant,'' in the sense that $\theta^\star (a)\notin [\min \Gamma_a,\max \Gamma_a]$---for such actions, $\int _\Theta u(a,\theta)\df \pi_a (\theta)\neq 0$ for all $\pi_a\in \Delta (\Gamma_a)$, so these actions are not induced by any optimal outcome. In turn, $\Gamma^\star$ is obtained from $\Gamma$ by removing redundant states from each $a$-section---if $\theta^\star (a)\in \{\min \Gamma_a,\max \Gamma_a\}$, then $\pi_a(\theta^\star(a))=1$ for any $\pi_a\in \Delta(\Gamma_a)$ such that $\int_\Theta u(a,\theta)\df \pi_a(\theta)=0$, so any states $\theta \neq \theta^\star(a)$ can be removed from $\Gamma_a$.\footnote{Thus, $\Gamma\setminus\Gamma^\star$ is the \emph{polar} subset of $\Gamma$ with respect to (P2), in the sense that a set $\Gamma^0\subset \Gamma$ satisfies $\pi(\Gamma^0)=0$ for all $\pi\in \Delta (\Gamma)$ satisfying (P2) iff $\Gamma^0\subset \Gamma\setminus \Gamma^\star$.} We provide examples illustrating these and other technical points in Appendix \ref{s:examples}.

\begin{remark}
Under Assumptions \ref{a:smooth}--\ref{a:sc}, there is a unique solution $p$ to (D). We give a proof of this fact following the proof of Theorem \ref{t:contact}. As shown by the examples in Appendix \ref{s:examples}, while the price function $p$ is unique, there can be multiple functions $q$ such that $(p,q)$ is a solution to (D). %
\end{remark}

Lemma \ref{l:dual} and Theorem \ref{t:contact} can be compared to results in the literature on \emph{martingale optimal transport (MOT)}. The MOT problem is to find an optimal joint distribution of two variables (say, $a$ and $\theta$) with given marginals, subject to the martingale constraint $\mathbb{E}_{\pi_a}[\theta]=a$ for all $a$. This problem coincides with our linear receiver case, but with an exogenously fixed distribution of the receiver's action. Motivated by problems in mathematical finance, \citet{BHP} (see also \citealt{BNT}) introduce MOT and prove that the primal and dual problems have the same value; however, they also show that their dual problem may not have a solution, unlike in our model with endogenous actions. Results in MOT also do not establish compactness of the contact set, which holds in our model as well as in standard optimal transport. Thus, MOT is related to our linear receiver case, but the endogenous action distribution apparently makes our model more tractable.

\section{Pairwise Disclosure and the Twist Condition}

\label{s:pairwise}

The contact set $\Gamma$ introduced above describes the set of pairs of actions $a$ and states $\theta$ that it is optimal for the sender to match together, in the sense that $(a,\theta)$ is contained in the support of an optimal outcome. At the same time, the $a$-section $\Gamma_{a}$ describes the set of states $\theta$ that it is optimal for the sender to pool together to induce action $a$. We say that the contact set is \emph{pairwise} if $|\Gamma^{\star}_{a}|\leq 2$ for all $a$. When the contact set is pairwise, it is strictly suboptimal for the sender to ever pool more than two states. In this section, we show that there always exist optimal signals that never pool more than two states, and we give conditions under which the contact set is pairwise, so that every optimal signal has this property.

A \textit{signal} $\tau\in \Delta (\Delta (\Theta))$ is a distribution over posterior beliefs $\mu \in \Delta (\Theta)$ such that the average posterior equals the prior: $\int\mu \mathrm{d} \tau=\phi$ (\citealt{AM}, \citealt{KG}). Uniqueness of the receiver's optimal action implies that any signal $\tau$ induces a unique outcome $\pi _{\tau }$ through the map $\mu \mapsto a^{\star}\left( \mu \right)$.\footnote{Conversely, any implementable outcome $\pi$ is induced by a signal $\tau ^{\pi }$ through the map $a\mapsto \pi_{a}$.}  
A signal $\tau $ is \emph{pairwise} if it induces posterior beliefs with at most binary support: $|\supp (\mu) |\leq 2$ for each $\mu \in \supp(\tau )$.%

For example, with a uniform prior $\phi$, for any cutoff $\hat\theta \in [0,1]$ the signal that reveals states below the cutoff and pools each pair of states $\theta$ and $1+\hat\theta-\theta$ for $\theta \in [\hat\theta,(1+\hat\theta)/2]$ to induce posterior $\mu =\delta _{\theta}/2+\delta _{1+\hat\theta-\theta}/2$ is pairwise. The special case where $\hat\theta =1$ is full-disclosure, which is also pairwise. In contrast, no-disclosure, where $\tau(\phi)=1$, is not pairwise.

If the receiver's utility is not quasi-concave, pairwise signals may be suboptimal. For example, suppose the sender rules three castles, one of which is undefended. The state $\theta$---the identity of the undefended castle---is uniformly distributed. Suppose the receiver can attack any two castles, and payoffs are $(-1,+1)$ for the sender and receiver, respectively, if the receiver attacks the undefended castle, and are $(+1,-1)$ otherwise. Then any pairwise signal narrows the set of possibly undefended castles to at most two, so the receiver always wins. But if the sender discloses nothing, the receiver wins only with probability $2/3$.\footnote{Another example of a persuasion problem where pairwise signals are suboptimal is the price-discrimination problem of \citeasnoun{BBM}. Note that the receiver's utility is not quasi-concave in the three-castles or price-discrimination examples.}

Our second main result is that pairwise signals are without loss under Assumptions \ref{a:smooth}--\ref{a:int}.\footnote{Our proof of this result does not require Assumption \ref{a:sc}, and also remains valid when $\Theta$ is an arbitrary compact metric space.} Moreover, equation \eqref{e:FOC} implies that if it is optimal to induce the same action $a$ at three states $\theta_1$, $\theta_2$, and $\theta_3$, then the vector $(v(a,\theta_1),v(a,\theta_2),v(a,\theta_3))$ must be a linear combination of the vectors $(u(a,\theta_1),u(a,\theta_2),u(a,\theta_3))$ and $(u_a(a,\theta_1),u_a(a,\theta_2),u_a(a,\theta_3))$. This observation gives a condition---which we call the \emph{twist condition}---under which pooling more than two states is suboptimal, so that every optimal signal is pairwise.\footnote{We use the notation $|\cdot |$ for both the cardinality of a set and the determinant of a matrix.}

\textbf{Twist Condition} \emph{For all $a$ and $\theta_1<\theta_2<\theta_3$ such that $\theta_1<\theta^\star(a)<\theta_3$, we have}
\begin{equation}\label{e:sing}
|S|:=
\begin{vmatrix}
v(a,\theta_1) & v(a,\theta_2) & v(a,\theta_3)\\
u(a,\theta_1) &u(a,\theta_2) &u(a,\theta_3)\\
u_a(a,\theta_1) &u_a(a,\theta_2) &u_a(a,\theta_3)
\end{vmatrix}
\neq 0.%
\end{equation}

We will apply this condition extensively in Section \ref{s:assortative}.

\begin{theorem}\label{p:S}
\label{t:pairwise} Let Assumptions \ref{a:smooth}--\ref{a:sc} hold.
\begin{enumerate}
\item For any signal $\tau$, there exists a pairwise signal $\hat{\tau}$ such that $\pi_{\hat{\tau}}=\pi_{\tau}$. 
\item If the twist condition holds, then $|\Gamma_a^\star|\leq 2$ for all $a$, and hence every optimal signal is pairwise.
\end{enumerate}
\end{theorem}

The intuition for part (1) is that for any posterior, there exists a hyperplane passing through it such that all posteriors on the hyperplane induce the same action, and the extreme points of the hyperplane in the simplex have at most binary support. Thus, any posterior that puts weight on more than two states can be split into posteriors with at most binary support without affecting the induced distribution on $A\times \Theta$. Figure 1 illustrates this argument for a posterior with weight on three states.

\begin{figure}[t]
\centering
\begin{tikzpicture}[scale=0.7]
\node[circle,fill=black,inner sep=0pt,minimum size=3pt,label=below:{$\theta_1$}] (a) at (0,0) {};
\node[circle,fill=black,inner sep=0pt,minimum size=3pt,label=above:{$\theta_2$}] (b) at (4,6.928) {};
\node[circle,fill=black,inner sep=0pt,minimum size=3pt,label=below:{$\theta_3$}] (c) at (8,0) {};
\draw (a) -- (b);
\draw (b) -- (c);
\draw (c) -- (a);
\node[circle,fill=black,inner sep=0pt,minimum size=3pt,label=left:{$\mu'$}] (d) at (2.664,4.616) {};
\node[circle,fill=black,inner sep=0pt,minimum size=3pt,label=below:{$\mu''$}] (e) at (3.2,0) {};
\draw (d) -- (e);
\node[circle,fill=black,inner sep=0pt,minimum size=3pt] (f) at (2.936,2.312) {};
\node (g) at (4,3) {$\leftarrow a^{\star}(\mu)$};
\node(h) at (2.6,2) {$\mu$};
\end{tikzpicture}
\caption{Pairwise Signals are Without Loss}
\caption*{\emph{Notes:} The optimal action at any posterior on the line between $\mu'$ and $\mu''$ equals $a^{\star}(\mu)$, so splitting $\mu$ into $\mu'$ and $\mu''$ eliminates a non-binary-support posterior without changing the outcome.}
\label{f:simplex}
\end{figure}

To get a sense of the proof of part (1), note that, for a given posterior $\mu$, another posterior $\mu'$ induces the same action as $\mu$ iff the action $a^\star(\mu)$ satisfies the first-order condition $\int u(a^\star(\mu),\theta)\df \mu'=0$. Since the first-order condition is a moment condition, the set of posteriors that induce action $a^\star(\mu)$ is the set of probability distributions that satisfy one moment condition. By Richter-Rogosinsky's theorem, the extreme points of this set have at most binary support. Hence, by Choquet's theorem, $\mu$ can be written as an expectation, with respect to some measure $\lambda_{\mu}\in \Delta(\Delta(\Theta))$, of distributions with at most binary support that all induce action $a^\star(\mu)$. Finally, by the measurable selection theorem, the mapping from $\mu$ to $\lambda_{\mu}$ can be taken to be measurable, and can thus be used to define a pairwise signal that induces the same distribution on $A\times \Theta$ as any given signal $\tau$.\footnote{This argument indicates how part (1) generalizes when actions are multi-dimensional: if $A$ is a compact, convex subset of $\mathbb{R}^N$ and the receiver's utility is strictly concave, then the receiver's optimal action is characterized by $N$ first-order conditions, so it is without loss to consider signals that induce posteriors which are supported on at most $N+1$ states.}

Part (2) follows easily from Theorem \ref{t:contact} (in particular equation \eqref{e:FOC}), but it also has a simple intuition based on pairwise signals. Consider a posterior distribution $\mu $ with $\supp(\mu )=\left\{ \theta_1 ,\theta_2,\theta_3 \right\} $. By part (1), we can split $\mu $ into posterior distributions $\mu_1$ and $\mu _2$ with at most binary support that both induce action $a ^{\ast }(\mu)$. For example, suppose that $\supp (\mu_1)=\{\theta_1 ,\theta_2\}$ and $\supp(\mu_2)=\{\theta_1 ,\theta_3\}$. Consider a perturbation that moves probability mass $\mathrm{d}p$ on $\theta_1 $ from $ \mu _1$ to $\mu _2$. This perturbation induces non-zero marginal changes in the receiver's action at $\mu _1$ and $ \mu _2$. Under the twist condition, these changes have a non-zero marginal effect on the sender's expected utility by the implicit function theorem. Therefore, either this perturbation or the reverse perturbation, where $\mathrm{d}p$ is replaced with $-\mathrm{d}p$, is strictly profitable.

Prior results by \citet{RS}, \citet{AC}, and \citet{ZZ} also give conditions under which all optimal signals are pairwise. Theorem \ref{p:S} easily implies these earlier results.\footnote{Proposition 4 in \citet{AC} states that if $u(a,\theta)=\theta - a$ and there do not exist $\zeta\leq 0$ and $\iota\in \R $ such that $v(a,\theta_i)=\zeta \theta_i+\iota$ for $i=1,2,3$, then it is not optimal to induce action $a$ at states $\theta_1$, $\theta_2$, and $\theta_3$. This result is too strong as stated, and it is not correct unless $\zeta$ is also allowed to be positive. Theorem \ref{p:S} implies this corrected version of \citeauthor{AC}'s result.} Note that the twist condition always fails in the linear case (i.e., $|S|=0$). Hence, in the linear case, Theorem \ref{p:S} never rules out pooling multiple states, and indeed pooling multiple states is often optimal (e.g., \citealt{KMZL}).\footnote{Of course, Theorem \ref{t:pairwise} shows that even when pooling multiple states is optimal, there also exists an optimal pairwise signal, where the ``multi-state pool'' is split into pairs. Conversely, if multiple posteriors all induce the same action, they can be pooled without affecting the outcome.} 

An immediate corollary of Theorem \ref{t:contact} is that no disclosure is generically suboptimal when there are at least three states, because for a fixed action $a$ a generic vector $(v(a,\theta))_{\theta \in \Theta}$ with $|\Theta|\geq 3$ coordinates cannot be expressed as a linear combination of two vectors $(u(a,\theta))_{\theta \in \Theta}$ and $(u_a(a,\theta))_{\theta \in \Theta}$, as is required by \eqref{e:FOC}.

\begin{corollary}\label{c:nodisc}
Let Assumptions \ref{a:smooth}--\ref{a:sc} hold.
For any $\phi$ with $|\supp(\phi)|\geq 3$ and any $u$, no disclosure is suboptimal for generic $v.$
\end{corollary}

Given \citeauthor{KG}'s concavification result, Corollary \ref{c:nodisc} implies that, generically, the sender's indirect utility is not concave in the posterior when there are more than two states. Also, observe that Corollary \ref{c:nodisc} allows the case where $u$ and $v$ always have the opposite sign, so the sender's and receiver's ordinal preferences are diametrically opposed. Hence, even in this case no-disclosure is generically suboptimal.

While Theorem \ref{t:contact} shows that the contact set always characterizes optimal outcomes---in that an implementable outcome $\pi$ is optimal iff $\supp(\pi)\in \Gamma$---when the contact set is pairwise it also directly determines the optimal conditional probability $\pi_a = \rho_a\delta_{t_1(a)}+(1-\rho_a)\delta_{t_2(a)}$, where $\Gamma^\star_a = \{t_1(a),t_2(a)\}$ with $t_1(a)\leq t_2(a)$, and $\rho_a\in [0,1)$ satisfies the obedience condition $\rho_a u(a,t_1(a))+(1-\rho_a)u(a,t_2(a))=0$. Thus, when $\Gamma^\star$ is pairwise all optimal outcomes have the same pairwise conditional probability $\pi_a$, and may differ only in the marginal distribution of actions $\alpha_\pi$.

\section{Single-Dipped and Single-Peaked Disclosure} \label{s:assortative}

The next two sections investigate optimal disclosure patterns: which actions $a$ should the sender induce at which states $\theta$?
In this section, we ask when it is optimal for the sender to induce higher actions at more or less extreme states: that is, when optimal outcomes are ``single-dipped'' or ``single-peaked.''\footnote{Mathematically, positive/negative assortativity correspond to monotonicity in the FOSD order, while single-dippedness/-peakedness correspond to monotonicity in a variability order that depends on $u$; when $u(a,\theta)=\theta-a$, this variability order is the usual convex order. %
}

Formally, a triple $(a_1,\theta_1)$, $(a_2,\theta_2)$, $(a_1,\theta_3)$ is \emph{single-dipped} (\emph{-peaked}) if $a_1\geq(\leq)a_2$ and $\theta_1<\theta_2<\theta_3$; similarly, such a triple is \emph{strictly single-dipped} (\emph{-peaked}) if $a_1 >(<)a_2$. A set $\Gamma^\dagger \subset A\times \Theta$ is \emph{single-dipped} (\emph{-peaked}) if it does not contain a strictly single-peaked (-dipped) triple of points; similarly, such a set is \emph{strictly single-dipped} (\emph{-peaked}) if it does not contain a single-peaked (-dipped) triple. Finally, an outcome $\pi$ is (\emph{strictly}) \emph{single-dipped} if it is concentrated on a (strictly) single-dipped set,\footnote{That is, there exists a Borel (strictly) single-dipped set $\Gamma^\dagger$ such that $\pi(\Gamma^\dagger)=1$.} and similarly for single-peakedness. In particular, by Theorem \ref{t:contact}, if $\Gamma$ or $\Gamma^\star$ is single-dipped/-peaked, then so is every optimal outcome. Most of our results for single-dippedness/peakedness are symmetric, in which case we provide proofs only for the single-dipped case.

\subsection{Variational Theorem}

Characterizing when optimal signals are single-dipped/-peaked involves some additional conditions on the sender's and receiver's preferences. The simplest of these is that the sender prefers higher actions.
\begin{assumption} 
\label{a:v>0}
$v(a,\theta)>0$ for all $(a,\theta)$.
\end{assumption}

We now introduce a matrix $R$, which is a non-local analog of the matrix $S$ from the twist condition. For any $a_1,a_2$ and $\theta_1<\theta_2<\theta_3$, we define 
\begin{equation*}\label{e:A}
R:=
\begin{pmatrix}
V(a_2,\theta_1)-V(a_1,\theta_1) &-(V(a_2,\theta_2)-V(a_1,\theta_2)) &V(a_2,\theta_3)-V(a_1,\theta_3)\\
-u(a_1,\theta_1) &u(a_1,\theta_2) &-u(a_1,\theta_3) \\
u(a_2,\theta_1) &-u(a_2,\theta_2) &u(a_2,\theta_3)
\end{pmatrix}.
\end{equation*}

The next result is our main tool for determining when optimal outcomes are single-dipped/-peaked.

\begin{theorem}\label{l:ssdd} 
Let Assumptions \ref{a:smooth}--\ref{a:v>0} hold.
Suppose that for all $\theta_1<\theta_2<\theta_3$ and all $a_2>(<)a_1$ such that $\theta_1\leq \theta^\star(a_1)\leq \theta_3$, there exists a vector $y\geq 0$ such that $Ry\geq 0$ and $Ry\neq 0$.
Then $\Gamma$ is single-dipped (-peaked), and hence so is any optimal outcome.
\end{theorem}

\begin{figure}[t]
\centering
\begin{tikzpicture}
\node (a1) at (0,0) {};
\node[label=right:{$a_1$}] (a2) at (10,0) {};
\node[circle,fill=black,inner sep=0pt,minimum size=3pt,label=below:{$\theta_1$}] (a3) at (1,0) {};
\node[circle,fill=black,inner sep=0pt,minimum size=3pt,label=below:{$\theta_2$}] (a4) at (5,0) {};
\node[circle,fill=black,inner sep=0pt,minimum size=3pt,label=below:{$\theta_3$}] (a5) at (9,0) {};
\draw (a1) -- (a2);
\node (b1) at (0,2) {};
\node[label=right:{$a_2$}] (b2) at (10,2) {};
\draw (b1) -- (b2);
\node[circle,fill=black,inner sep=0pt,minimum size=3pt] (b3) at (1,2) {};
\node[circle,fill=black,inner sep=0pt,minimum size=3pt] (b4) at (5,2) {};
\node[circle,fill=black,inner sep=0pt,minimum size=3pt] (b5) at (9,2) {};
\draw[->] (a3) -- (b3);
\node (c3) at (1.3,1) {$y_1$};
\draw[->] (b4) -- (a4);
\node (c3) at (5.3,1) {$y_2$};
\draw[->] (a5) -- (b5);
\node (c3) at (9.3,1) {$y_3$};
\end{tikzpicture}
\caption{A Profitable Perturbation of a Non-Single-Dipped Outcome}
\caption*{\emph{Notes:} The figure shows a perturbation of an outcome that shifts
 weights $y_1$ and $y_3$ on $\theta_1$ and $\theta_3$ from $a_1 $ to $a_2$, and shifts weight $y_2$ on $\theta_2$ from $a_2$ to $a_1$. This perturbation is profitable if it increases the receiver's expected marginal utility at $a_1$ and $a_2$ and also increases the sender's expected utility for fixed $a_1$ and $a_2$.}
\label{f:sdd}
\end{figure}

The economic idea behind Theorem \ref{l:ssdd} is very simple. The condition for single-dippedness says that an outcome that assigns positive probability to a strictly single-peaked triple $(a_1,\theta_1)$, $(a_2,\theta_2)$, $(a_1,\theta_3)$ can be improved by re-allocating mass $y_1$ on $\theta_1$ and mass $y_3$ on $\theta_3$ from $a_1$ to $a_2$, while re-allocating mass $y_2$ on $\theta_2$ from $a_2$ to $a_1$. See Figure \ref{f:sdd} for an illustration. Indeed, this re-allocation is profitable for the sender, because the sender's expected utility increases when $a_1$ and $a_2$ are held fixed (i.e., the first coordinate of $Ry$ is non-negative); the receiver's marginal utility conditional on being recommended $a_1$ increases (i.e., the second coordinate of $Ry$ is non-negative), which increases the receiver's action, and hence increases the sender's expected utility by Assumption \ref{a:v>0}; and the receiver's marginal utility conditional on being recommended $a_2$ also increases (i.e., the third coordinate of $Ry$ is non-negative), which again increases the sender's expected utility. Moreover, at least one of these improvements is strict (i.e., $Ry \neq 0$). The same logic applies for an outcome whose support contains a strictly single-peaked triple (even if this triple occurs with $0$ probability), except now mass must be re-allocated from small intervals around $\theta_1$, $\theta_2$, and $\theta_3$.

We also make use of the following stability result, which implies that if the conditions of Theorem \ref{l:ssdd} hold only weakly but can be approximated by strict conditions, then there exists an optimal single-dipped/-peaked outcome $\pi$ (however, in this case there could also be other optimal outcomes that are not single-dipped/-peaked). 
For example, this result implies that in the linear case there is an optimal single-dipped outcome as well as an optimal single-peaked outcome.\footnote{Section 4.3 in \citet{KMS} and Theorem 1 in \citet{ABSY20} establish a result in the linear case that is somewhat related to this observation. They show that there exists an optimal signal that partitions the state space into singletons and intervals, with each singleton state being disclosed and each interval of states being pooled into one or two distinct posterior means. This result easily implies that there exist both an optimal single-dipped outcome and an optimal single-peaked outcome (see, e.g., Corollary 2 in \citealt{ABSY20}).}$^,$\footnote{The proof of Lemma \ref{l:stab} is complicated by the fact that the Hausdorff limit of single-dipped sets is not necessarily single-dipped. This point is illustrated in Example \ref{ex:stab} in Appendix \ref{s:examples}, which also shows that the lemma's conclusion cannot be strengthened to the claim that there exists an optimal outcome that is \emph{supported} on a single-dipped/-peaked set (rather than merely being concentrated on such a set).}

\begin{lemma}\label{l:stab}
Let Assumptions \ref{a:smooth}--\ref{a:sc} hold.
Suppose that $v^n$ is a sequence of continuous functions converging uniformly to $v$, and suppose that the corresponding contact sets $\Gamma^n$ are single-dipped (-peaked). Then there exists a single-dipped (-peaked) optimal outcome.
\end{lemma}

\subsection{Sufficient Conditions}
We now impose an additional assumption requiring some extra smoothness (cf. Assumption \ref{a:smooth}) and, more substantively, strengthening strict single-crossing of $u$ in $\theta$ (Assumption \ref{a:sc}) to strict monotonicity.
\begin{assumption}\label{a:u_s}
$v(a,\theta)$, $u(a,\theta)$, and $u_a(a,\theta)$ have partial derivatives in $\theta$, denoted by $v_{\theta}(a,\theta)$, $u_{\theta}(a,\theta)$, and $u_{a\theta}(a,\theta)$. In addition, $u_{\theta}(a,\theta)>0$ for all $(a,\theta)$.
\end{assumption}

The following result gives general sufficient conditions for optimality of single-dipped/-peaked disclosure. As we will see, these conditions cover several prior models, as well as some new applications.

\begin{theorem}\label{t:SDPD}
Let Assumptions \ref{a:smooth}--\ref{a:u_s} hold. If
$u_{a\theta}(a,\theta)/u_{\theta}(a,\theta)$ and  $v_\theta (a_2,\theta)/u_\theta(a_1,\theta)$  are increasing (decreasing) in $\theta$ for all $a$ and $a_2\geq (\leq) a_1$,
then there exists an optimal single-dipped (-peaked) outcome.

If in addition either $u_{a\theta}(a,\theta)/u_{\theta}(a,\theta)$ or $v_\theta (a_2,\theta)/u_\theta(a_1,\theta)$  is strictly increasing (decreasing) in $\theta$ for all $a$ and $a_2\geq (\leq) a_1$, then $\Gamma$ is single-dipped (-peaked) and $\Gamma^\star$ is strictly single-dipped (-peaked), and hence every optimal outcome is strictly single-dipped (-peaked).
\end{theorem}

The proof of Theorem \ref{t:SDPD} verifies the conditions in Theorems \ref{t:pairwise} and \ref{l:ssdd} and Lemma \ref{l:stab}, with a perturbation that holds fixed actions $a_1$ and $a_2$ while increasing the sender's expected utility in the single-dipped case, and a perturbation that holds fixed a higher action $a_1$ and the sender's expected utility (for fixed $a_1 , a_2$) while increasing a lower action $a_2$ in the single-peaked case.

The intuition for Theorem \ref{t:SDPD} is relatively straightforward in the linear receiver and state-independent sender cases. In the linear receiver case, $u_{a\theta}(a,\theta)/u_{\theta}(a,\theta)=0$ and $v_\theta (a_2,\theta)/u_\theta(a_1,\theta)=v_\theta (a_2,\theta)$, so our sufficient conditions for single-dipped disclosure to be optimal are satisfied iff $v$ is convex in $\theta$.\footnote{In the separable and translation-invariant subcases, convexity of $v$ simplifies to convexity of $w$ and $P'$, respectively.} To see why, note that for any strictly single-peaked triple $(a_1,\theta_1)$, $(a_2,\theta_2)$, $(a_1,\theta_3)$, the perturbation that moves mass on $\theta_1$ and $\theta_3$ from $a_1$ to $a_2$ and moves mass on $\theta_2$ in the opposite direction, so as to hold fixed the receiver's marginal utility conditional on being recommended either action, has the effect of also holding fixed the probability of each recommendation, while spreading out the state conditional on action $a_2$ and concentrating the state conditional on action $a_1$. This perturbation is profitable when the difference $V(a_2,\theta)-V(a_1,\theta)$ is convex in $\theta$, which holds whenever $v$ is convex in $\theta$.\footnote{The careful reader may notice that this argument did not invoke Assumption \ref{a:v>0}, because the receiver's actions $a_1$ and $a_2$ were held fixed in the relevant perturbation. Indeed, in the linear receiver case, Theorem \ref{t:SDPD} holds even without Assumption \ref{a:v>0}, as shown in Appendix \ref{app:SR}.}

In the state-independent sender case, $v_\theta (a_2,\theta)/u_\theta(a_1,\theta)=0$, so our sufficient conditions for single-dipped disclosure to be optimal are satisfied iff $u_\theta$ is log-supermodular in $(a,\theta)$, or equivalently $u$ is more log-convex in $\theta$ at higher actions $a$.\footnote{In the translation-invariant subcase, log-supermodularity of $u_\theta$ simplifies to log-concavity of $T'$.} To see why, note that for any strictly single-peaked triple $(a_1,\theta_1)$, $(a_2,\theta_2)$, $(a_1,\theta_3)$, the perturbation that moves mass on $\theta_1$ and $\theta_3$ from $a_1$ to $a_2$ and moves mass on $\theta_2$ in the opposite direction, so as to hold fixed the receiver's marginal utility conditional on being recommended $a_1$ as well as the total probability of each recommendation, has the effect of increasing the receiver's marginal utility conditional on being recommended $a_2$. This follows because, by log-supermodularity of $u_\theta$, for the receiver's expected marginal utility the marginal rate of substitution between ``shifting weight from $\theta_1 $ to $\theta_2$'' and ``shifting weight from $\theta_2$ to $\theta_3$'' is higher at $a_1$ than $a_2$. Finally, when $V$ is state-independent and increasing in $a$, this perturbation increases the sender's expected utility.\footnote{In the linear receiver and state-independent sender cases, the sufficient conditions for the optimality of strict single-dipped/-peaked disclosure in Theorem \ref{t:SDPD} are ``almost necessary,'' because the condition $|S|\neq 0$ on  $A\times \ol\Theta$ implies that $|S|$ has a constant sign on $A\times \ol\Theta$, which can be shown to be equivalent to strict convexity/concavity of $v$ in the linear receiver case, and to strict log-supermodularity/log-submodularity of $u_\theta$ in the state-independent sender case. By Theorem \ref{t:pairwise}, a necessary condition for the optimality of strictly single-dipped/-peaked disclosure is that $|S|\neq 0$ on the restricted domain where $\theta_1<\theta^\star (a)<\theta_3$.}

There are close antecedents to the conditions in Theorem \ref{t:SDPD} for the linear receiver and state-independent sender cases. In MOT, \citet{BJ} introduce the notions of single-dipped/-peaked outcomes under the names ``left-curtain/right-curtain couplings,'' and show that these outcomes are optimal when the planner's (sender's) marginal utility is convex in $\theta$---a condition referred to in this literature as the ``martingale Spence-Mirrlees condition.''\footnote{More precisely, \citet{BJ} show that the unique optimal outcome is single-dipped in the translation-invariant subcase if $P'$ is strictly convex (Theorem 6.1), and in the separable subcase if $w$ is strictly convex (Theorem 6.3). Theorem 5.1 in \citet{HT} and Theorem 3.3 in \citet{BHT} extend this conclusion to the general linear receiver case where $v(a,\theta)$ is strictly convex in $\theta$. In all these papers, the marginal distribution over actions is fixed.} Earlier, in a model of \emph{partisan gerrymandering}, \citet{FH} show that, under an ``informative signal property,'' if it is optimal to assign two voter types to the same district, then all voter types in between these two must be assigned to districts with less favorable median voters. Partisan gerrymandering is equivalent to the state-independent sender case (as the map-maker cares only about winning seats, and not directly about the composition of districts), the above property of districting is equivalent to single-dippedness, and the informative signal property is equivalent to log-supermodularity of $u_\theta$.\footnote{We further investigate the connection between gerrymandering and persuasion in a companion paper, \citet{KW}.} Theorem \ref{t:SDPD} thus unifies and generalizes these disparate contributions.

\subsection{Uniqueness} 

We now show that strict single-dippedness/-peakedness implies that there is a unique optimal outcome (under a regularity condition).

\begin{theorem}\label{t:brenier}
Let Assumptions \ref{a:smooth}--\ref{a:sc} hold. If $\Gamma^\star$ is strictly single-dipped (-peaked), then $\Gamma^\star_a=\{t_1(a),t_2(a)\}$ for all $a\in A_\Gamma$, where $t_1,t_2:A_\Gamma\rightarrow \Theta$ are measurable functions satisfying $t_1(a)\leq \theta^\star (a)\leq t_2(a)$, $t_2(a)\leq t_2(a')$, and $t_1(a')\notin (t_1(a),t_2(a))$ ($t_1(a)\leq t_1(a')$, and $t_2(a)\notin (t_1(a'),t_2(a'))$) for all $a<a'$ in $A_\Gamma$.	Moreover, if $\phi$ has a density and the set $\{a\in A_\Gamma:t_1(a)<t_2(a)\}$ is the union of finitely many intervals, then the optimal outcome is unique.
\end{theorem}

The regularity condition that the set $\{a\in A_\Gamma:t_1(a)<t_2(a)\}$ is the union of finitely many intervals rules out pathological cases, such as when this set is the complement of the Cantor set. This condition is satisfied in every example in the literature that we know of.

Theorem \ref{t:brenier} is somewhat akin to Brenier's theorem in optimal transport, which shows that the optimal transport plan is unique under a suitable complementarity-type condition, called the twist or generalized Spence-Mirrlees condition (\citealt{Brenier}, \citealt{GangboMcCann}; or see Section 1.3 in \citealt{santambrogio}). In martingale optimal transport, the optimal plan is unique under the martingale Spence-Mirrlees condition (e.g., Proposition 3.5 in \citealt{BHT}), which as noted above coincides with our condition for the optimality of strict single-dippedness in the linear receiver case. The key implication of Theorem \ref{t:brenier} is that the optimal marginal distribution of actions $\alpha_\pi$ is unique; there is no analog of this result in optimal transport, where both marginals are fixed.

To see the intuition, consider the case where $\phi$ is discrete and $t_2$ is strictly increasing: when the recommended action is higher, the highest possible state under the induced posterior is also (strictly) higher. Suppose toward a contradiction that there are two distinct optimal outcomes, $\pi$ and $\pi'$. Since every optimal signal is pairwise, we know that $\pi$ and $\pi'$ have the same conditional distribution: $\pi_a=\pi_a'=\rho_a \delta_{t_1(a)}+(1-\rho_a)\delta_{t_2(a)}$ for all $a\in A_\Gamma$, where $\rho_a$ is pinned down by obedience. Thus, the marginal distributions $\alpha_\pi$ and $\alpha_{\pi'}$ must differ. So let $\hat a=\sup \{a\in A:\alpha_\pi([0,a])\neq\alpha_{\pi'}([0,a])\}$, and consider the state $t_2(\hat a)$. Since $t_2$ is strictly increasing, the state $t_2(\hat a)$ can only induce actions $a\geq \hat a$: thus, $\pi ([0,\hat a),t_2(\hat a))=\pi' ([0,\hat a),t_2(\hat a))=0$. Since the marginals $\alpha_\pi$ and $\alpha_{\pi'}$ coincide on $(\hat a, 1]$ (by the definition of $\hat a$), and the conditionals $\pi_a$ and $\pi'_a$ coincide everywhere, we also have $\pi ((\hat a,1],t_2(\hat a))=\pi' ((\hat a,1],t_2(\hat a))$. Thus, since $\pi (A,t_2(\hat a))=\pi' (A,t_2(\hat a))$ by (P1), we can conclude that $\pi (\hat a,t_2(\hat a))=\pi' (\hat a,t_2(\hat a))$, and hence
\begin{align*}
0=\pi(\hat a,t_2(\hat a))-\pi'(\hat a,t_2(\hat a)) = (1-\rho_{\hat a}) (\alpha_{\pi} (\hat a)-\alpha_{\pi'}(\hat a)).
\end{align*}
As, by convention, we have $\rho_{\hat a} <1$, it follows that $\alpha_{\pi} (\hat a)=\alpha_{\pi'}(\hat a)$, and hence
$\alpha_\pi ([0,\hat a))=\alpha_{\pi'} ([0,\hat a))$. Finally, when $\phi$ is discrete, this implies that $\alpha_\pi$ and $\alpha_{\pi'}$ coincide on $(\hat a-\varepsilon,1]$ for some $\varepsilon >0$, which contradicts the definition of $\hat a$. When instead $\phi$ has a density and our regularity condition holds, a similar argument delivers the same conclusion. Moreover, when $\phi$ has a density, the possibility that $t_2$ may be only weakly increasing does not threaten uniqueness of the optimal joint distribution $\pi$, because the set of states corresponding to flat regions of $t_2$ has measure $0$ (i.e., $\phi(\{\theta : \exists a<a' \text{ s.t. } \theta=\theta_2 (a)=\theta_2 (a')\})=0$).

\section{Full Disclosure and Negative Assortative Disclosure}\label{s:MNAD}

Our last set of results gives conditions for the optimality of two simple disclosure patterns: \emph{full disclosure}, where each state is disclosed, and \emph{negative assortative disclosure}, where all states are paired in a negatively assortative manner.

A note on terminology: in Section \ref{s:assortative}, we considered ``assortativity'' between states and actions, asking whether higher actions should be induced at more or less extreme states. In the current section, ``negative assortative disclosure'' refers to assortativity between pairs of states. One can also view full disclosure as capturing ``positive assortativity'' between states, by matching identical states to form degenerate ``pairs.''

\subsection{Full Disclosure}
An implementable outcome $\pi$ is \emph{full disclosure} if its support is $\cup_{\theta \in \Theta}(a^\star (\delta_\theta),\theta)$, so that each state $\theta$ induces action $a^\star (\delta_\theta)$. There is a unique such outcome.

If for all states $\theta_1$ and $\theta_2$, and all probabilities $\rho $, the sender prefers to split the posterior $\mu=\rho \delta_{\theta_1}+(1-\rho )\delta_{\theta_2}$ into degenerate posteriors
 $\delta_{\theta_1}$ and $\delta_{\theta_2}$, then the
sender prefers full disclosure to any pairwise signal. Since pairwise signals are without loss by part (1) of Theorem \ref{t:pairwise}, full
disclosure is then optimal. Conversely, if the sender strictly prefers not to split $\mu =\rho \delta
_{\theta_1 }+(1-\rho )\delta _{\theta_2}$ into $\delta _{\theta_1 }$ and $%
\delta _{\theta _2}$ for some states $\theta_1$ and $\theta _2$ and some probability $\rho $, then the sender strictly
prefers the pairwise signal that differs from full disclosure only in that it
pools states $\theta_1$ and $\theta_2$ into $\mu $; so full
disclosure is not optimal.\footnote{This argument is valid when $\phi $ has finite support. The general case (Lemma \ref{p:full}) uses duality and is adaptated from part (2) of Proposition 1 in \citet{Kolotilin2017}; we give a simpler proof using Theorem \ref{t:contact} and also establish uniqueness.} Recalling that belief $\mu
=\rho \delta _{\theta_1 }+(1-\rho )\delta _{\theta_2}$ induces action $a ^{\star }(\mu )$ satisfying $\rho u(a ^{\star }(\mu ),\theta_1 )+(1-\rho )u(a ^{\star }(\mu ) ,\theta_2)= 0$, we obtain the following result.

\begin{lemma}
\label{p:full}
Let Assumptions \ref{a:smooth}--\ref{a:sc} hold. Full disclosure is optimal iff, for all $\mu=\rho \delta _{\theta_1}+(1-\rho )\delta _{\theta_2}$ with $\theta_1 <\theta_2$ in $\Theta$ and $\rho \in (0,1)$, we have
\begin{equation}
\rho V(a ^{\star }(\mu),\theta_1) +(1-\rho )V(a ^{\star }(\mu),\theta_2)\leq \rho V(a ^{\star }(\delta_{\theta_1}),\theta_1) +(1-\rho )V(a ^{\star }(\delta_{\theta_2}),\theta_2).   \label{e:full}
\end{equation}
Moreover, full disclosure is uniquely optimal if \eqref{e:full} holds with strict inequality for all such $\mu$.
\end{lemma}

In the linear case, condition \eqref{e:full} holds iff $V$ is convex in $a$. In the state-independent sender case, condition \eqref{e:full} simplifies as follows:
\setcounter{corollary}{1}
\begin{corollary}
\label{c:fd} In the state-independent sender case, full disclosure is optimal iff, for all $\mu =\rho \delta _{\theta_1 }+(1-\rho )\delta _{\theta_2}$ with $\theta_1,\theta_2\in \Theta$ and $\rho\in (0,1) $, we have
\begin{equation}
V\left( a ^{\star }\left( \mu \right) \right) \leq \rho V\left(a ^{\star }\left( \delta _{\theta_1 }\right) \right) +\left(1-\rho \right) V\left( a ^{\star }\left( \delta _{\theta_2}\right)
\right).
\label{e:fullsi}
\end{equation}
\end{corollary}

In a classical one-to-one matching model, \citeasnoun{B73} showed that if the utility from matching two
types $h\left( \theta_1 ,\theta _2\right) $ is supermodular, then it
is optimal to match like types. \citeasnoun{LN} refer to this extreme form
of positive assortative matching as \emph{segregation}. Their Propositions 4
and 9 show that segregation is optimal iff $h\left( \theta_1
,\theta_1 \right) +h\left( \theta _2,\theta_2\right) \geq
2h\left( \theta _1,\theta_2\right) $ for all $\theta_1 ,\theta_2$ (which is a strictly weaker property than supermodularity). In the
context of persuasion, segregation corresponds to full disclosure. Note that
if we fix $p=1/2$ and let $h\left( \theta_1 ,\theta _2\right) =V\left(
a^{\star }\left( \delta _{\theta_1 }/2+\delta _{\theta _2}/2\right)
\right) $, then (\ref{e:fullsi}) reduces to \citeauthor{LN}'s
condition. Intuitively, full disclosure is ``less likely'' be optimal in persuasion than in classical matching, because in persuasion the designer has an extra degree of freedom $\rho $ in designing matches.

In the linear receiver case, there is a simple sufficient condition for \eqref{e:full}:
\begin{manualcorollary}{\ref{c:fd}'}
\label{c:fds} In the linear receiver case, full
disclosure is optimal if $V(a ,\theta )$ is convex in $a$ and
satisfies $V(\theta_1 ,\theta_2)+V(\theta_2,\theta_1 )\leq V(\theta_1 ,\theta_1 )+V(\theta_2,\theta_2)$ for all $\theta_1,\theta_2\in \Theta$.
\end{manualcorollary}

A sufficient condition for $V(\theta_1 ,\theta_2)+V(\theta_2,\theta_1 )\leq V(\theta_1 ,\theta_1 )+V(\theta_2,\theta_2)$
is supermodularity of $V$: for all $\theta_1 <\theta_2$ and $a_1<a_2$, $V(a_1 ,\theta_1 )+V(a_2,\theta_2)\geq V(a_1 ,\theta_2)+V(a_2,\theta_1 )$. Thus, in the linear receiver case, full
disclosure is optimal whenever the sender's utility is convex in $a$ and supermodular in $\left(a ,\theta \right) $. This sufficient condition for full disclosure
generalizes that given by \citet{RS} for the separable subcase.\footnote{Their condition is that $w$ is increasing in $\theta$ and $G$ is convex in $a$, where $V(a,\theta)=w(\theta)G(a)$. In the sub-subcase with $G(a)=a$, \eqref{e:full} holds iff $w$ is increasing in $\theta$, because \eqref{e:full} simplifies to $\rho (1-\rho)(w(\theta_2)-w(\theta_1))(\theta_2-\theta_1)\geq 0$.}

When the prior has full support and the contact set is pairwise (e.g., the twist condition holds), full disclosure is uniquely optimal whenever it is optimal.\footnote{Full support is necessary for this result, as shown by Example \ref{e:1}.} To see the intuition, suppose full disclosure is optimal, and suppose there is another optimal signal that pools some states $\theta_1$ and $\theta_2$ to induce an action $a$. Then the signal that discloses all other states while pooling $\theta_1$ and $\theta_2$ to induce $a$ is also optimal. But then the signal that discloses all other states while pooling $\theta_1$, $\theta_2$, and the third state $\theta^\star (a)\neq \theta_1, \theta_2$  to induce $a$ is also optimal, and this signal is not pairwise.
\begin{theorem}\label{t:FDu} Let $\Theta=[0,1]$ and let Assumptions \ref{a:smooth}--\ref{a:sc} hold. If the contact set is pairwise and full disclosure is optimal, then full disclosure is uniquely optimal.
\end{theorem}

\subsection{Negative Assortative Disclosure}
A set $\Gamma^\dagger$ is \emph{single-dipped (-peaked) negative assortative} if there exist a decreasing (increasing) function $t_1:A_{\Gamma^\dagger} \rightarrow \Theta$ and an increasing (decreasing) function $t_2:A_{\Gamma^\dagger} \rightarrow \Theta$ such that $t_1(a)\leq \theta^\star (a)\leq t_2(a)$ and $\Gamma^\dagger_a =\{t_1(a),t_2(a)\}$ for all $a$. An outcome $\pi$ is \emph{single-dipped (-peaked) negative assortative} if it is concentrated on such a set, so that states $t_1 (a)$ and $t_2 (a)$ are pooled to induce action $a$.

The main result of this section is that if strictly single-dipped (-peaked) disclosure is optimal and the sender strictly prefers to pool any two states, then single-dipped (-peaked) negative assortative disclosure is optimal. Moreover, if the prior has a density, then the optimal outcome is unique (by Theorem \ref{t:brenier}) and is characterized as the solution to a system of two ordinary differential equations.

To see the intuition, note that if strictly single-dipped disclosure is optimal, then any two pairs of pooled states $\{\theta_1,\theta_3\}$ and $\{\theta'_1,\theta'_3\}$ with (without loss) $\theta_1<\theta_3$, $\theta'_1<\theta'_3$, and $\theta_1\leq \theta'_1$, must be either ordered (i.e., $\theta_1<\theta_3\leq\theta'_1<\theta'_3$) or nested (i.e., $\theta_1\leq \theta'_1<\theta'_3\leq\theta_3$). This follows because if the pairs overlap (i.e., $\theta_1<\theta'_1<\theta_3<\theta'_3$), then either $(\theta_1, \theta'_1, \theta_3)$ or $(\theta'_1, \theta_3, \theta'_3)$, together with the corresponding actions, would form a single-peaked triple. Hence, for any pair of pooled states $\{\theta_1,\theta_3\}$, there must exist a disclosed state $\theta_2 \in (\theta_1,\theta_3)$: intuitively, there must exist pairs of pooled states in the interval $(\theta_1,\theta_3)$ that are closer and closer together, until the pair degenerates into a single disclosed state. Therefore, if any two pairs of pooled states $\{\theta_1,\theta_3\}$ and $\{\theta'_1,\theta'_3\}$ are ordered, there would exist two distinct disclosed states $\theta_2 \in (\theta_1,\theta_3)$ and $\theta'_2 \in (\theta'_1,\theta'_3)$. But if the sender strictly prefers to pool any two states, this is impossible. Finally, if pairs of pooled states cannot overlap or be ordered, the only remaining possibility is that all pairs of pooled states are nested: that is, disclosure is negative assortative.\footnote{In this argument, the existence of the two disclosed states relies on the assumption that $\Theta=[0,1]$. Example \ref{ex:NAD} in Appendix \ref{s:examples} shows that when $\Theta\neq [0,1]$, the set $\Gamma^\star$ is not necessarily negative assortative even if $\Gamma^\star$ is strictly single-dipped (-peaked) and \eqref{e:nd} holds for all $\theta_1<\theta_2$.}

\begin{theorem}\label{t:NAD}
Let $\Theta=[0,1]$, and let Assumptions \ref{a:smooth}--\ref{a:sc} hold. If $\Gamma^\star$ is strictly single-dipped (-peaked) and for all $\theta_1<\theta_2$ there exists $\rho\in (0,1)$ such that
\begin{equation}
\rho V(a ^{\star }(\mu),\theta_1) +(1-\rho )V(a ^{\star }(\mu),\theta_2)> \rho V(a ^{\star }(\delta_{\theta_1})),\theta_1) +(1-\rho )V(a ^{\star }(\delta_{\theta_2}),\theta_2),   \label{e:nd}
\end{equation}
with $\mu=\rho \delta _{\theta_1}+(1-\rho )\delta _{\theta_2}$, then $\Gamma^\star $ is single-dipped (-peaked) negative assortative. Moreover, if $\phi$ has a density $f$, then the functions $t_1$ and $t_2$ are continuous and solve the system of the two differential equations,
\begin{gather}
u(a,t_1(a))(-\df \phi ([0,t_1(a)]))+u(a,t_2(a))\df \phi ([0,t_2(a)])=0, \label{e:obed} \\
\begin{gathered}
\frac{\df}{\df a} \left(\frac{v(a,t_1(a))u(a,t_2(a))-v(a,t_2(a))u(a,t_1(a))}{u(a,t_1(a))u_a(a,t_2(a))-u(a,t_2(a))u_a(a,t_1(a))}\right) \\
=\frac{v(a,t_1(a))u_a(a,t_2(a))-v(a,t_2(a))u_a(a,t_1(a))}{u_a(a,t_1(a))u(a,t_2(a))-u_a(a,t_2(a))u(a,t_1(a))},	\label{e:q'} 
\end{gathered}	
\end{gather}
for all $a\in (\ul a,\ol a]$ where $\ul a=\min A_\Gamma$ and $\ol a=\max A_\Gamma$, with the boundary conditions
\begin{equation}
\begin{gathered}
	(t_1(\ol a),t_1(\ul a),t_2(\ul a),t_2(\ol a))=(0,\theta^\star (\ul a),\theta^\star (\ul a),1)\\
	\left((t_1(\ul a),t_1(\ol a),t_2(\ol a),t_2(\ul a))=(0,\theta^\star (\ol a),\theta^\star (\ol a),1)\right). \label{e:boundary}
\end{gathered}
\end{equation}
\end{theorem}

Similarly to equation \eqref{e:full} in the previous subsection, equation \eqref{e:nd} simplifies in special cases. In the linear case, \eqref{e:nd} holds iff $V$ is strictly concave in $a$.\footnote{In the linear case, $V$ is strictly concave iff no disclosure is uniquely optimal for all priors, by Corollary 1 in \citet{KMZ}.} In the state-independent sender case, it holds iff $V(a ^{\star }(\mu))> \rho V(a ^{\star }(\delta_{\theta_1})) +(1-\rho )V(a ^{\star }(\delta_{\theta_2}))$. In the linear receiver case, it holds if $V(a,\theta)$ is concave in $a$ and satisfies $V(\theta_1 ,\theta_2)+V(\theta_2,\theta_1 )> V(\theta_1 ,\theta_1 )+V(\theta_2,\theta_2)$ for all $\theta_1<\theta_2$; a sufficient condition for the latter property is strict submodularity of $V$. These conditions generalize the sufficient condition for pooling given by \citet{RS} for the separable subcase.\footnote{Their condition is that $w$ is strictly decreasing in $\theta$ and $G$ is concave in $a$, where $V(a,\theta)=w(\theta)G(a)$.  In the sub-subcase with $G(a)=a$, \eqref{e:nd} holds iff $w$ is strictly decreasing in $\theta$.}
 
To understand the differential equations, note that if $t_1$ and $t_2$ are differentiable then \eqref{e:obed} can be written as $$u(a,t_1(a))f(t_1(a))t_1'(a)=u(a,t_2(a))f(t_2(a))t_2'(a).$$ This is the obedience constraint conditional on recommendation $a$, as the posterior conditional on recommendation $a$ is 
$$
\pi_a=\frac{-f(t_1(a))t_1'(a)}{-f(t_1(a))t_1'(a)+f(t_2(a))t_2'(a)}\delta_{t_1(a)}+\frac{f(t_2(a))t_2'(a)}{-f(t_1(a))t_1'(a)+f(t_2(a))t_2'(a)}\delta_{t_2(a)}.\footnote{This equation is a version of the Monge-Ampere equation in optimal transport.}
$$ 
In addition, \eqref{e:q'} results from solving the system of equations (from the sender's FOC, \eqref{e:FOC}), 
\begin{gather*}
v(a,t_1(a))+ q(a)u_a(a,t_1(a))+q'(a)u(a,t_1(a))=0, \\
v(a,t_2(a))+ q(a)u_a(a,t_2(a))+q'(a)u(a,t_2(a))=0,
\end{gather*}
for $q(a)$ and $q'(a)$, and recalling that $q'$ is the derivative of $q$.\footnote{This argument shows that \eqref{e:q'} holds for all $a$ with $t_1(a)<\theta^\star (a)<t_2(a)$ even if $\Gamma^\star$ is only pairwise and not also single-dipped/-peaked.} Finally, the boundary condition \eqref{e:boundary} for the single-dipped case follows because the lowest induced action $\underline{a}$ is induced at the disclosed state $\theta^\star (\ul{a})=t_1 (\ul{a})=t_2 (\ul{a})$, and the highest induced action $\ol a$ is induced at states $0=t_1 (\ol a)$ and $1=t_2 (\ol a)$. The boundary condition for the single-peaked case is analogous.\footnote{In the linear receiver case, \eqref{e:q'} simplifies to
\[
\frac{\df}{\df a} \left(v(a,t_1(a))\frac{t_2(a)-a}{t_2(a)-t_1(a)}+v(a,t_2(a))\frac{a-t_1(a)}{t_2(a)-t_1(a)}\right) =-\frac{v(a,t_2(a))-v(a,t_1(a))}{t_2(a)-t_1(a)}.
\]
Geometrically, this says that the slope of the curve $a\rightarrow \E_{\pi_a} [v(a,\theta)]$ is equal to the negative of the slope of the secant passing through the points $(t_1(a),v(a,t_1(a)))$ and  $(t_2(a),v(a,t_2(a)))$. \citet{NP} derive this condition for the separable sub-subcase with $v(a,\theta)=w(\theta)$.}

Next, we give primitive conditions on $V$ and $u$ for \eqref{e:nd} to hold, and hence for negative assortative disclosure to be optimal.

\begin{corollary}\label{c:ndSDD}
Let $\Theta=[0,1]$, let all partial derivatives of $V(a,\theta)$ and $u(a,\theta)$ of order at most $2$ be differentiable, and let Assumptions \ref{a:smooth}--\ref{a:u_s} hold. Furthermore, let $u_{a\theta}(a,\theta)/u_{\theta}(a,\theta)$ and  $v_\theta (a,\theta)/u_\theta(a,\theta)$  be increasing (decreasing) in $\theta$ for all $a$, with at least one of these functions being strictly increasing (decreasing). Then for all $\theta_1<\theta_2$ there exists $\rho\in (0,1)$ such that \eqref{e:nd} holds iff 
\begin{equation}\label{e:ndSDD}
\begin{gathered}
v_a(a,\theta^\star (a)) \leq \tfrac{v(a,\theta^\star (a)) u_{aa}(a,\theta^\star (a))}{u_a(a,\theta^\star (a))}
+2\tfrac{v_\theta(a,\theta^\star (a)) u_a(a,\theta^\star (a)) - v(a,\theta^\star (a)) u_{a\theta}(a,\theta^\star (a))}{u_\theta(a,\theta^\star (a))},
\end{gathered}
\end{equation}
for all $a\in A$. In particular, if \eqref{e:ndSDD} holds in addition to the above conditions (with monotonicity of $v_\theta (a,\theta)/u_\theta (a,\theta)$  strengthened to monotonicity of $v_\theta (a_2,\theta)/u_\theta (a_1,\theta)$ for all $a_2 \geq (\leq) a_1$), then $\Gamma^\star$ is single-dipped (-peaked) negative assortative.
\end{corollary}

Equation \eqref{e:ndSDD} is a local necessary condition for \eqref{e:nd}: if  \eqref{e:ndSDD} fails, then  \eqref{e:nd} also fails for $\theta_1<\theta_2$ sufficiently close to $\theta^\star(a)$. When $\Gamma^\star$ is strictly single-dipped (-peaked), this local necessary condition turns out to be globally sufficient for \eqref{e:nd}. Equation \eqref{e:ndSDD} simplifies dramatically in some special cases. In the linear receiver case, \eqref{e:ndSDD} simplifies to $v_a(a,a)+2v_\theta(a,a) \leq 0$; in the translation-invariant subcase of the linear receiver case, this simplifies further to $P''(0)\geq 0$. In the translation-invariant subcase of the state-independent sender case, \eqref{e:ndSDD} simplifies to ${v_a(a)}/{v(a)}\leq {T''(0)}/{T'(0)}$.

We give some examples of optimal single-dipped negative assortative disclosure.

\begin{example}\label{e:RS} Consider the linear receiver case with $A=\Theta=[1/e,e]$, $f(\theta)=1/(2\theta)$, and $V(a,\theta)=a/\theta$.\footnote{In Examples \ref{e:RS} and \ref{ex:segpair} and Appendix \ref{s:ZZ}, $\Theta$ and $A$ are compact intervals, which can be rescaled to the unit interval.} We claim that the unique optimal outcome matches each state $\theta_1\in [1/e, 1]$ with state $\theta_2=1/\theta_1$ with equal weights, so that the induced action is $a={\theta_1}/{2}+{1}/({2\theta_1})$. Thus, $t_1(a)=a-\sqrt{a^2-1}$, and $t_2=a+\sqrt{a^2-1}$ for all $a\in A_{\Gamma} =[1,{e}/{2}+{1}/(2e)]$.

 Indeed, by Theorem \ref{t:SDPD}, $\Gamma^\star$ is strictly single-dipped, since $w(\theta)=1/\theta$ is strictly convex. By Corollary \ref{c:ndSDD}, \eqref{e:nd} holds, since $w'<0$. Hence, by Theorem \ref{t:NAD}, $\Gamma^\star $ is single-dipped negative assortative and satisfies \eqref{e:obed}--\eqref{e:boundary}. For $\theta_2=1/\theta_1$ and $a={\theta_1}/{2}+{1}/({2\theta_1})$, \eqref{e:obed} holds because
\begin{gather*}
u(a,\theta_2)=\left( \frac{1}{2\theta_1}-\frac {\theta_1}{2}\right)=-\left(\frac{\theta_1}{2}-\frac{1}{2\theta_1}\right)=-u(a,\theta_1),\\
f(\theta_2)\frac {\df \theta_2}{\df a}=\frac{1}{{2}/{\theta_1}} \left(-\frac {1}{\theta_1^2}\frac{\df \theta_1}{\df a}\right)=-\frac {1}{2\theta_1}\frac{\df \theta_1}{\df a}=-f(\theta_1)\frac{\df \theta_1}{\df a},
\end{gather*}
\eqref{e:q'} holds because
\begin{gather*}
\frac{\df}{\df a} \left(w(\theta_1)\frac{1}{2}+w(\theta_2)\frac 12\right)=\frac{\df}{\df a}\left(\frac{1}{2\theta_1}+\frac {\theta_1}{2}\right)=\frac{\df }{\df a}a=1,\\
\frac{w(\theta_2)-w(\theta_1)}{\theta_2-\theta_1}=\frac{\theta_1-{1}/{\theta_1}}{{1}/{\theta_1}-\theta_1}=-1,
\end{gather*}
and \eqref{e:boundary} holds because $1/(1/e)=e$ and $1/1=1$. Note that we can instead solve this example by using Theorem \ref{t:contact} directly, because, for $q(a)=a$, the function
$V(a,\theta)+q(a)u(a,\theta)={a}/{\theta} +a(\theta -a)$ is maximized at $a=\theta/2+1/(2\theta)$ for all $\theta \in [1/e,e]$.
\end{example}
  
\begin{example}[Quantile Persuasion] \label{e:quantile} Consider the quantile sub-subcase of the state-independent sender case, where $u(a,\theta)= \1 \{\theta\geq a\}-\kappa$ with $\kappa\in (0,1)$.
Let $\phi$ have a density on $[0,1]$. Assuming that the receiver breaks ties in favor of the sender, we obtain that, for $\theta_1<\theta_2$,
\[
a^\star(\rho \delta_{\theta_1}+(1-\rho )\delta_{\theta_2})=
\begin{cases}
\theta_2, &\rho \leq 1-\kappa,\\
\theta_1, &\rho >1-\kappa.
\end{cases}
\]
Note that \eqref{e:nd} always holds for $\rho \in (0,1-\kappa)$. We claim that there exists an optimal single-dipped negative assortative outcome $\pi$ with $\alpha_\pi([a,1])=\phi ([a,1])/\kappa$ and $\pi_a=(1-\kappa)\delta_{t_1(a)}+\kappa \delta_{t_2(a)}$ for all $a\in \supp (\alpha_\pi)=[\ul a,1]$, where $t_1(a)$ solves $\kappa \phi ([0,t_1(a)])=(1-\kappa)\phi ([a,1])$, and $\ul a$ solves $\kappa \phi ([0,\ul a])=(1-\kappa)\phi ([\ul a,1])$. See Section \ref{s:quantileproof} for the proof. A notable feature of this outcome is that, with the informed receiver interpretation, it would remain optimal even if the sender knew the receiver's type and could condition disclosure on it.
\end{example}

\begin{example}[A Stochastic Optimal Signal\footnote{This example is an adaptation of Example 2 in \citet{KW}.}] \label{ex:segpair}
Consider the translation-invariant subcase of the state-independent sender case. Let $A=\Theta=[-1,3]$, let $\phi$ have a density $f$ with $f(-a)\geq 3f(3a)$ for all $a\in (0,1]$, let $u(a,\theta)=T(\theta-a$) with $T(0)=0$ and strictly log-concave $T'$, and let $V(a,\theta)=T(2a)$. With the informed receiver interpretation, this captures a case where, for example, $\kappa=1/2$, the distribution of $\varepsilon$ is $N(0,\sigma^2)$, and the distribution of $t$ is $N(0,(\sigma/2)^2)$.\footnote{By symmetry and strict log-concavity of $T'$, $v_a(a)/v(a)=2T''(a)/T'(a) >(<)T''(0)/T'(0)=0$ for $0>(<)a,$ showing that \eqref{e:ndSDD} fails for $a<0$, and thus Theorem \ref{t:NAD} does not apply.}

 We claim that $A_\Gamma=[-1,1]$ and $\Gamma_a=\{t_1(a),t_2(a)\}$ for all $a\in A_\Gamma$ where
\[
t_1(a)=
\begin{cases}
a, &a\in [-1,0],\\
-a, &a\in (0,1],
\end{cases}
\quad \text{and}\quad 
t_2(a)=
\begin{cases}
a, &a\in [-1,0),\\
3a, &a\in (0,1],
\end{cases}
\]
so that $\pi_a=\rho_a\delta_{t_1(a)}+(1-\rho_a)\delta_{t_2(a)}$ with $\rho_a=1/2$ for all $a\in A_\Gamma$, and $\alpha_\pi$ has a density $h$ given by
\begin{align*}
h(a)&=
\begin{cases}
6f(3a), &a\in (0,1],\\
f(-a)-3f(3a), &a\in [-1,0). 
\end{cases}
\end{align*}
Note that the unique optimal outcome is single-dipped negative assortative iff $f(-a)=3f(3a)$ for all $a\in (0,1]$. In contrast, if $f(-a)>3f(3a)$ for all $a\in (0,1]$, then each state $\theta\in [-1,0)$ is mixed between recommendations $a=\theta$ and $a=-\theta$. Specifically, the conditional distribution $\pi_\theta$ of $a$ given $\theta$ is
\[
\pi_\theta=
\begin{cases}
\delta_{\theta/3}, &\theta\in [0,3],\\
\frac{f(\theta)-3f(-3\theta)}{f(\theta)}\delta_{\theta}+\frac{3f(-3\theta)}{f(\theta)}\delta_{-\theta}, &\theta\in [-1,0). 
\end{cases}
\]
See Figure \ref{f:segpair}. Note that in this case the unique optimal signal randomizes conditional on the state, even though the state is atomless. See Section \ref{s:segpair} for the proof.

\begin{figure}[t]
\centering
\begin{tikzpicture}[scale = 1]
	\small
	\begin{axis}		
		[axis x line = middle,
		axis y line = middle,
		xmin = -1.3, xmax = 3.3,
		ymin = -1.3, ymax = 1.3,
		xlabel=$\theta$,
		ylabel=$a$,
		xtick={-1.01, 0, 3.01},
		xticklabels={$-1$, $0$, $3$},
		ytick={-1.01, 0, 1.01},
		yticklabels={$-1$, , $1$},
		clip=false]

		\draw [dashed, black]			(-1, -1) -- (-1, 1);
		\draw [dashed, black]			(-1, -1) -- (0, -1);
		\draw [dashed, black]			(3, 0) -- (3,1);

		\draw [thick, solid, black]		(-1,-1) -- (0,0);
		\draw [thick, solid, black]		(-1,1) -- (0,0);
		\draw [thick, solid, black]		(0,0) -- (3,1);
		\draw [solid, red]				(-1, 1) -- (3,1);
		\draw [solid, red]				(-3/4, 3/4) -- (9/4, 3/4);
		\draw [solid, red]				(-2/4, 2/4) -- (6/4, 2/4);
		\draw [solid, red]				(-1/4, 1/4) -- (3/4, 1/4);
	\end{axis}
\end{tikzpicture}
	\caption{The Optimal Outcome in Example \ref{ex:segpair}}
	\caption*{\emph{Notes:} The contact set equals the three black line segments. The red line segments indicate pairs of states that may be pooled at an optimal outcome. If the prior density satisfies $f(-a)>3f(3a)$ for all $a\in (0,1]$, the unique optimal outcome is supported on the entire contact set. In this case, for each state $\theta <0$, the unique optimal signal randomizes between disclosing $\theta$ (and inducing action $\theta$) and pooling $\theta$ with state $-3\theta$ (and inducing action $-\theta$).}
	\label{f:segpair}
\end{figure}

\end{example}

\section{Conclusion}
\label{s:conclusion}

This paper has developed a first-order approach to persuasion with non-linear preferences, based on duality and connections to optimal transport. Our substantive results provide conditions for all optimal signals to be pairwise, for higher actions to be induced at more or less extreme states, and for full or negative assortative disclosure to be optimal. In some cases, we can characterize optimal signals as the solution to a pair of ordinary differential equations, or even solve them in closed form.

Our analysis generalizes prior studies that have found single-dipped or single-peaked disclosure to be optimal in various contexts. We have already described how our results generalize those of \citet{FH} in the gerrymandering literature and \citet{BJ} and others in the MOT literature. In Appendix \ref{s:other}, we show how our analysis accommodates three well-known applications in the persuasion literature: \citeauthor{ZZ}'s \citeyear{ZZ} model of information disclosure in contests, \citeauthor{GS2018}'s \citeyear{GS2018} model of persuasion with a privately informed receiver whose type is affiliated with the state, and \citeauthor{GL}'s \citeyear{GL} model of optimal stress tests.\footnote{These applications also illustrate some new technical points. Appendix \ref{s:ZZ} illustrates how directly applying Theorem \ref{l:ssdd} can yield weaker sufficient conditions for the optimality of single-dipped/-peaked disclosure than those presented in Theorem \ref{t:SDPD}. Appendices \ref{s:GS} and \ref{s:GL} illustrate how our analysis extends when some of our assumptions are violated: in Appendix \ref{s:GS}, Assumption \ref{a:int} fails, so the receiver's optimal action may be at the boundary and thus violate the first-order condition; in Appendix \ref{s:GL}, Assumption \ref{a:v>0} fails, so the sender only weakly prefers higher actions.} We hope that an awareness of the common theoretical structure in these papers will facilitate further progress on such models.

We close with a few open issues. First, while the persuasion literature has made progress by allowing unrestricted disclosure policies, the pairwise signals that we have highlighted are not always realistic. (For example, in reality it is probably not feasible to design a stress test that pools only the weakest and strongest banks.) An alternative, complementary approach is to restrict the sender to partitioning the state space into intervals, as in \citet{Rayo} or \citet{OR}. An interesting observation is that, at least in the separable subcase of our model considered by \citeauthor{Rayo} and \citeauthor{OR}, our condition \eqref{e:nd} is equivalent to the condition that complete pooling is uniquely optimal among monotone partitions for all prior distributions. This suggests that, under our conditions for the optimality of single-dipped/-peaked disclosure, negative assortative disclosure might be the optimal unrestricted disclosure policy for all priors iff no-disclosure is the optimal monotone policy for all priors. More generally, analyzing the relationship between the optimal pairwise signals we have characterized and simpler signals such as monotone partitions is an important direction for future research.

Second, in the informed receiver interpretation of our model mentioned in Section \ref{s:model}, our analysis pertains to disclosure mechanisms that do not first elicit the receiver's type, or \emph{public persuasion} in the language of \citet{KMZL}. Public persuasion turns out to be without loss in \citet{KMZL}, as well as in \citet{GS2018}. It would be interesting to investigate conditions for the optimality of public persuasion in our more general model, and in particular to see how they relate to our conditions for the optimality of full or negative assortative disclosure.

Finally, our model could be generalized to allow multidimensional states or actions. We suspect that the results of Sections \ref{s:duality}--\ref{s:pairwise} can be generalized, although our analysis is facilitated by the existence of a bijection between actions $a$ and states $\theta^\star (a)$ such that $u(a,\theta^\star (a))=0$ (cf. Assumption \ref{a:sc}). Generalizing our other results would require a more general notion of single-dippedness/-peakedness. With a unidimensional action and a multidimensional state, one can still define a notion of single-dippedness as inducing higher actions at more extreme states; with multidimensional actions, the appropriate generalization is unclear.\footnote{Possibly relevant recent work on multidimensional martingale optimal transport includes \citet{GKL} and \citet{DeMarchTouzi}.} For results on multidimensional persuasion focusing on the linear case, see \citet{DK}.

\bibliographystyle{econometrica}
\bibliography{persuasionlit}

@article{GT,
  title={Information Design in Competitive Insurance Markets},
  author={Garcia, Daniel and Tsur, Matan},
  journal={Journal of Economic Theory},
  volume={191},
  pages={105--160},
  year={2021},
  publisher={Elsevier}
}

@article{Mirrlees75,
  title={The Theory of Moral Hazard and Unobservable Behaviour: Part I},
  author={Mirrlees, James A},
  journal={Review of Economic Studies},
  volume={66},
  number={1},
  pages={3--21},
  year={1999},
  publisher={Wiley-Blackwell}
}

@article{DeMarchTouzi,
  title={Irreducible Convex Paving for Decomposition of Multidimensional Martingale Transport Plans},
  author={De March, Hadrien and Touzi, Nizar},
  journal={Annals of Probability},
  volume={47},
  number={3},
  pages={1726--1774},
  year={2019}
}

@article{GKL,
  title={Structure of Optimal Martingale Transport Plans in General Dimensions},
  author={Ghoussoub, Nassif and Kim, Young-Heon and Lim, Tongseok},
  journal={Annals of Probability},
  volume={47},
  number={1},
  pages={109--164},
  year={2019}
}

@article{Brenier,
  title={Polar Factorization and Monotone Rearrangement of Vector-Valued Functions},
  author={Brenier, Yann},
  journal={Communications on Pure and Applied Mathematics},
  volume={44},
  number={4},
  pages={375--417},
  year={1991}
}

@article{GangboMcCann,
  title={The Geometry of Optimal Transportation},
  author={Gangbo, Wilfrid and McCann, Robert},
  journal={Acta Mathematica},
  volume={177},
  number={2},
  pages={113--161},
  year={1996}
}

@article{MCS,
  title={Persuasion by Dimension Reduction},
  author={Malamud, Semyon and Schrimpf, Andreas},
  journal={Working paper},
  year={2021}
}

@article{ABS22,
  title={Persuasion as Transportation},
  author={Arieli, Itai and Babichenko, Yakov and Sandomirskiy, Fedor},
  journal={Working paper},
  year={2022}
}

@article{LinLiu,
  title={Credible Persuasion},
  author={Lin, Xiao and Liu, Ce},
  journal={Working paper},
  year={2022}
}

@article{PRS,
  title={Test Design Under Falsification},
  author={Perez-Richet, Eduardo and Skreta, Vasiliki},
  journal={Econometrica},
  volume={90},
  number={3},
  pages={1109--1142},
  year={2022}
}

@article{GLP,
  title={The Value of Data Records},
  author={Galperti, Simone and Levkun, Aleksandr and Perego, Jacopo},
  journal={Working paper},
  year={2021}
}

@article{DizdarKovac,
  title={A Simple Proof of Strong Duality in the Linear Persuasion Problem},
  author={Dizdar, Deniz and Kov{\'a}{\v{c}}, Eugen},
  journal={Games and Economic Behavior},
  volume={122},
  pages={407--412},
  year={2020},
  publisher={Elsevier}
}

@article{KMZ,
  title={Censorship as Optimal Persuasion},
  author={Kolotilin, Anton and Mylovanov, Timofiy and Zapechelnyuk, Andriy},
  journal={Theoretical Economics},
  volume={17},
  number={2},
  pages={561--585},
  year={2022},
  publisher={Wiley Online Library}
}

@article{BHP,
  title={Model-Independent Bounds for Option Prices--a Mass Transport Approach},
  author={Beiglb{\"o}ck, Mathias and Henry-Labordere, Pierre and Penkner, Friedrich},
  journal={Finance and Stochastics},
  volume={17},
  number={3},
  pages={477--501},
  year={2013},
  publisher={Springer}
}

@article{BNT,
  title={Complete Duality for Martingale Optimal Transport on the Line},
  author={Beiglb{\"o}ck, Mathias and Nutz, Marcel and Touzi, Nizar},
  journal={Annals of Probability},
  volume={45},
  number={5},
  pages={3038--3074},
  year={2017}
}

@article{Rayo,
  title={Monopolistic Signal Provision},
  author={Rayo, Luis},
  journal={BE Journal of Theoretical Economics},
  volume={13},
  number={1},
  pages={27--58},
  year={2013},
  publisher={De Gruyter}
}

@article{OR,
  title={Conveying Value via Categories},
  author={Onuchic, Paula and Ray, Debraj},
  journal={Working paper},
  year={2022}
}

@article{HT,
  title={An Explicit Martingale Version of the One-Dimensional Brenier Theorem},
  author={Henry-Labord{\`e}re, Pierre and Touzi, Nizar},
  journal={Finance and Stochastics},
  volume={20},
  number={3},
  pages={635--668},
  year={2016},
  publisher={Springer}
}

@article{BHT,
  title={Monotone Martingale Transport Plans and Skorokhod Embedding},
  author={Beiglb{\"o}ck, Mathias and Henry-Labord{\`e}re, Pierre and Touzi, Nizar},
  journal={Stochastic Processes and their Applications},
  volume={127},
  number={9},
  pages={3005--3013},
  year={2017},
  publisher={Elsevier}
}

@book{anderson1987,
	Author = {Anderson, E. J. and Peter Nash},
	Publisher = {John Wiley \& Sons, New York},
	Title = {{Linear Programming in Infinite-Dimensional Space}},
	Year = {1987}}

@book{aliprantis2006,
	Author = {Aliprantis, Charalambos D. and Border, Kim},
	Publisher = {Springer},
	Title = {{Infinite Dimensional Analysis: A Hitchhiker's Guide}},
	Year = {2006}}

@article{B73,
  title={A Theory of Marriage: Part I},
  author={Becker, Gary S},
  journal={Journal of Political Economy},
  volume={81},
  number={4},
  pages={813--846},
  year={1973},
  publisher={The University of Chicago Press}
}

@book{villani,
  title={Optimal Transport: Old and New},
  author={Villani, C{\'e}dric},
  volume={338},
  year={2009},
  publisher={Springer}
}

@book{abs,
  title={Lectures on Optimal Transport},
  author={Ambrosio, Luigi and Bru{\'e}, Elia and Semola, Daniele},
  year={2021},
  publisher={Springer}
}

@book{santambrogio,
  title={Optimal Transport for Applied Mathematicians},
  author={Santambrogio, Filippo},
  journal={Birk{\"a}user, NY},
  volume={55},
  pages={94},
  year={2015},
  publisher={Springer}
}

@article{winkler,
  title={Extreme Points of Moment Sets},
  author={Winkler, Gerhard},
  journal={Mathematics of Operations Research},
  volume={13},
  number={4},
  pages={581--587},
  year={1988},
  publisher={INFORMS}
}

@book{AM,
  title={Repeated Games with Incomplete Information},
  author={Aumann, Robert J and Maschler, Michael},
  year={1995},
  publisher={MIT press}
}

@article{KX,
  title={An Optimal Transport Problem with Backward Martingale Constraints Motivated by Insider Trading},
  author={Kramkov, Dmitry and Xu, Yan},
  journal={Annals of Applied Probability},
  volume={32},
  number={1},
  pages={294--326},
  year={2022},
  publisher={Institute of Mathematical Statistics}
}

@article{T18,
  title={Bayesian Persuasion with Quadratic Preferences},
  author={Tamura, Wataru},
  journal={Working paper},
  year={2018}
}

@article{ABSY20,
  title={Optimal Persuasion via Bi-Pooling},
  author={Arieli, Itai and Babichenko, Yakov and Smorodinsky, Rann and Yamashita, Takuro},
  journal={Theoretical Economics},
  volume={forthcoming},
  year={2022}
}

@article{KMS,
  title={Extreme Points and Majorization: Economic Applications},
  author={Kleiner, Andreas and Moldovanu, Benny and Strack, Philipp},
  journal={Econometrica},
  volume={89},
  number={4},
  pages={1557--1593},
  year={2021},
  publisher={Wiley Online Library}
}

@article{IP20,
  title={Persuasion in Global Games with Application to Stress Testing},
  author={Inostroza, Nicolas and Pavan, Alessandro},
  journal={Working paper},
  year={2022}
}

@article{GL,
  title={Stress Tests and Information Disclosure},
  author={Goldstein, Itay and Leitner, Yaron},
  journal={Journal of Economic Theory},
  volume={177},
  pages={34--69},
  year={2018},
  publisher={Elsevier}
}

@article{LW,
  title={Model Secrecy and Stress Tests},
  author={Leitner, Yaron and Williams, Basil},
  journal={Journal of Finance},
  volume={Forthcoming},
  year={2022}
}

@article{LN,
  title={Monotone Matching in Perfect and Imperfect Worlds},
  author={Legros, Patrick and Newman, Andrew F},
  journal={Review of Economic Studies},
  volume={69},
  number={4},
  pages={925--942},
  year={2002},
  publisher={Wiley-Blackwell}
}

@article{DK,
  title={The Persuasion Duality},
  author={Dworczak, Piotr and Kolotilin, Anton},
  journal={Working paper},
  year={2022}
}

@article{FH,
  title={Optimal Gerrymandering: Sometimes Pack, but Never Crack},
  author={Friedman, John N and Holden, Richard T},
  journal={American Economic Review},
  volume={98},
  number={1},
  pages={113--44},
  year={2008}
}

@article{AC,
  title={Bayesian Persuasion with Heterogeneous Priors},
  author={Alonso, Ricardo and C{\^a}mara, Odilon},
  journal={Journal of Economic Theory},
  volume={165},
  pages={672--706},
  year={2016},
  publisher={Elsevier}
}

@article{CS,
  title={Ordinal Aggregation Results via Karlin's Variation Diminishing Property},
  author={Choi, Michael and Smith, Lones},
  journal={Journal of Economic Theory},
  volume={168},
  pages={1--11},
  year={2017},
  publisher={Elsevier}
}

@article{GS2018,
	Author = {Guo, Yingni and Shmaya, Eran},
	Journal = {Econometrica},
	Title = {The Interval Structure of Optimal Disclosure},
	Volume={87},
  	Number={2},
  	Pages={653--675},
  	Year={2019}}

@article{Holmstrom79,
	Author = {Holmstr{\"o}m, Bengt},
	Journal = {Bell Journal of Economics},
	Pages = {74--91},
	Title = {Moral Hazard and Observability},
	Volume = {10},
	Number = {1},
	Year = {1979}}

@article{DM,
author = {Dworczak, Piotr and Martini, Giorgio},
title = {The Simple Economics of Optimal Persuasion},
journal = {Journal of Political Economy},
Volume = {127},
Number = {5},
Pages = {1993--2048},
year={2019}
}

@article{KMZL,
	Author = {Kolotilin, Anton and Mylovanov, Tymofiy and Zapechelnyuk, Andriy and Li, Ming},
	Journal = {Econometrica},
	Pages = {1949--1964},
	Title = {Persuasion of a Privately Informed Receiver},
	Volume = {85},
	Year = {2017}}

@article{BBM,
	Author = {Bergemann, Dirk and Brooks, Benjamin and Morris, Stephen},
	Journal = {American Economic Review},
	Pages = {921--957},
	Publisher = {American Economic Association},
	Title = {The Limits of Price Discrimination},
	Volume = {105},
	Year = {2015}}

@article{Quah2012,
	Author = {John Quah and Bruno Strulovici},
	Journal = {Econometrica},
	Pages = {2333--2348},
	Title = {Aggregating the Single Crossing Property},
	Volume = {80},
	Year = {2012}}

@article{GK-RS,
	Author = {Matthew Gentzkow and Emir Kamenica},
	Journal = {American Economic Review, Papers \& Proceedings},
	Pages = {597--601},
	Title = {A {Rothschild-Stiglitz} Approach to {Bayesian} Persuasion},
	Volume = {106},
	Year = 2016}

@article{rogerson,
  title={The First-Order Approach to Principal-Agent Problems},
  author={Rogerson, William P},
  journal={Econometrica},
  pages={1357--1367},
  year={1985},
  publisher={JSTOR}
}

@book{gale,
  title={The Theory of Linear Economic Models},
  author={Gale, David},
  year={1989},
  publisher={University of Chicago press}
}

@article{jewitt,
  title={Justifying the First-Order Approach to Principal-Agent Problems},
  author={Jewitt, Ian},
  journal={Econometrica},
  pages={1177--1190},
  year={1988},
  publisher={JSTOR}
}

@article{Kolotilin2017,
	Author = {Kolotilin, Anton},
	Journal = {Theoretical Economics},
	Pages = {607--636},
	Title = {Optimal Information Disclosure: A Linear Programming Approach},
	Volume = {13},
	Year = {2018}}

@article{ZZ,
  title={Information Disclosure in Contests: A Bayesian Persuasion Approach},
  author={Zhang, Jun and Zhou, Junjie},
  journal={Economic Journal},
  volume={126},
  number={597},
  pages={2197--2217},
  year={2016},
  publisher={Wiley Online Library}
}

@article{BJ,
  title={On a Problem of Optimal Transport under Marginal Martingale Constraints},
  author={Beiglb{\"o}ck, Mathias and Juillet, Nicolas},
  journal={Annals of Probability},
  volume={44},
  number={1},
  pages={42--106},
  year={2016},
  publisher={Institute of Mathematical Statistics}
}

@article{SY,
  title={Information Design in Concave Games},
  author={Smolin, Alex and Yamashita, Takuro},
  journal={Working paper},
  year={2022}
}

@article{GHT,
  title={A Stochastic Control Approach to No-Arbitrage Bounds Given marginals, with an Application to Lookback Options},
  author={Galichon, Alfred and Henry-Labordere, Pierre and Touzi, Nizar},
  journal={Annals of Applied Probability},
  volume={24},
  number={1},
  pages={312--336},
  year={2014},
  publisher={Institute of Mathematical Statistics}
}

@article{KG,
	Author = {Emir Kamenica and Matthew Gentzkow},
	Journal = {American Economic Review},
	Month = {October},
	Pages = {2590--2615},
	Title = {Bayesian Persuasion},
	Volume = {101},
	Year = 2011}

@article{RS,
	Author = {Luis Rayo and Ilya Segal},
	Journal = {Journal of Political Economy},
	Pages = {949--987},
	Title = {Optimal Information Disclosure},
	Volume = {118},
	Year = 2010}

@article{NP,
  title={Conjugate Information Disclosure in an Auction with Learning},
  author={Nikandrova, Arina and Pancs, Romans},
  journal={Journal of Economic Theory},
  volume={171},
  pages={174--212},
  year={2017},
  publisher={Elsevier}
}

@article{KW,
  title={The Economics of Partisan Gerrymandering},
  author={Kolotilin, Anton and Wolitzky, Alexander},
  journal={Working paper},
  year={2020},
}

\renewcommand{\thesection}{A}

\section{Characterization of Aggregate Quasi-Concavity}
\label{a:ad} 

We present two alternative conditions that are equivalent to strict aggregate quasi-concavity of $U$. Condition (2) is analogous to the ``signed-ratio monotonicity'' conditions for weak aggregate quasi-concavity in Theorem 1 of \citet{Quah2012} and Corollary 2 of \citet{CS}. We give a shorter proof based on the optimality of pairwise signals (see Section \ref{proof:ASC}). Condition (3) is novel. It corresponds to strict concavity of $U$ (i.e., $u_a(a,\theta)<0$), up to a normalizing factor $g(a)>0$. 
\begin{lemma}\label{l:ASC}
Let Assumption \ref{a:smooth} hold. The following statements are equivalent:
\begin{enumerate}
	\item Assumption \ref{a:qc} holds.
	\item For all $\theta$, $\theta'$, and $a$, we have
\begin{align}
u(a,\theta)=0 &\implies u_a(a,\theta)<0,\label{1}	\\
u(a,\theta)<0<u(a,\theta') &\implies u(a,\theta')u_a(a,\theta)-u(a,\theta)u_a(a,\theta')<0.\label{2}
\end{align}
	\item There exists a differentiable function $g(a)>0$ such that $\tilde u(a,\theta)=u(a,\theta)/g(a)$ satisfies $\tilde u_a(a,\theta)<0$ for all $(a,\theta)$.
\end{enumerate} 
\end{lemma}

\renewcommand{\thesection}{B}

\section{Proofs}

\subsection{Proof of Lemma \ref{l:dual}}

The proofs of primal attainment and strong duality (points 1 and 3 in the lemma) are standard and are deferred to the Online Appendix. Here we prove dual attainment (point 2).

For any nonempty, compact interval $I\subset \mathbb{R}$, let $B\left(
A,I\right) $ denote the set of bounded, measurable functions $q$ such that $%
q\left( A\right) \subset I$. Define%
\begin{equation*}
F_{I}=\left\{ p\in \mathbb{R}^{\Theta }:\exists q\in B\left(
A,I\right) ,p\left( \theta \right) =\sup_{a\in
A}\left\{ V\left( a,\theta \right) +q\left( a\right) u\left( a,\theta
\right) \right\} \text{ for all $\theta\in \Theta$}\right\} \text{,}
\end{equation*}
and consider the problem
\begin{equation*}
\inf_{p \in F_{I}}\int_{\Theta }p(\theta )\df\phi (\theta ). \tag{D'} 
\end{equation*}
(D') is a reformulated version of (D) that involves only the function $p$. 
Denote the value of (D) by $\mathbf{D}$, and denote the value of (D') (which depends on the interval $I$) by $\mathbf{D_{I}'}$. 
We first show that there exists a solution $p\in F_I$ to (D') and that $p$, together with any measurable selection $q$ from $Q$, is feasible for (D) (Lemma \ref{l:select_q}).  
Finally, we show that for a sufficiently large interval $I=[-C,C]$, $\mathbf{D}=\mathbf{D_{I}'}$ (Lemma \ref{l:D=D'}), so $(p,q)$ solve (D).

For the moment, let $I=[-C,C]$ for an arbitrary choice of $C>0$. The existence of a solution to (D') relies on the following lemma.
\begin{lemma}\label{l:D'attain}
The family of functions $F_{I}$ is uniformly bounded and equicontinuous. Thus, there exists a convergent sequence $p_n \to p$ such that $p_n \in F_I$ for all $n$, $p \in C(\Theta)$, and $\int{pd\phi}=\mathbf{D_{I}'}$.
\end{lemma}

\begin{proof}
For each $p\in F_{I}$, there exists $q\in
B\left( A,I\right)$ such that $p\left( \theta \right) =\sup_{a\in A}\{ V\left( a,\theta \right)
+q\left( a\right) u\left( a,\theta \right) \}$ for all $\theta \in \Theta$,
and thus
\[
\sup_{\theta \in \Theta} |p(\theta)|\leq \sup_{(a,\theta) \in A\times \Theta }\left\vert  V\left( a,\theta
\right) +q\left( a\right) u\left( a,\theta \right) \right\vert \leq
\sup_{\left( a,\theta ,r\right) \in A\times \Theta \times I}\left\vert
V\left( a,\theta \right) +ru\left( a,\theta \right) \right\vert\text{,}
\]
This upper bound is finite by compactness of $A\times \Theta \times I$ and continuity of $V
$ and $u$, so the family of functions $F_{I}$ is uniformly
bounded. 

Next, since $V $ and $u $ are continuous on the compact set $A\times\Theta $, they are uniformly continuous on $A\times\Theta $. This implies that there exists an increasing, continuous function $w:\mathbb{R}%
_{+}\rightarrow \mathbb{R}_{+}$ (known as the modulus of continuity)
such that $w\left( 0\right) =0$ and, for
every $\theta ,\theta ^{\prime }\in \Theta $ and $a\in A$, we have%
\begin{equation*}\label{e:w(d)}
\begin{aligned}
&\left\vert V\left( a,\theta \right) +q(a)u\left( a,\theta \right) -V\left(
a,\theta ^{\prime }\right) -q(a)u\left( a,\theta ^{\prime }\right) \right\vert \\
\leq &\left\vert V\left( a,\theta \right) -V\left( a,\theta ^{\prime
}\right) \right\vert +\left\vert q(a)\right\vert \left\vert u\left( a,\theta
\right) -u\left( a,\theta ^{\prime }\right) \right\vert  \\
\leq &\left\vert V\left( a,\theta \right) -V\left( a,\theta ^{\prime
}\right) \right\vert + C \left\vert u\left( a,\theta
\right) -u\left( a,\theta ^{\prime }\right) \right\vert  \\
\leq &w\left( d\left( \theta ,\theta ^{\prime }\right) \right),
\end{aligned}
\end{equation*}%
where $d(\theta,\theta')$ denotes the distance between $\theta,\theta'\in\Theta$. We claim that for all $p\in F_{I}$ and $\theta ,\theta
^{\prime }\in \Theta $, we have $|p(\theta)-p(\theta')|\leq w(d(\theta,\theta'))$. Indeed, for each $a\in A$, 
\begin{align*}
	p(\theta)=& \sup_{\tilde a\in A}\left\{ V\left( \tilde a,\theta \right) +q\left( \tilde a\right) u\left( \tilde a,\theta \right)\right\}\\
	\geq& V\left( a,\theta \right) +q\left( a\right) u\left( a,\theta \right) \\
	\geq& V\left( a,\theta' \right) +q\left( a\right) u\left( a,\theta' \right) - w\left( d\left( \theta ,\theta ^{\prime }\right) \right).
\end{align*}
Taking the supremum over $a\in A$ gives $p(\theta)\geq p(\theta')-w(d(\theta,\theta'))$, and switching the roles of $\theta$ and $\theta'$ gives $|p(\theta)-p(\theta')|\leq w(d(\theta,\theta'))$. Consequently, the family of functions $F_{I}$ is equicontinuous.

Now consider a minimizing sequence $p_n\in F_I$ such that
$\int p_n(\theta)\df \phi(\theta) \rightarrow \mathbf{D_{I}'}$.
Since $\Theta$ is compact, and $F_I$ is uniformly bounded and equicontinuous, Arzel\`{a}%
-Ascoli's theorem implies that there exists a subsequence $p_{n_k}$ uniformly converging to some function $p\in C(\Theta)$, and thus $\int p_{n_k}(\theta)\df \phi(\theta) \rightarrow  \int p(\theta)\df \phi(\theta)=\mathbf{D_{I}'}$.
\end{proof}

Now fix $p\in C(\Theta)$ as in Lemma \ref{l:D'attain}. To show that $p\in F_I$, recall the correspondence%
\begin{equation}
\label{e:corresp_Q}
Q\left( a\right) =\left\{ r\in I:p\left( \theta
\right) \geq V\left( a,\theta \right) +ru\left( a,\theta \right)\text{ for all $\theta \in \Theta$} \right\},
\quad \text{for all $a\in A$.}
\end{equation}%
We first derive some properties of this correspondence, which will also
be used in the subsequent analysis.\footnote{In the current proof, the correspondence $Q$ is defined in reference to the price function $p$ defined in Lemma \ref{l:D'attain}. In the text, $Q$ is defined in reference to an optimal price function. We will see that $p$ is indeed optimal, so the definitions are equivalent.}

\begin{lemma}\label{Q:uhc}
The correspondence $Q$ is nonempty, convex and compact valued, and upper
hemicontinuous, and hence admits a measurable selection $q$.
\end{lemma}

\begin{proof}
By Lemma \ref{l:D'attain}, there exists a sequence $p_{n}\in F_{I}$, such that $%
p_{n}\rightarrow p$ uniformly. For every $n\in \mathbb{N}$, define%
\[
Q_{n}\left( a\right) =\left\{ r\in I:p_{n}\left(
\theta \right) \geq V\left( a,\theta \right) +ru\left( a,\theta \right)
\text{ for all $\theta \in \Theta$}\right\} \quad \text{for all $a\in A$.}
\]%
For every $a\in A$ and $n\in \mathbb{N}$, we have $Q_{n}\left(
a\right) \neq \emptyset $ since $p_{n}\in F_{I}$. Fix $a\in A$, and for
every $n\in \mathbb{N}$ fix $r_{n}\in Q_{n}\left( a\right) \subset I$.
Since $I$ is compact, there exists a convergent
subsequence $r_{n_{k}}\rightarrow r$ with $r\in I$. For all $%
k\in \mathbb{N}$, we have $p_{n_k}\left( \theta \right) \geq V\left( a,\theta \right) +r_{n_k}u\left(
a,\theta \right)$ for all $\theta \in \Theta$, which implies that $p\left( \theta \right) \geq V\left( a,\theta \right) +ru\left( a,\theta
\right)$ for all $\theta \in \Theta$. This shows that $r\in Q\left( a\right) $. Since $a$ was arbitrary, it follows that $Q$ is nonempty valued.

Next, for all $a\in A$, $Q(a)$ is closed because $V$ and $u$ are continuous, and $Q(a)$ is convex because it is defined by a linear inequality. Now consider a sequence $ \left(
a_{n},r_{n}\right)$ in the graph of $Q$ such that $\left(
a_{n},r_{n}\right) \rightarrow \left( a,r\right) $. For every $n\in \mathbb{N%
}$, we have $p\left( \theta \right) \geq V\left( a_{n},\theta \right) +r_{n}u\left(
a_{n},\theta \right)$ for all $\theta \in \Theta$. By continuity of $V$ and $u$, this implies that $p\left( \theta \right) \geq V\left( a,\theta \right) +ru\left( a,\theta
\right)$ for all $\theta \in \Theta$. This shows that $\left( a,r\right)$ is in the graph of $Q$. By the closed-graph
theorem, the correspondence $Q$ is upper hemicontinuous. Finally, by Theorem
18.20 in \citet{aliprantis2006}, $Q$ admits a
measurable selection $q$.
\end{proof}

We next show that $(p,q)$ is feasible for (D), for any measurable selection $q$ from $Q$. Consider the problem
\begin{equation*}
\inf \left\{ \int p(\theta) \df\phi (\theta) : p\in C\left( \Theta \right) ,\exists
q\in B\left( A,I\right) \text{ such that } ( p,q) \text{ satisfy (D1)}\right\} .\tag{D''} 
\end{equation*}
Denote the the value of (D'') by $\mathbf{D_{I}''}$

\begin{lemma}\label{l:select_q}
\label{lm:q_invar}For every measurable selection $q$ from $Q$, we have%
\begin{equation*}
p\left( \theta \right) =\sup_{a\in A}\left\{ V\left( a,\theta \right)
+q\left( a\right) u\left( a,\theta \right) \right\}, \quad \text{for all $\theta
\in \Theta$,} 
\end{equation*}
and hence $p \in F_{I}$, and $(p,q)$ satisfy (D1). Therefore, $\mathbf{D_{I}'}=\mathbf{D_{I}''}$.
\end{lemma}

\begin{proof}
Fix a measurable selection $q$ from $Q$, and let $\hat{p}(\theta):=\sup_{a\in A}\{ V\left( a,\theta \right) +q\left( a\right) u\left( a,\theta
\right) \}$ for all $\theta \in \Theta$. Note that $\hat{p} \in F_{I} \subset C(\Theta)$, and that $p(\theta)\geq \hat{p}(\theta)$ for all $\theta \in \Theta$ by construction of $Q$. Conversely, if $\hat{p}(\theta)< p(\theta)$ for some $\theta \in \Theta$, then $\int_{\Theta }\hat{p}\left( \theta \right) \df \phi \left( \theta \right)
<\int_{\Theta }p\left( \theta \right) \df \phi \left( \theta \right)$ (by continuity of $p$ and $\hat{p}$, together with full support of $\phi$), which contradicts the definition of $p$. Hence, $p=\hat{p}$, establishing the first part of the lemma. Next, since $p$ is continuous and $(p,q)$ satisfy (D1) for every selection $q$ from $Q$, $p$ is feasible for (D''). Moreover, for any $\tilde p\in C\left( \Theta \right) $ that satisfies (D1) for some $q\in
B\left( A,I\right) $, the function $\hat p$ defined above satisfies $\hat{p}\left( \theta \right) \leq \tilde p\left(
\theta \right) $ for all $\theta \in \Theta $, so we have $\int p \df \phi \leq \int \hat{p}\df \phi \leq \int
\tilde p \df \phi$. Hence, $p$ solves (D''), and $\mathbf{D_{I}''}= \mathbf{D_{I}'}$.
\end{proof}

The following lemma implies that for a sufficiently large interval $I=\left[ -C,C\right] $, we have $\mathbf{D}=\mathbf{D_{I}''}$, so that the pair $(p,q)$ constructed in Lemma \ref{l:select_q} solve (D). This proves dual attainment.

\begin{lemma}
\label{l:D=D'}There exists $C>0$ such that $\mathbf{D}=\mathbf{D_{I}''}$, where $I=\left[ -C,C\right] $.
\end{lemma}

\begin{proof}
It is enough to find $C>0$ such that the additional constraint $q(a)\in \lbrack -C,C]$ for
all $a\in A$ is non-binding in (D).

Define 
\begin{equation*}
\widetilde{q}(a,\theta )=%
\begin{cases}
\frac{v(a,\theta )}{-u_{a}(a,\theta )}, & u(a,\theta )=0, \\%
\frac{V(a^{\star }(\delta _{\theta }),\theta )-V(a,\theta )}{u(a,\theta )}, & 
u(a,\theta )\neq 0.
\end{cases}%
\end{equation*}%
Recall that Assumption \ref{a:qc} requires
that $u_{a}(a,\theta )<0$ when $u(a,\theta )=0$; so $\tilde{q}(a,\theta )$
is well-defined. Since $a^{\star }(\delta _{\theta })$ is a unique maximizer
of a continuous function $U(a,\theta )$, it is continuous in $\theta $ by Berge's theorem. 

We now prove that $\widetilde{q}$ is continuous at
each $\left( a ,\theta \right) \in A\times \Theta $. First, $\widetilde q$ is continuous at each $(a,\theta)$ such that $u(a,\theta)\neq 0$, because $V$, $u$, and $a^\star$ are continuous.
Next, consider $(a,\theta)$ such that $u(a,\theta)=0$, or equivalently $a=a^\star(\delta_\theta)$. For each $(a',\theta')\in A\times \Theta$, there exists $\hat a$ between $a^\star(\delta_{\theta'})$ and $a'$ such that
\[
[V(a^\star(\delta_{\theta'}),\theta')-V(a',\theta')]u_a(\hat a,\theta')=-v(\hat a,\theta')u(a',\theta'),
\]
by the mean value theorem applied to the function  
\[
[V(a^\star(\delta_{\theta'}),\theta')-V(\tilde a,\theta')]u(a',\theta')-[V(a^\star(\delta_{\theta'}),\theta')-V(a',\theta')]u(\tilde a,\theta'),
\]
where the argument $\tilde a$ is between $a^\star(\delta_{\theta'})$ and $a'$. Thus,
\[
\widetilde q(a',\theta')-\widetilde q(a,\theta)=\frac{v(\hat a,\theta')}{-u_a(\hat a,\theta')}-\frac{v(a,\theta)}{-u_a(a,\theta)}.
\]
If $(a',\theta')\rightarrow (a,\theta)$ then $(\hat a,\theta')\rightarrow (a,\theta)$, because $a^\star(\delta_\theta)$ is continuous in $\theta$. Hence, $\widetilde q(a',\theta')\rightarrow \widetilde q(a,\theta)$, because $v$ and $u_a$ are continuous. This shows that $\widetilde{q}$ is continuous on $A\times \Theta$. 

Next, define $\underline{C}=\min_{(a,\theta) \in A\times \Theta }\widetilde{q}%
(a,\theta )-1$ and $\overline{C}=\max_{(a,\theta) \in
A\times \Theta }\widetilde{q}(a,\theta )+1$,
where $\underline{C}$ and $\overline{C}$ are finite because $\widetilde{q}$
is continuous on the compact set $A\times \Theta $. To see why the
constraint $q(a)\leq \overline{C}$ is non-binding, notice that decreasing
$q(a)$ weakly tightens (D1) for $\theta $ such that $u(a,\theta )<0$, and
weakly relaxes (D1) for $\theta $ such that $u(a,\theta )\geq 0$. If $q(a)>%
\overline{C}$ and $u(a,\theta )<0$, then $V(a,\theta )+q(a)u(a,\theta
)<V(a^{\star }(\delta _{\theta }),\theta )$, so decreasing $q(a)$ to $\overline{C}$ does not strictly tighten (D1),
because $p(\theta )\geq V(a^{\star }(\delta _{\theta }),\theta )$ by feasibility. Thus, since the dual objective function does not depend on $q(a)$, adding the constraint $q(a)\leq \overline{C}$ does not affect the value of (D). Similarly, increasing $q(a)$ to $\underline{C}$ does not strictly tighten (D1) for $\theta $ such that $u(a,\theta )>0$, and weakly
relaxes (D1) for $\theta $ such that $u(a,\theta )\leq 0$; so we can add the non-binding
constraint $q(a)\geq \underline{C}$. 

In sum, adding the constraint $q\left( A\right) \subset I=\left[ -C,C\right] $ where  $C=\max \{|\underline{C}|,|\overline{C}|\}$
does not alter the value of (D), so $\mathbf{D}=\mathbf{D_{I}''}$.
\end{proof}

\subsection{Proof of Theorem \ref{t:contact}}

Let
\[
\mathring q(a)=\frac{\min Q(a) +\max Q(a)}{2}.
\]
Define the set of \emph{$a$-contact points of type 1} as
\[
\Psi_1=\{a\in A:\ \theta^\star(a)\in \Theta\text{ and }p(\theta^\star(a))=V(a,\theta^\star(a))\},
\]
and the set of \emph{$a$-contact points of type 2} as
\[
\Psi_2=\{a\in A\setminus \Psi_1: \exists \theta\in \Theta: p(\theta)=V(a,\theta)+\mathring q(a)u(a,\theta) \}.
\]
Note that
\[
q(a)=
\begin{cases}
\frac{v(a,\theta^\star(a))}{-u_a(a,\theta^\star(a))}, &a \in \Psi_1,\\
\mathring q(a), &\text{otherwise}.
\end{cases}
\]
Part (1) of the theorem follows from Lemmas \ref{l:12}--\ref{l:suppG}, and part (2) of the theorem follows from Lemmas \ref{l:G*} and \ref{l:pia}. 

\begin{lemma}\label{l:12}
$\Gamma_a$ is non-empty iff $a\in \Psi_1\cup\Psi_2$. That is, $\Psi_1\cup\Psi_2=A_\Gamma$.
\end{lemma}
\begin{proof}
Clearly, $\theta^\star(a)\in \Gamma_a$ if $a\in \Psi_1$. By the definition of $\Psi_2$, $\Gamma_a$ is non-empty if $a\in \Psi_2$, and $\Gamma_a$ is empty if $a\notin \Psi_1\cup \Psi_2$.
\end{proof}
\begin{lemma}\label{l:inQ}
(p,q) solves (D).
\end{lemma}
\begin{proof}
Note that $q\in B(A)$, as follows from the proof of Lemma \ref{l:D=D'} (measurability of $q$ follows from continuity of $p$, $v$, $u_a$, and $\theta^\star$). Thus, by Lemma \ref{l:select_q}, it suffices to show that $q(a_1)\in Q(a_1)$ for each $a_1\in \Psi_1$: that is,
\[p(\theta) \geq V(a_1,\theta) + q(a_1)u(a_1,\theta)\quad \text{ for all $a_1\in \Psi_1$ and $\theta \in \Theta$.}\]

 Fix any $a_1\in \Psi_1$ and $\theta\in \Theta$, and let $\theta_1=\theta^\star (a_1)$. For any $\varepsilon \in (0,1)$, define $a_{\varepsilon}\in A$ as a unique solution to
$
(1-\varepsilon)u(a_\varepsilon ,\theta_1)+\varepsilon u(a_\varepsilon,\theta )=0.
$
By the implicit function theorem,
\[
\lim_{\varepsilon \downarrow 0}\frac{a_\varepsilon -a_1}{\varepsilon} = \frac{u(a_1,\theta)}{-u_a(a_1,\theta_1)}.
\]
By (D1), we have
\begin{align*}
	V(a_1,\theta_1)\geq V(a_\varepsilon,\theta_1)+\mathring q(a_\varepsilon) u(a_\varepsilon,\theta_1) \quad \text{and} \quad
	p(\theta)\geq V(a_\varepsilon,\theta)+\mathring q(a_\varepsilon) u(a_\varepsilon,\theta).
\end{align*}
Adding the first inequality multiplied by $1-\varepsilon$ and the second inequality multiplied by $\varepsilon$, and taking into account the definition of $a_\varepsilon$, we get
\[
p(\theta) \geq V(a_1,\theta) +\frac{(1-\varepsilon) [V(a_\varepsilon,\theta_1)-V(a_1,\theta_1)]+\varepsilon [V(a_\varepsilon,\theta)-V(a_1,\theta)]}{\varepsilon}.
\]
Taking the limit $\varepsilon \rightarrow 0$ gives
\[
p(\theta)\geq V(a_1,\theta)+\frac{v(a_1,\theta_1)}{-u_a(a_1,\theta_1)}u(a_1,\theta)=V(a_1,\theta)+ q(a_1)u(a_1,\theta).\qedhere
\]
\end{proof}

\begin{lemma}\label{l:-+}
For each $a\in \Psi_1$, we have $\inf \Gamma_a\leq\theta^\star(a)\leq\sup \Gamma_a$. For each $a\in \Psi_2$, we have $\theta^\star(a)\notin \Gamma_a$, $\inf \Gamma_a<\theta^\star (a)<\sup \Gamma_a$, and $\mathring q(a)=\min Q(a)=\max Q(a)$.
\end{lemma}
\begin{proof}
We have $\theta^\star(a)\in\Gamma_a\subset [\inf \Gamma_a,\sup \Gamma_a]$ for each $a\in \Psi_1$, by the definition of $\Psi_1$.

Fix  $a\in \Psi_2$. By the definition of $\Psi_2$, we have $a\notin \Psi_1$, so 
\[p(\theta^\star (a))>V(a,\theta^\star(a))=V(a,\theta^\star(a))+\mathring q(a) u(a,\theta^\star(a)), \] 
showing that $\theta^\star (a)\notin \Gamma_a$. By the definition of $\Psi_2$, $\Gamma_a$ is non-empty, so it contains some $\theta\neq \theta^\star(a)$. Suppose for concreteness that $\theta>\theta^\star(a)$, so we write $\theta=\theta_+$ (the case $\theta<\theta^\star(a)$ is analogous and omitted).
By the definition of $\mathring q(a)$ and $Q(a)$,
\[
\mathring q(a)\leq \max Q(a) \leq \frac{p(\theta_+)-V(a,\theta_+)}{u(a,\theta_+)},
\]
and, by the definition of $\Gamma_a$,
\[
\mathring q(a)=\frac{p(\theta_+)-V(a,\theta_+)}{u(a,\theta_+)}.
\]
Hence,
\[
\mathring q(a)= \max Q(a) = \frac{p(\theta_+)-V(a,\theta_+)}{u(a,\theta_+)}.
\]
Then, 
\[
\mathring q(a)=\min Q(a) =\sup_{\tilde \theta<\theta^\star(a)} \frac{V(a,\tilde \theta)-p(\tilde \theta)}{-u(a,\tilde \theta)},
\]
where the first equality is by the definition of $\mathring q(a)$ and $\mathring q(a)=\max Q(a)$, and the second equality is by the definition of $Q(a)$. (Inspecting the definition gives $\min Q(a) =\sup_{\tilde \theta : u(a,\tilde \theta)<0}(V(a,\tilde \theta)-p(\tilde \theta))/(-u(a,\tilde \theta))$, and $u(a,\tilde \theta)<0$ iff $\tilde \theta < \theta^\star(a)$.)

Since $p$, $V$, and $u$ are continuous and since $p(\theta^\star(a))>V(a,\theta^\star(a))$, the supremum is attained at some $\theta_-<\theta^\star(a)$. Thus $\theta_-\in \Gamma_a$, by the definition of $\Gamma_a$. The lemma follows since $\inf \Gamma_a\leq \theta_-<\theta^\star(a)<\theta_+\leq \sup \Gamma_a$.
\end{proof}
\begin{lemma}\label{l:FOC}
For each $a\in \Psi_1\cup \Psi_2$ such that $\inf \Gamma_a<\theta^\star(a)<\sup \Gamma_a$, the function $q$ has a derivative $ q'(a)$, and \eqref{e:FOC} holds for all $\theta\in \Gamma_a$.
\end{lemma}
\begin{proof} Fix $a\in \Psi_1\cup \Psi_2$ such that there exist $\theta_-,\theta_+\in \Gamma_a$ with $\theta_-<\theta^\star(a)<\theta_+$. By (D1) and the definition of $\Gamma$, for every $\tilde a \in A$,  we have
\[
V(a,\theta_- )+ q(a)u(a,{\theta_- }) \geq  V(\tilde a,{\theta_- })+{q}(\tilde a)u(\tilde a,{\theta_- }).
\]
Therefore, for every $\tilde a > a$, we have
\[
\frac{{q}(\tilde a)-{q}(a )}{\tilde a -a }\geq 
\frac{1}{-u(\tilde a,\theta_-)}\left[ \frac{V(\tilde a,%
\theta_-)-V(a ,\theta_-)}{\tilde a-a }%
+{q}(a )\frac{u(\tilde a,\theta_-)-u(a ,%
\theta_-)}{\tilde a-a}\right] .
\]
Since $V$ and $u$ have continuous partial derivatives in $a $, we have 
\[
\underline{{q}}_+^{\prime }(a ):=\liminf_{\tilde a
\downarrow a }\frac{{q}(\tilde a)-{q}(a )}{%
\tilde a-a }\geq {C_-},
\]%
where 
\[
{C_-}=-\frac{1}{u(a ,{\theta_-})}[v(a ,%
{\theta_- })+{q}(a )u_a(a ,{%
\theta_- })].
\]
Applying a similar argument for $\theta=\theta_-$ and $\tilde a<a $, we get 
\[
\overline{q}_-^{\prime }(a ):=\limsup_{\tilde{a}\uparrow a }\frac{{q}(\tilde a)-{q}(a )}{\tilde a-a }\leq {C_-}.
\]
Similarly, considering $\theta =\theta_+$ with $\tilde a>a $
and $\tilde a<a $, we get 
\begin{align*}
\overline{q}_+^{\prime }(a ): =\limsup_{\tilde a\downarrow a}\frac{{q}(\tilde a)-{q}(a )}{\tilde a-a }\leq {C_+} \quad \text{and} \quad
\underline{{q}}_-^{\prime }(a ): =\liminf_{\tilde a\uparrow a }\frac{{q}(\tilde a)-{q}(a )}{\tilde a-a }\geq {C_+},
\end{align*}
where 
\[
{C_+}=-\frac{1}{u(a ,{\theta_+ })}[v(a ,{\theta_+  })+{q}(a )u_{a }(a ,{\theta_+  })].
\]
In sum, we have
\[
{C_-}\leq \underline{{q}}_+^{\prime }(a )\leq \overline{{q}}^{\prime }_+(a)\leq {C_+}  \quad \text{and} \quad {C_+}\leq \underline{{q}}_-^{\prime }(a )\leq \overline{{q}}_-^{\prime }(a )\leq {C_-}.
\]
We see that $ C_- =C_+$ and all four Dini derivatives of $q$ at $a$ coincide, so ${q}$ has a derivative ${q}'(a)$
at $a$ that satisfies $q^{\prime }(a )={C_-}={C_+}$. 

Since $\theta_-,\theta_+\in \Gamma_a$ are arbitrary, the lemma follows for $\theta\in \Gamma_a$ with $\theta\neq \theta^\star(a)$. For $\theta\in \Gamma_a$ with $\theta=\theta^\star(a)$, we have $a\in \Psi_1$, and the lemma follows by the definition of $q(a)$.
\end{proof}
\begin{lemma}
The sets $\Gamma$ and $\Psi_1\cup\Psi_2$ are compact.
\end{lemma}
\begin{proof}
To show that $\Gamma$ is compact, we need to show that if $(a_n,\theta_n)\rightarrow (a,\theta)$ with $(a_n,\theta_n)\in \Gamma$, then $(a,\theta)\in \Gamma$. By Lemma \ref{l:12}, $\Gamma_a$ is non-empty iff $a\in \Psi_1\cup \Psi_2$. Thus, $a_n\in \Psi_1\cup \Psi_2$. There are three cases to consider, up to taking a suitable subsequence.

(1) $a_n\in \Psi_1$ for all $n$. Since $p$, $V$, and $\theta^\star$ are continuous, the set $\Psi_1$ is closed. Thus, $a\in \Psi_1$. Since $(a_n,\theta_n)\in \Gamma$, we have $p(\theta_n)=V(a_n,\theta_n)+q(a_n) u(a_n,\theta_n)$. Since $v$, $u_a$, and $\theta^\star$ are continuous, $q$ is continuous on $\Psi_1$. Since $p$, $V$, and $u$ are also continuous, passing to the limit we have $p(\theta)=V(a,\theta)+ q(a)u(a,\theta)$, so $(a,\theta)\in \Gamma$.

(2) $a_n\in \Psi_2$ for all $n$, and $a\notin \Psi_1$. Since $a_n\in \Psi_2$, we have $a_n\notin \Psi_1$, and hence, by Lemma \ref{l:-+}, $\theta^\star(a_n)\notin \Gamma_{a_n}$. Taking another subsequence if necessary, we can assume that  $\theta_n- \theta^\star(a_n)$ has the same sign for all $n$. Suppose for concreteness that $\theta_n>\theta^\star(a_n)$ (the case $\theta_n<\theta^\star(a_n)$ is analogous). 

Since $a_n\in \Psi_2$, there exists $\tilde \theta_n \in \Gamma_{a_n}$ with $\tilde \theta_n<\theta^\star(a_n)$, by Lemma \ref{l:-+}. Taking yet another subsequence, we can assume that 
\[
\tilde \theta_n \rightarrow \tilde\theta \leq \theta^\star(a)\quad \text{and}\quad \mathring q(a_n) \rightarrow r\in Q(a).
\]
(Such a subsequence must exist because $A\times \Theta$ is compact, $\theta^\star$ is continuous, and $Q$ is upper hemi-continuous.) Moreover, by continuity of $p$, $V$, and $u$, we have
\begin{align*}
p(\tilde\theta)=V(a,\tilde\theta)+ru(a,\tilde\theta) \quad \text{and} \quad p(\theta)=V(a,\theta)+ru(a,\theta).	
\end{align*}
Since $a \notin \Psi_1$, we have $\tilde\theta,\theta\neq\theta^\star(a)$. Thus, $\tilde \theta<\theta^\star(a)<\theta$. Next, $\theta>\theta^\star(a)$ implies that $r=\max Q(a)$; otherwise, $p(\theta)<V(a,\theta)+\max Q(a) u(a,\theta)$, contradicting the definition of $Q(a)$. Similarly, $\tilde\theta<\theta^\star(a)$ implies that $r=\min Q(a)$. Hence, 
\[r=\min Q(a)=\max Q(a)=\mathring q(a).\]
We have $p(\theta)=V(a,\theta)+\mathring q(a)u(a,\theta)$. Since $a \notin \Psi_1$, this says that $(a,\theta)\in \Gamma$.

(3) $a_n\in \Psi_2$ for all $n$, and $a\in \Psi_1$. If $\theta=\theta^\star(a)$, then $(a,\theta)\in \Gamma$ because $p$, $V$, and $u$ are continuous, and $q$ is bounded. So suppose for concreteness that $\theta>\theta^\star(a)$ (the case $\theta<\theta^\star(a)$ is analogous). Taking another subsequence if necessary, we can assume that $\theta_n>\theta^\star(a_n)$ for all $n$. By Lemma \ref{l:-+}, for each $n$ there exist $\tilde \theta_n\in \Gamma_{a_n}$ with $\tilde \theta_n<\theta^\star(a_n)$. Taking a subsequence again, we can assume that
\[
\tilde \theta_n \rightarrow \tilde\theta \leq \theta^\star(a)\quad \text{and}\quad \mathring q(a_n) \rightarrow r\in Q(a).
\]
Passing to the limit, we get
\begin{align*}
p(\tilde\theta)&=V(a,\tilde\theta)+ru(a,\tilde\theta) \quad \text{and} \quad p(\theta)=V(a,\theta)+ru(a,\theta).
\end{align*}
If $\tilde \theta<\theta^\star(a)$, then as in the previous case $r=\min Q(a)=\max Q(a)$. Since $q(a)\in Q(a)$ by Lemma \ref{l:inQ}, this yields $r=q(a)$, and hence $(a,\theta)\in \Gamma$.

Finally, if $\tilde \theta=\theta^\star (a)$, then by Lemma \ref{l:FOC} and $a_n \in \Psi_2$ we have
\begin{align*}
v(a_n,\theta_n)+\mathring q(a_n)u_a(a_n,\theta_n)+\mathring q'(a_n)u(a_n,\theta_n)&=0,\\
v(a_n,\tilde\theta_n)+\mathring q(a_n)u_a(a_n,\tilde\theta_n)+\mathring q'(a_n)u(a_n,\tilde\theta_n)&=0.
\end{align*}
Thus, 
\[
\mathring q(a_n)=\frac{v(a_n,\tilde \theta_n)u(a_n,\theta_n)-v(a_n,\theta_n)u(a_n,\tilde \theta_n)}{-u_a(a_n,\tilde \theta_n)u(a_n,\theta_n)+u_a(a_n,\theta_n)u(a_n,\tilde\theta_n)}.
\]
As $\mathring q(a_n) \rightarrow r$ and $a \in \Psi_1$, passing to the limit we have
\[
r=\frac{v(a,\theta^\star(a))}{-u_a(a,\theta^\star(a))}= q(a).
\]
This shows that $(a,\theta)\in \Gamma$.

We have shown that $\Gamma$ is compact. By Lemma \ref{l:12}, $\Psi_1\cup \Psi_2 = A_\Gamma$, and thus is compact as the projection of a compact set.
\end{proof}

\begin{lemma}\label{l:suppG}
An implementable outcome $\pi$ is optimal iff $\supp (\pi) \subset \Gamma$.
\end{lemma}
\begin{proof}
For any implementable outcome $\pi$, we have, by (P1), (D1), and (P2),
\begin{align*}
\int_{\Theta} p(\theta) \df \phi(\theta) &= \int_{A\times \Theta} p(\theta)\df \pi(a,\theta)\\
 &\geq \int_{A\times \Theta} (V(a,\theta)+ q(a)u(a,\theta))\df \pi(a,\theta)\\
 &=\int_{A\times \Theta } V(a,\theta) \df \pi 	(a,\theta).
\end{align*}
By Lemma \ref{l:dual}, $\pi$ is optimal iff the inequality holds with equality, or equivalently $\pi(\Gamma)=1$. In turn, since $\Gamma$ is compact, $\pi(\Gamma)=1$ iff $\supp(\pi)\subset \Gamma$, because $\supp(\pi)$ is defined as the smallest compact set of measure one.
\end{proof}
\begin{lemma}\label{l:G*}
The set $\Gamma^\star$ is Borel, and \eqref{e:FOC} holds for all $(a,\theta)\in \Gamma^\star$.
\end{lemma}
\begin{proof}
Since $\Gamma$ is compact, $\min \Gamma_a$ and $\max \Gamma_a$ are measurable functions from $A_\Gamma$ to $\Theta$ that satisfy $\min \Gamma_a,\max \Gamma_a\in \Gamma_a$ for all $a\in A_\Gamma$. Since $\theta^\star(a)$ is a continuous function that satisfies $\min \Gamma_a\leq \theta^\star(a)\leq \max \Gamma_a$ for all $a\in A_\Gamma$, it follows that $\Gamma^\star$ is a Borel subset of $\Gamma$. Finally, if $\Gamma_a^\star=\{\theta^\star(a)\}$ then \eqref{e:FOC} holds at $(a,\theta^\star (a))$ by the definition of $q(a)$; otherwise, $\min \Gamma_a <\theta^\star(a) < \max \Gamma_a$, so \eqref{e:FOC} holds at $(a,\theta)$ for all $\theta \in \Gamma_a = \Gamma^\star_a$, by Lemma \ref{l:FOC}.
\end{proof}
\begin{lemma}\label{l:pia}
An implementable outcome $\pi$ satisfies $\supp (\pi)\subset \Gamma$ iff there exists a conditional probability $\pi_a$ such that $\supp (\pi_a)\subset \Gamma_a^\star$ and $\int u(a,\theta) \df \pi_a(\theta)=0$ for all $a\in \supp (\alpha_\pi)$.	
\end{lemma}
\begin{proof}
If an outcome $\pi$ admits such a conditional probability then $\pi(\Gamma^\star)=\pi(\Gamma)=1$, so $\supp(\pi)\subset \Gamma$.
Now fix an implementable outcome $\pi$ such that $\supp(\pi)\subset \Gamma$. Recall that $\alpha_\pi$ is the $a$-marginal distribution. Let $\pi_a$ be any version of the conditional probability. By (P2) and $\supp(\pi) \subset \Gamma$, there exists a Borel set $S_\pi\subset \supp(\alpha_\pi)$ with $\alpha_\pi(S_\pi)=1$ such that
\[
\supp(\pi_a)\subset \Gamma_a \quad \text{and}\quad \int_\Theta u(a,\theta) \df\pi_a(\theta)=0\quad \text{for all $a\in S_\pi$}.
\]
Hence, for each $a\in S_\pi$, \[\text{either} \quad
\min \Gamma_a<\theta^\star(a) <\max \Gamma_a\quad \text{or} \quad \min \Gamma_a=\theta^\star(a) =\max \Gamma_a.
\]
By definition, $\Gamma^\star_a$ coincides with $\Gamma_a$ for such $a$, so
\[
\supp(\pi_a)\subset \Gamma_a = \Gamma_a^\star\quad  \text{for all $a\in S_\pi$}.
\] 
Finally, for all $a\in A_\Gamma\setminus S_\pi$, we can redefine $\pi_a$ as follows:
\begin{gather*}
\pi_a = \rho _a \delta_{\min \Gamma_a^\star}+(1-\rho  _a)\delta_{\max \Gamma_a^\star},\\
\text{where}\quad \rho _a=\frac{u(a,\max \Gamma_a^\star)}{u(a,\max \Gamma_a^\star)-u(a,\min \Gamma_a^\star)}\1 \{\min \Gamma_a^\star<\max \Gamma_a^\star\}.
\end{gather*}
With this definition, $\pi_a$ automatically satisfies the conditions of the lemma for all $a\in A_\Gamma\setminus S_\pi$. Lastly, since $\alpha_\pi(S_\pi)=1$, the redefined $\pi_a$ coincides with the original $\pi_a$ for $\alpha_\pi$-almost all $a$, and thus is a valid version of the conditional probability.
\end{proof}

\begin{lemma}
There exists a unique $p\in C(\Theta)$ that solves (D).
\end{lemma}
\begin{proof}
Recall that in the main text we take an arbitrary solution $p$ to (D). Then we select $q(a)\in Q(a)$ such that the associated contact set $\Gamma$ is compact. By the definition of the contact set, we have
\[
p(\theta)=V(a,\theta)+q(a)u(a,\theta),\quad \text{for all $(a,\theta)\in \Gamma$}.
\]
Fix any solution $\pi$ to (P). By Theorem \ref{t:contact}, $\Sigma:=\supp (\pi)\subset \Gamma$. Let $\Sigma_a$ denote the $a$-section of $\Sigma$. Define the set $\Sigma^\star\subset \Sigma$ by letting its $a$-section be given by
\[
\Sigma^\star_a=
\begin{cases}
\{\theta^\star(a)\}, &\theta^\star (a)\in \{\min\Sigma_a,\max \Sigma_a\},\\
\Sigma_a, &\text{otherwise},
\end{cases}
\quad \text{for all $a\in A$}.
\]
Since $\Sigma\subset \Gamma$, we get $\Sigma^\star\subset \Gamma^\star $.
By Lemma \ref{l:pia}, $\pi(\Sigma^\star)=1$. Let the projection of $\Sigma^\star$ on $\Theta$ be defined as $\Theta_{\Sigma^\star} = \{\theta\in \Theta: (a,\theta)\in \Sigma^\star\text{ for some }a\in A\}$. Then, $\phi(\Theta_{\Sigma^\star})=1$ and the closure of $\Theta_{\Sigma^\star}$ is $\Theta$. 

Next take any $\theta\in \Theta_{\Sigma^\star}$. If $(a^\star (\delta_\theta),\theta)\in \Sigma^\star$, then $p(\theta)=V(a^\star(\delta_\theta),\theta)$. Otherwise, by the definition of $\Sigma^\star$, there exist $a\in A$ and $\theta'\in \Theta$ such that $(a,\theta),(a,\theta')\in \Sigma^\star$ and either $\theta<\theta^\star(a)<\theta'$ or $\theta'<\theta^\star(a)<\theta$. Suppose that $\theta<\theta^\star(a)<\theta'$ (the other case is analogous and omitted). By Theorem \ref{t:contact}, we have
\begin{align*}
v(a,\theta)+q(a)u_a(a,\theta)+q'(a)u(a,\theta) &=0,\\
v(a,\theta')+q(a)u_a(a,\theta')+q'(a)u(a,\theta') &=0.	
\end{align*}
Adding the first equation multiplied by $u(a,\theta')$ and the second multipliled by $-u(a,\theta)$, we obtain
\[
q(a)=-\frac{v(a,\theta)u(a,\theta')-v(a,\theta')u(a,\theta)}{u_a(a,\theta)u(a,\theta')-u_a(a,\theta')u(a,\theta)},
\]
which is well-defined because the denominator is strictly negative by Assumption \ref{a:qc}. Consequently, $p(\theta)=V(a,\theta)+q(a)u(a,\theta)$. In sum, for each $\theta\in \Theta_{\Sigma^\star}$, an arbitrary solution $p(\theta)$ to (D) is determined by $\Sigma^\star$, which is constructed from a fixed solution $\pi$ to (P). Moreover, since  $\Theta$ is the closure of $\Theta_{\Sigma^\star}$, there is a unique continuous extension of $p$ from $\Theta_{\Sigma^\star}$ to $\Theta$. This shows that there is a unique $p\in C(\Theta)$ that solves (D).
\end{proof}

\subsection{Proof of Theorem \ref{p:S}}
We first prove part (2).
Suppose by contradiction that there exist $(a,\theta_1)$, $(a,\theta_2)$, and $(a,\theta_3)$ in $\Gamma^\star$ with $\theta_1<\theta_2<\theta_3$. Then, by the definition of $\Gamma^\star$, we have $\min  \Gamma_a^\star<\theta^\star(a)<\max \Gamma_a^\star$. Thus, by redefining $\theta_1=\min \Gamma_a^\star$ and $\theta_3=\max \Gamma_a^\star$ if necessary, we can assume that $\theta_1<\theta^\star(a)<\theta_3$, so \eqref{e:sing} holds. But this implies that the rows of the matrix $S$ are linearly independent, which contradicts the fact that \eqref{e:FOC} holds at $(a,\theta_1)$, $(a,\theta_2)$, and $(a,\theta_3)$. Thus, $|\Gamma_a^\star|\leq 2$ for all $a\in A$. 

We now turn to part (1). For any $\mu \in
\Delta (\Theta )$, denote the set of distributions of posteriors with average posterior equal to $\mu$ by
\begin{equation*}
\Delta _{2}\left( \mu \right) =\left\{ \tau \in \Delta (\Delta (\Theta
)):\int_{\Delta (\Theta )}\eta \df \tau \left( \eta \right) =\mu \right\} \text{%
.}
\end{equation*}

Let $\Delta _{2}^{Bin}(\mu )\subset \Delta _{2}(\mu)$ denote the
set of such distributions where in addition the posterior is always supported on at most two states: 
\begin{equation*}
\Delta _{2}^{Bin} (\mu)=\left\{ \tau \in \Delta _{2}(\mu):\supp(\tau)\subset  \Delta_{1}^{Bin} \right\} \text{,}
\end{equation*}%
where%
\begin{equation*}
\Delta _{1}^{Bin}=\left\{ \eta \in \Delta (\Theta):\left\vert \supp(\eta)
\right\vert \leq 2\right\} \text{.}
\end{equation*}

We wish to show that for each $\tau\in \Delta_2(\phi)$, there exists $\hat \tau \in
\Delta_2^{Bin}(\phi) $ such that $\pi_{\hat \tau}=\pi_\tau$. 

We set the stage by defining some key objects and establishing their
properties. Define $\Delta _{1}=\Delta \left( \Theta \right) $ and $\Delta
_{2}=\Delta \left( \Delta \left( \Theta \right) \right) $. Since $\Theta$ is
compact, the sets $\Delta_1$ and $\Delta_2$ are also compact (in the weak*
topology), by Prokhorov's Theorem (Theorem 15.11 in \citealt{aliprantis2006}). Moreover, $\Delta
_{2}\left( \mu \right)$ is compact, since it is a closed subset of the
compact set $\Delta_2$.

Define the correspondence $P:\Delta _{1}\rightrightarrows \Delta _{1}$
as%
\begin{equation*}
P\left( \mu \right) =\left\{ \eta \in \Delta _{1}:\int u\left( a ^{\star
}\left( \mu \right) ,\theta \right) \df \eta \left( \theta \right) =0 \right\} 
\text{.}
\end{equation*}%
For each $\mu \in \Delta _{1}$, $P\left( \mu \right) $ is a \emph{moment set}---a set of probability measures $\eta \in \Delta _{1}$ satisfying a given
moment condition (e.g., \citealt{winkler}). By Assumption \ref{a:qc}, we
have, for all $\mu ,\eta \in \Delta _{1}$,%
\begin{equation}  \label{e:Prho*}
\eta \in P\left( \mu \right) \iff a ^{\star }\left( \mu \right) =a
^{\star }\left( \eta \right).
\end{equation}%
Clearly, $P(\mu)$ is nonempty (as $\mu\in P(\mu)$) and convex. Since $u$ is
continuous in $\theta$, $P(\mu)$ is a closed subset of $\Delta_1$, and hence is
compact. Moreover, the correspondence $P$ has a closed graph. Indeed, consider
two sequences $\mu_n\rightarrow \mu \in \Delta_1$ and $\eta_n\rightarrow
\eta\in \Delta_1$ with $\mu_n\in \Delta_1$ and $\eta_n\in P(\mu_n),$ so that
\ 
\begin{equation*}
\int u\left( a ^{\star }\left( \mu _{n}\right) ,\theta \right) \df \eta
_{n}\left( \theta \right) =0\text{.}
\end{equation*}%
Note that $a^\star(\mu)$ is a continuous function of $\mu$, by Berge's theorem (Theorem 17.31 in \citealt{aliprantis2006}). Since $u$ is also continuous, by Corollary 15.7 in \citet{aliprantis2006} we
have%
\begin{equation*}
\int u\left( a ^{\ast }\left( \mu \right) ,\theta \right) \df \eta \left(
\theta \right) =0,
\end{equation*}%
proving that $\eta\in P(\mu) $, so $P$ has a closed graph. 

Define the correspondence $E:\Delta _{1}\rightrightarrows \Delta _{1}$ as%
\begin{equation*}
E\left( \mu \right) =P\left( \mu \right) \cap \Delta _{1}^{Bin}=\left\{ \eta
\in P\left( \mu \right) :\left\vert \supp \eta \right\vert \leq 2\right\}.
\end{equation*}
Notice that for each $\mu\in \Delta_1$, the support of $\mu$ is well
defined, by Theorem 12.14 in \citet{aliprantis2006}. Moreover, from the proof of Theorem 15.8 in
\citet{aliprantis2006}, it follows that $\Delta_1^{Bin}$ is a closed subset of $\Delta_1$, so
both $\Delta_1^{Bin}$ and $E(\mu)$ are compact.

Define the correspondence $\Lambda :\Delta _{1}\rightrightarrows \Delta _{2}$
as%
\begin{equation*}
\Lambda \left( \mu \right) =\left\{ \lambda \in \Delta \left( E\left( \mu
\right) \right) :\mu =\int_{E\left( \mu \right) }\eta \df \lambda \left( \eta
\right) \right\} \text{.}
\end{equation*}
Lemma \ref{lm_binary} shows that the correspondence $\Lambda$ admits a measurable selection. In turn, Lemma \ref{lm_binary} relies on the following lemma, which follows immediately from the Choquet Theorem (Theorem 3.1 in \citealt{winkler}) and Richter-Rogosinsky's Theorem (Theorem 2.1 in  \citealt{winkler}).
\begin{lemma}\label{l:Choquet}
Let Assumptions \ref{a:smooth} and \ref{a:qc} hold. For any $a\in A$ and $\mu\in \Delta(\Theta)$ such that $\int u(a,\theta)\df\mu=0$, there exists $\lambda_\mu \in \Delta (\Delta(\Theta))$ such that $\int \eta \df\lambda_\mu =\mu$ and for each $\eta\in \supp (\lambda_\mu)$ we have $\int u(a,\theta) \df\eta=0$ and $|\supp (\eta)|\leq 2$.
\end{lemma}

\begin{lemma}\label{lm_binary} There exists a measurable function $\mu \mapsto \lambda
_{\mu }\in \Lambda \left( \mu \right) $. 
\end{lemma}

\begin{proof}
The correspondence $\Lambda $ is nonempty-valued, by Lemma \ref{l:Choquet}.
Next, fix $\mu \in \Delta _{1}$, and consider a sequence $\lambda
_{n}\rightarrow \lambda\in \Delta_2 $ with $\lambda _{n}\in \Lambda \left(
\mu \right)$. By the Portmanteau Theorem (Theorem 15.3 in \citealt{aliprantis2006}), we have%
\begin{equation*}
\int_{E\left( \mu \right) }\eta \df \lambda _{n}\left( \eta \right) \rightarrow
\int_{E\left( \mu \right) }\eta \df \lambda \left( \eta \right) \quad \text{and} \quad \lim \sup_{n}\lambda _{n}\left( E\left( \mu \right) \right) \leq \lambda
\left( E\left( \mu \right) \right),
\end{equation*}%
where the last inequality holds because $E\left( \mu \right) $ is closed.
Thus, 
\begin{equation*}
\int_{E\left( \mu \right) }\eta \df \lambda \left( \eta \right) =\mu \quad \text{and} \quad 1=\lim \sup_{n}\lambda _{n}\left( E\left( \mu \right) \right) \leq \lambda
\left( E\left( \mu \right) \right) \leq 1,
\end{equation*}%
proving that $\lambda \in \Lambda \left( \mu \right) $. Thus, $\Lambda $ is closed-valued.

Next, consider two sequences $\mu_n\rightarrow \mu \in \Delta_1$ and $%
\lambda_n\rightarrow \lambda\in \Delta_2$ with $\mu_n\in \Delta_1$ and $%
\lambda_n\in \Lambda(\mu_n)$, so that 
\begin{gather*}
\mu _{n}=\int \eta \df \lambda _{n}\left( \eta \right), \quad \lambda _{n}\left( \Delta _{1}^{Bin}\right) =1, \quad \text{ and } \quad \lambda _{n}\left( P\left( \mu _{n}\right) \right) =1.
\end{gather*}%
The Portmanteau Theorem implies that $\mu =\int \eta \df \lambda \left( \eta
\right)$ and $\lambda \left( \Delta _{1}^{Bin}\right) =1$, since $\Delta
_{1}^{Bin}$ is closed. Define $\overline P(\mu_n)$ as the closure of $%
\cup_{k=n}^\infty P(\mu_k)$. By construction, $P(\mu_k)\subset \overline
P(\mu_k)\subset \overline P(\mu_n)$ for $k\geq n$, so the Portmanteau
Theorem implies that $\lambda (\overline P(\mu_n))=1.$ Moreover, $\overline
P (\mu_n)\downarrow \overline P \subset P(\mu)$, because $P$ has a closed
graph. Hence, $\lambda (P(\mu))=1$, by the continuity of probability
measures (Theorem 10.8 in \citealt{aliprantis2006}). That is, $\lambda \in \Lambda (\mu)$, showing
that the correspondence $\Lambda$ has a closed graph.

Therefore, $\Lambda $ is measurable, by Theorem 18.20 in \citet{aliprantis2006}, as well as nonempty- and closed-valued. Hence, there
exists a measurable function $\mu \mapsto \lambda _{\mu }\in \Lambda \left(
\mu \right) $, by Theorem 18.13 in \citet{aliprantis2006}.
\end{proof}

Finally, taking a measurable selection, for each $\tau \in \Delta _{2}\left( \phi \right) $, define $\hat{%
\tau}\in \Delta _2 $ as%
\begin{equation}  \label{e:tauhat}
\hat{\tau}\left( \widetilde \Delta_1\right) =\int_{\Delta _{1}}\lambda _{\mu
}\left( \widetilde \Delta_1\right) \df \tau \left( \mu \right)
\end{equation}%
for every measurable set $\widetilde \Delta_1\subset \Delta _{1}$. By
construction, $\hat{\tau}\in \Delta _{2}^{Bin}(\phi)$, since 
\begin{align*}
\hat{\tau}(\Delta_{1}^{Bin})=\int_{\Delta_{1}}\lambda_{\mu}(%
\Delta_{1}^{Bin})\df \tau(\mu)=1
\end{align*}
and 
\begin{align*}
\phi=\int_{\Delta_1} \mu \df \tau (\mu) = \int_{\Delta_1}
\left(\int_{E(\mu)}\eta \df \lambda_\mu(\eta)\right) \df \tau
(\mu)=\int_{\Delta_1} \eta \df \hat \tau (\eta),
\end{align*}
where the first equality holds by $\tau \in \Delta_2(\phi)$, the second by $\lambda_\mu\in \Lambda$, and the third by \eqref{e:tauhat}. Similarly, for each
measurable $\widetilde A\subset A$ and $\widetilde \Theta\subset \Theta$, we
have 
\begin{align*}
\pi_{\tau} (\widetilde A,\widetilde \Theta) &=
\int_{\Delta_1}1\{a^*(\mu)\in \widetilde A\}\mu(\widetilde \Theta)\df \tau
(\mu) \\
&=\int_{\Delta_1}1\{a^*(\mu)\in \widetilde
A\}\left(\int_{E(\mu)}\eta(\widetilde \Theta)d\lambda_\mu(\eta)\right)\df \tau
(\mu) \\
&=\int_{\Delta_1}\left(\int_{E(\mu)}1\{a^*(\eta)\in \widetilde
A\}\eta(\widetilde \Theta)\df \lambda_\mu(\eta)\right)\df \tau (\mu) \\
&=\int_{\Delta_1} 1\{a^*(\eta)\in \widetilde A\}\eta (\widetilde
\Theta)\df \hat \tau (\eta) \\
&=\pi_{\hat\tau} (\widetilde A, \widetilde \Theta),
\end{align*}
where the second equality holds by $\lambda_\mu\in \Lambda$, the third by %
\eqref{e:Prho*} and $E(\mu)\subset P(\mu)$, and the fourth by %
\eqref{e:tauhat}.

\subsection{Proof of Corollary \ref{c:nodisc}}
Let $a$ be such that $\int u(a,\theta)\df \phi=0$. Since $|\Theta|\geq 3$, Assumption \ref{a:sc} and $\int u(a,\theta)\df \phi=0$ imply that there exist $\theta_1<\theta_2<\theta_3$ in $\Theta$ such that $\theta_1<\theta^\star(a)<\theta_3$. 

Suppose that no disclosure is optimal. Then, by part (2) of Theorem \ref{t:contact}, it follows that $\Gamma^\star_{a}=\Gamma_{a}=\Theta$ and \eqref{e:FOC} holds for all $\theta\in \Theta$, so there exist constants $q(a),q'(a)\in \R$ such that
\[
v(a,\theta)=-q(a)u_a(a,\theta)-q'(a)u(a,\theta)\quad \text{for all $\theta\in \Theta$}.
\]
That is, $v(a,\cdot)$ lies in a linear space $L$ spanned by $u_a(a,\cdot)$ and $u(a,\cdot)$, whose dimension is at most 2. But the space of functions $v(a,\cdot)$ satisfying Assumption \ref{a:smooth} is the linear space $C(\Theta)$, whose dimension is at least 3, since $|\Theta|\geq 3$. Hence, the space $L$ is a proper subspace of $C(\Theta)$, so generically $v(a,\cdot)$ does not belong to $L$, and thus generically no disclosure is suboptimal.

\subsection{Proof of Theorem \ref{l:ssdd}}
We give the proof of the theorem for the single-dipped case. We start with an appropriate version of the theorem of alternative.
\begin{lemma}\label{l:alter}
Exactly one of the following two alternatives holds.
\begin{enumerate}
	\item There exists $x>0$ such that $xR\leq 0$.
	\item There exists $y\geq 0$ such that $Ry\geq 0$ and $Ry\neq 0$.
\end{enumerate}
\end{lemma}
\begin{proof}
Clearly, (1) and (2) cannot both hold, because premultiplying $Ry\geq 0$ with $Ry\neq0$ by $x> 0$ yields $xRy>0$, whereas postmultiplying $xR\leq 0$ by $y\geq 0$ yields $xRy\leq 0.$

Now suppose that (1) does not hold. Then there does not exist $x\geq 0$ such that
\begin{equation*}
x
\begin{pmatrix}
	R &-I 
\end{pmatrix}
\leq
\begin{pmatrix}
	0 &-e
\end{pmatrix}
\end{equation*}
where $I$ is an identity matrix and $e$ is a row vector of ones. Thus, by the theorem of alternative (e.g., Theorem 2.10 in \citealt{gale}), there exists $y\geq 0$ and $z\geq 0$ such that
\begin{equation*}
\begin{pmatrix}
	R\\
	-I
\end{pmatrix}
\begin{pmatrix}
	y &z
\end{pmatrix}
\geq 0\quad \text{and}\quad -ez<0,
\end{equation*}
which in turn shows that (2) holds.
\end{proof}

We prove the theorem by contraposition. Suppose that $\Gamma$ is not single-dipped, so it contains a strictly single-peaked triple $(a_1,\theta_1)$, $(a_2,\theta_2)$, $(a_1,\theta_3)$. Without loss, we can assume that $\theta_1\leq \theta^\star(a)\leq \theta_3$. This is because $\min \Gamma_{a_1}\leq \theta^\star (a_1)\leq \max \Gamma_{a_1}$ by Theorem \ref{t:contact}, and thus by Assumption \ref{a:sc} the triple $(a_1,\min \Gamma_{a_1})$, $(a_2,\theta_2)$, $(a_1,\max \Gamma_{a_1})$ is strictly single-peaked triple and lies in $\Gamma$.	

By (D1) and Theorem \ref{t:contact}, we have
\begin{align*}
V(a_1,\theta_1)+q(a_1) u(a_1,\theta_1) &\geq V(a_2,\theta_1) +q(a_2) u(a_2,\theta_1),\\
V(a_2,\theta_2)+q(a_2) u(a_2,\theta_2) &\geq V(a_1,\theta_2) +q(a_1) u(a_1,\theta_2),\\
V(a_1,\theta_3)+q(a_1) u(a_1,\theta_3) &\geq V(a_2,\theta_3) +q(a_2) u(a_2,\theta_3).
\end{align*}
By \eqref{eqn:q}, for an optimal $\pi_a$, and for $i\in\{1,2\}$, we have
\[
q(a_i)=\frac{\E_{\pi_{a_i}}[v(a_i,\theta)]}{-\E_{\pi_{a_i}} [u_a(a_i,\theta)]}>0,
\]
where the inequality follows from Assumptions \ref{a:qc} and \ref{a:sc}. Thus, the vector $x=(1,q(a_1),q(a_2))$ is strictly positive and satisfies $x R\leq 0$. By Lemma \ref{l:alter}, there does not exist a vector $y\geq 0$ such that $Ry\geq 0$ and $Ry\neq 0$, as desired.

\subsection{Proof of Theorem \ref{t:SDPD}}

The set $\Gamma$ is single-dipped (-peaked) by Theorem \ref{l:ssdd} with 
\[
y=
\begin{pmatrix}
u(a_2,\theta_3)u(a_1,\theta_2)-u(a_2,\theta_2)u(a_1,\theta_3)\\
u(a_2,\theta_3)u(a_1,\theta_1)-u(a_2,\theta_1)u(a_1,\theta_3)\\
u(a_2,\theta_2)u(a_1,\theta_1)-u(a_2,\theta_1)u(a_1,\theta_2)
\end{pmatrix}
\quad
\left( y=
\begin{pmatrix}
\frac{u(a_2,\theta_1)}{V(a_2,\theta_1)-V(a_1,\theta_1)} \\
\frac{u(a_2,\theta_2)}{V(a_2,\theta_2)-V(a_1,\theta_2)} \\
\frac{u(a_2,\theta_3)}{V(a_2,\theta_3)-V(a_1,\theta_3)}
\end{pmatrix}\right),
\]
as follows from Lemma \ref{l:R>0} and Lemma \ref{l:ySDD} (Lemma \ref{l:ySPD}). Moreover, $|\Gamma_a^\star|\leq 2$ for all $a$ by Theorem \ref{t:pairwise} and Lemma \ref{l:S><0}, showing that $\Gamma^\star$ is strictly single-dipped (-peaked).
Finally, consider
\[
v^n(a,\theta)=v(a,\theta)+\int_0^\theta \frac{\tilde v(\theta)}{n}u_\theta (a,\tilde \theta) \df \tilde \theta,
\]
where $\tilde v (\theta)$ is a continuous, strictly positive, and strictly increasing (decreasing) function on $\ol \Theta$. Then $v^n(a,\theta)>0$ because $v(a,\theta)>0$ and $u_\theta(a,\theta)>0$ for all $(a,\theta)$, by Assumptions \ref{a:v>0} and \ref{a:u_s}. Moreover, for all $a_2\geq (\leq)a_1$,
\[
\frac{v_\theta^n(a_2,\theta)}{u_\theta(a_1,\theta)}=\frac{v_\theta(a_2,\theta)}{u_\theta(a_1,\theta)} + \frac{\tilde v(\theta)}{n} \frac{u_\theta(a_2,\theta)}{u_\theta(a_1,\theta)}
\]
is strictly increasing (decreasing) in $\theta$, because  $\tilde v (\theta)$ is strictly positive and strictly increasing (decreasing) in $\theta$; ${v_\theta(a_2,\theta)}/{u_\theta(a_1,\theta)}$ is increasing (decreasing) in $\theta$; and ${u_\theta(a_2,\theta)}/{u_\theta(a_1,\theta)}$ is increasing in $\theta$, since $u_{a\theta}(a,\theta)/u_{\theta}(a,\theta)$ is increasing (decreasing) in $\theta$. Thus, by Lemma \ref{l:stab}, there exists an optimal single-dipped (-peaked) outcome.

\begin{lemma}\label{l:S><0}
If $u_{a \theta}(a,\theta)/u_\theta(a,\theta)$ and $v_\theta(a,\theta)/u_\theta(a,\theta)$ are increasing (decreasing) in $\theta$ for all $a$, with at least one of them strictly increasing (decreasing), then $|S|>(<)0$ for all $a$ and $\theta_1<\theta_2<\theta_3$ such that $\theta_1<\theta^\star(a)<\theta_3$.
\end{lemma}

\begin{proof}
We consider the case where $u_{a \theta}/u_\theta$ and $v_\theta/u_\theta$ are increasing in $\theta$; the case where $u_{a \theta}/u_\theta$ and $v_\theta/u_\theta$ are decreasing in $\theta$ is analogous and thus omitted. 

Fix $\theta_1<\theta_2<\theta_3$ and $a$ such that $u(a,\theta_1)<0<u(a,\theta_3)$. The inequality $|S|>0$ follows from the following displayed equations: 
\[
u(a,\theta_3)-u(a,\theta_1)=\int_{\theta_1}^{\theta_3} u_\theta (a,\theta)\df \theta>0,
\]
where the inequality holds by Assumption \ref{a:u_s};
\[
\begin{vmatrix}
	u(a,\theta_1) &u(a,\theta_3)\\
	u_{a}(a,\theta_1) &u_{a}(a,\theta_3)
\end{vmatrix}=-u(a,\theta_3)u_{a}(a,\theta_1) + u(a,\theta_1)u_{a}(a,\theta_3)>0,
\]
where the inequality holds by part (2) of Lemma \ref{l:ASC};
\[
\begin{vmatrix}
	v(a,\theta_1) & v(a,\theta_3)\\
	u(a,\theta_1) & u(a,\theta_3)
\end{vmatrix} 
=u(a,\theta_3)v(a,\theta_1) -u(a,\theta_1)v(a,\theta_3)>0,
\]
where the inequality holds by Assumption \ref{a:v>0}; 
\begin{gather*}
-\begin{vmatrix}
v(a,\theta_2)-v(a,\theta_1) &v(a,\theta_3)-v(a,\theta_2)\\
u(a,\theta_2)-u(a,\theta_1) &u(a,\theta_3)-u(a,\theta_2)
\end{vmatrix}\\
=(v(a,\theta_3)-v(a,\theta_2))(u(a,\theta_2)-u(a,\theta_1))-(v(a,\theta_2)-v(a,\theta_1))(u(a,\theta_3)-u(a,\theta_2))\\
=\int_{\theta_2}^{\theta_3}\int_{\theta_1}^{\theta_2} (v_\theta(a,\tilde \theta)u_\theta(a, \theta) -v_\theta(a,\theta)u_\theta(a,\tilde \theta))\df \theta \df \tilde \theta \geq (>) 0,
\end{gather*}
where the inequality holds by Assumption \ref{a:u_s} and (strict) monotonicity of $v_\theta /u_\theta$ in $\theta$;
\begin{gather*}
\begin{vmatrix}
u(a,\theta_2)-u(a,\theta_1) &u(a,\theta_3)-u(a,\theta_2)\\
u_a(a,\theta_2)-u_a(a,\theta_1) &u_a(a,\theta_3)-u_a(a,\theta_2)
\end{vmatrix}\\
=(u(a,\theta_2)-u(a,\theta_1))(u_a(a,\theta_3)-u_a(a,\theta_2))-(u(a,\theta_3)-u(a,\theta_2))(u_a(a,\theta_2)-u_a(a,\theta_1))\\
=\int_{\theta_2}^{\theta_3}\int_{\theta_1}^{\theta_2} (u_\theta(a,\theta) u_{a\theta}(a,\tilde \theta) -u_\theta(a,\tilde \theta)u_{a\theta}(a,\theta)) \df \theta \df \tilde\theta \geq(>)0,
\end{gather*}
where the inequality holds by Assumption \ref{a:u_s} and (strict) monotonicity of $u_{a\theta}/u_{\theta}$ in $\theta$;
\begin{align*}
&\frac{
\begin{vmatrix}
v(a,\theta_1) &v(a,\theta_2) &v(a,\theta_3)\\
u(a,\theta_1) &u(a,\theta_2) &u(a,\theta_3) \\
u_a(a,\theta_1) &u_a(a,\theta_2) &u_a(a,\theta_3) 
\end{vmatrix}
}
{
\begin{vmatrix}
	u(a,\theta_1) &u(a,\theta_3)\\
	u_{a}(a,\theta_1) &u_{a}(a,\theta_3)
\end{vmatrix}
}
(u(a,\theta_3)-u(a,\theta_1))\\
=-&\begin{vmatrix}
v(a,\theta_2)-v(a,\theta_1) &v(a,\theta_3)-v(a,\theta_2)\\
u(a,\theta_2)-u(a,\theta_1) &u(a,\theta_3)-u(a,\theta_2)
\end{vmatrix}\\
+&
\frac{
\begin{vmatrix}
	v(a,\theta_1) & v(a,\theta_3)\\
	u(a,\theta_1) & u(a,\theta_3)
\end{vmatrix} 
}
{
\begin{vmatrix}
	u(a,\theta_1) &u(a,\theta_3)\\
	u_{a}(a,\theta_1) &u_{a}(a,\theta_3)
\end{vmatrix}
}
\begin{vmatrix}
u(a,\theta_2)-u(a,\theta_1) &u(a,\theta_3)-u(a,\theta_2)\\
u_a(a,\theta_2)-u_a(a,\theta_1) &u_a(a,\theta_3)-u_a(a,\theta_2)
\end{vmatrix},
\end{align*}
where the equality holds by rearrangement.
\end{proof}

\begin{lemma}\label{l:R>0}
If $u_{a\theta}(a,\theta)/u_{\theta}(a,\theta)$ and $v_\theta (a_2,\theta)/u_\theta(a_1,\theta)$ are increasing (decreasing) in $\theta$ for all $a$ and $a_2\geq(\leq) a_1$, with at least one of them strictly increasing (decreasing), then $|R|>(<)0$ for all $\theta_1<\theta_2<\theta_3$ and all $a_2>(<)a_1$ such that $\theta_1\leq \theta^\star(a_1)\leq \theta_3$.
\end{lemma}
\begin{proof}
We consider the case where $u_{a\theta}/u_{\theta}$ and $v_\theta /u_\theta$ are increasing in $\theta$; the case where $u_{a \theta}/u_\theta$ and $v_\theta/u_\theta$ are decreasing in $\theta$ is analogous and thus omitted. 

Fix $\theta_1<\theta_2<\theta_3$ and $a_2>a_1$ such that $u(a_1,\theta_1)<0<u(a_1,\theta_2)$. The inequality $|R|>0$ follows from the following displayed equations: 
\[
u(a_1,\theta_3)-u(a_1,\theta_1)=\int_{\theta_1}^{\theta_3} u_\theta (a_1,\theta)\df \theta>0,
\]
where the inequality holds by Assumption \ref{a:u_s};
\begin{gather*}
\begin{vmatrix}
	u(a_1,\theta_1) &u(a_1,\theta_3)\\
	u(a_2,\theta_1) &u(a_2,\theta_3)
\end{vmatrix}\\
=-u(a_1,\theta_3)u(a_2,\theta_1) + u(a_1,\theta_1)u(a_2,\theta_3)\\
=-g(a_1)\tilde u(a_1,\theta_3)g(a_2)\tilde u(a_2,\theta_1) + g(a_1)\tilde u(a_1,\theta_1)g(a_2)\tilde u(a_2,\theta_3)\\
=g(a_1)g(a_2)[-\tilde u(a_1,\theta_3)(\tilde u(a_2,\theta_1)-\tilde u(a_1,\theta_1))+\tilde u(a_1,\theta_1)(\tilde u(a_2,\theta_3)-\tilde u(a_1,\theta_3))]\\
=g(a_1)g(a_2)\int_{a_1}^{a_2}[-\tilde u(a_1,\theta_3) \tilde u_a(a,\theta_1)+\tilde u(a_1,\theta_1)\tilde u_a (a,\theta_3) ]\df a>0,
\end{gather*}
where  the inequality and the second equality hold by parts (2) and (3) of Lemma \ref{l:ASC};
\begin{gather*}
\begin{vmatrix}
	V(a_2,\theta_1)-V(a_1,\theta_1) & V(a_2,\theta_3)-V(a_1,\theta_3)\\
	u(a_1,\theta_1) & u(a_1,\theta_3)
\end{vmatrix} \\
=u(a_1,\theta_3)\int_{a_1}^{a_2} v(a,\theta_1)\df a -u(a_1,\theta_1)\int_{a_1}^{a_2} v(a,\theta_3)\df a>0,	
\end{gather*}
where the inequality holds by Assumption \ref{a:v>0};
\begin{gather*}
-\begin{vmatrix}
V(a_2,\theta_2)-V(a_1,\theta_2)-V(a_2,\theta_1)+V(a_1,\theta_1) &V(a_2,\theta_3)-V(a_1,\theta_3)-V(a_2,\theta_2)+V(a_1,\theta_2)\\
u(a_1,\theta_2)-u(a_1,\theta_1) &u(a_1,\theta_3)-u(a_1,\theta_2)
\end{vmatrix}\\
=(V(a_2,\theta_3)-V(a_1,\theta_3)-V(a_2,\theta_2)+V(a_1,\theta_2))(u(a_1,\theta_2)-u(a_1,\theta_1))\\
-(V(a_2,\theta_2)-V(a_1,\theta_2)-V(a_2,\theta_1)+V(a_1,\theta_1))(u(a_1,\theta_3)-u(a_1,\theta_2))\\
=\int_{a_1}^{a_2}\int_{\theta_2}^{\theta_3}\int_{\theta_1}^{\theta_2} (v_\theta(a,\tilde \theta)u_\theta(a_1, \theta) -v_\theta(a,\theta)u_\theta(a_1,\tilde \theta))\df \theta \df \tilde \theta \df a \geq (>) 0,
\end{gather*}
where the inequality holds by Assumption \ref{a:u_s} and (strict) monotonicity of $v_\theta /u_\theta$ in $\theta$;
\begin{gather*}
\begin{vmatrix}
u(a_1,\theta_2)-u(a_1,\theta_1) &u(a_1,\theta_3)-u(a_1,\theta_2)\\
u(a_2,\theta_2)-u(a_2,\theta_1) &u(a_2,\theta_3)-u(a_2,\theta_2)
\end{vmatrix}\\
=(u(a_1,\theta_2)-u(a_1,\theta_1))(u(a_2,\theta_3)-u(a_2,\theta_2))-(u(a_1,\theta_3)-u(a_1,\theta_2))(u(a_2,\theta_2)-u(a_2,\theta_1))\\
= \int_{\theta_2}^{\theta_3}\int_{\theta_1}^{\theta_2}  (u_\theta(a_1,\theta)u_\theta(a_2,\tilde \theta)-u_\theta(a_1,\tilde \theta)u_\theta(a_2,\theta) ) \df \theta \df \tilde\theta  \geq(>)0,
\end{gather*}
where the inequality holds by Assumption \ref{a:u_s} and (strict) monotonicity of $u_{a\theta}/u_{\theta}$ in $\theta$, which imply that, for $a_2>a_1$ and $\tilde \theta>\theta$, we have
\begin{align*}
\ln \frac{u_\theta(a_1,\theta)u_\theta(a_2,\tilde \theta)}{u_\theta(a_1,\tilde \theta)u_\theta(a_2,\theta)} &= \int_{a_1}^{a_2}\frac{\partial }{\partial a} [\ln u_\theta (a,\tilde \theta)-\ln u_\theta (a,\theta)]\df a =\int_{a_1}^{a_2} \left[\frac{u_{a\theta}(a,\tilde\theta)}{u_\theta(a,\tilde\theta)}-\frac{u_{a\theta}(a,\theta)}{u_\theta(a,\theta)}\right]\df a\geq (>)0;
\end{align*}
\begin{align*}
&\frac{
\begin{vmatrix}
V(a_2,\theta_1)-V(a_1,\theta_1) &-(V(a_2,\theta_2)-V(a_1,\theta_2)) &V(a_2,\theta_3)-V(a_1,\theta_3)\\
-u(a_1,\theta_1) &u(a_1,\theta_2) &-u(a_1,\theta_3) \\
u(a_2,\theta_1) &-u(a_2,\theta_2) &u(a_2,\theta_3)
\end{vmatrix}
}
{
\begin{vmatrix}
	u(a_1,\theta_1) &u(a_1,\theta_3)\\
	u(a_2,\theta_1) &u(a_2,\theta_3)
\end{vmatrix}
}
(u(a_1,\theta_3)-u(a_1,\theta_1))\\
=-&\begin{vmatrix}
V(a_2,\theta_2)-V(a_1,\theta_2)-V(a_2,\theta_1)+V(a_1,\theta_1) &V(a_2,\theta_3)-V(a_1,\theta_3)-V(a_2,\theta_2)+V(a_1,\theta_2)\\
u(a_1,\theta_2)-u(a_1,\theta_1) &u(a_1,\theta_3)-u(a_1,\theta_2)
\end{vmatrix}\\
+&
\frac{
\begin{vmatrix}
	V(a_2,\theta_1)-V(a_1,\theta_1) & V(a_2,\theta_3)-V(a_1,\theta_3)\\
	u(a_1,\theta_1) & u(a_1,\theta_3)
\end{vmatrix}
}
{
\begin{vmatrix}
	u(a_1,\theta_1) &u(a_1,\theta_3)\\
	u(a_2,\theta_1) &u(a_2,\theta_3)
\end{vmatrix}
}
\begin{vmatrix}
u(a_1,\theta_2)-u(a_1,\theta_1) &u(a_1,\theta_3)-u(a_1,\theta_2)\\
u(a_2,\theta_2)-u(a_2,\theta_1) &u(a_2,\theta_3)-u(a_2,\theta_2)
\end{vmatrix},
\end{align*}
where the equality holds by rearrangement.
\end{proof}
\begin{lemma}\label{l:ySDD}
If $u_{a\theta}(a,\theta)/u_{\theta}(a,\theta)$ is increasing in $\theta$ for all $a$, then for all $\theta_1<\theta_2<\theta_3$ and all $a_2>a_1$ such that $\theta_1\leq \theta^\star(a_1)\leq \theta_3$, we have
\begin{align*}
u(a_2,\theta_3)u(a_1,\theta_1)> u(a_2,\theta_1)u(a_1,\theta_3),\\
u(a_2,\theta_2)u(a_1,\theta_1)> u(a_2,\theta_1)u(a_1,\theta_2),\\
u(a_2,\theta_3)u(a_1,\theta_2)> u(a_2,\theta_2)u(a_1,\theta_3).
\end{align*}
\end{lemma}
\begin{proof}
Fix $\theta_1<\theta_2<\theta_3$ and $a_2>a_1$ such that $u(a_1,\theta_1)<0<u(a_1,\theta_3)$. The first claimed inequality follows as in the proof of Lemma \ref{l:R>0}, by Assumption \ref{a:qc} and $u(a_1,\theta_1)<0<u(a_1,\theta_3)$. We thus focus on the second and third inequalities.

As in the proof of Lemma \ref{l:R>0}, Assumption \ref{a:u_s} and monotonicity of $u_{a\theta}/u_{\theta}$ in $\theta$ yield
\begin{gather*}
u(a_1,\theta_3)>u(a_1,\theta_2)>u(a_1,\theta_1),\\
\frac{u(a_2,\theta_3)-u(a_2,\theta_2)}{u(a_1,\theta_3)-u(a_1,\theta_2)}\geq\frac{u(a_2,\theta_2)-u(a_2,\theta_1)}{u(a_1,\theta_2)-u(a_1,\theta_1)}.
\end{gather*}
There are three cases to consider.

(1) $u(a_1,\theta_2)=0$. In this case, $u(a_2,\theta_2)<0$, by Assumption \ref{a:qc}. Thus,
\begin{align*}
u(a_2,\theta_2)u(a_1,\theta_1)&>0=u(a_2,\theta_1)u(a_1,\theta_2),\\
u(a_2,\theta_3)u(a_1,\theta_2)&=0> u(a_2,\theta_2)u(a_1,\theta_3).
\end{align*}

(2) $u(a_1,\theta_2)>0$. In this case, as follows from the proof of Lemma \ref{l:R>0},
\[
u(a_2,\theta_2)u(a_1,\theta_1)>u(a_2,\theta_1)u(a_1,\theta_2),
\]
by Assumption \ref{a:qc} and $u(a_1,\theta_1)<0<u(a_1,\theta_2)$. Thus,
\begin{gather*}
\frac{u(a_2,\theta_3)-u(a_2,\theta_2)}{u(a_1,\theta_3)-u(a_1,\theta_2)}\geq\frac{u(a_2,\theta_2)-u(a_2,\theta_1)}{u(a_1,\theta_2)-u(a_1,\theta_1)}>\frac{u(a_2,\theta_2)}{u(a_1,\theta_2)}\\
\implies u(a_2,\theta_3)u(a_1,\theta_2)> u(a_2,\theta_2)u(a_1,\theta_3).
\end{gather*}

(3) $u(a_1,\theta_2)<0$. In this case, as follows from the proof of Lemma \ref{l:R>0},
\[
u(a_2,\theta_3)u(a_1,\theta_2)> u(a_2,\theta_2)u(a_1,\theta_3),
\]
by Assumption \ref{a:qc} and $u(a_1,\theta_2)<0<u(a_1,\theta_3)$. Thus,
\begin{gather*}
\frac{u(a_2,\theta_2)}{u(a_1,\theta_2)}>\frac{u(a_2,\theta_3)-u(a_2,\theta_2)}{u(a_1,\theta_3)-u(a_1,\theta_2)}\geq\frac{u(a_2,\theta_2)-u(a_2,\theta_1)}{u(a_1,\theta_2)-u(a_1,\theta_1)}\\
\implies u(a_2,\theta_2)u(a_1,\theta_1)>u(a_2,\theta_1)u(a_1,\theta_2).\qedhere
\end{gather*}
\end{proof}

\begin{lemma}\label{l:ySPD}
If $v_\theta (a_2,\theta)/u_\theta(a_1,\theta)$ is decreasing in $\theta$ for all $a_2\leq a_1$, then for all $\theta_1<\theta_2<\theta_3$ and all $a_2<a_1$ such that $\theta_1\leq \theta^\star(a_1)\leq \theta_3$, we have
\begin{align*}
\frac{u(a_1,\theta_1)}{V(a_1,\theta_1)-V(a_2,\theta_1)}<\frac{u(a_1,\theta_2)}{V(a_1,\theta_2)-V(a_2,\theta_2)}<\frac{u(a_1,\theta_3)}{V(a_1,\theta_3)-V(a_2,\theta_3)}.
\end{align*}
\end{lemma}
\begin{proof}

Fix $\theta_1<\theta_2<\theta_3$ and $a_2<a_1$ such that $\theta_1\leq \theta^\star(a_1)\leq \theta_3$. As in the proof of Lemma \ref{l:R>0}, Assumptions \ref{a:v>0}, \ref{a:u_s}, and monotonicity of $v_\theta /u_\theta$ in $\theta$ yield
\begin{gather}
V(a_1,\theta_j)-V(a_2,\theta_j)>0\quad \text{for $j=1,2,3$},\label{e:dV>0}\\
u(a_1,\theta_3)>u(a_1,\theta_2)>u(a_1,\theta_1),\label{e:du>0}\\
\begin{gathered}
\frac{V(a_1,\theta_3)-V(a_2,\theta_3)-V(a_1,\theta_2)+V(a_2,\theta_2)}{u(a_1,\theta_3)-u(a_1,\theta_2)}\\ \leq\frac{V(a_1,\theta_2)-V(a_2,\theta_2)-V(a_1,\theta_1)+V(a_2,\theta_1)}{u(a_1,\theta_2)-u(a_1,\theta_1)}.\label{d:dV/du}
\end{gathered}
\end{gather}

There are two cases to consider.

(1) $u(a_1,\theta_2)\geq 0$. In this case,  we have
\[
\frac{u(a_1,\theta_1)}{V(a_1,\theta_1)-V(a_2,\theta_1)}<\frac{u(a_1,\theta_2)}{V(a_1,\theta_2)-V(a_2,\theta_2)},
\]
by \eqref{e:dV>0} and $u(a_1,\theta_1)<0\leq u(a_1,\theta_2)$, and
\[
\frac{u(a_1,\theta_2)}{V(a_1,\theta_2)-V(a_2,\theta_2)}<\frac{u(a_1,\theta_3)}{V(a_1,\theta_3)-V(a_2,\theta_3)},
\]
by
\begin{align*}
u(a_1,\theta_2)(V(a_1,\theta_3)-V(a_2,\theta_3))
\leq &u(a_1,\theta_2) \frac{u(a_1,\theta_3)-u(a_1,\theta_1)}{u(a_1,\theta_2)-u(a_1,\theta_1)}(V(a_1,\theta_2)-V(a_2,\theta_2))  \\
<&u(a_1,\theta_3)(V(a_1,\theta_2)-V(a_2,\theta_2)),
\end{align*}
where the first inequality holds by \eqref{d:dV/du}, $V(a_1,\theta_1)>V(a_2,\theta_1)$, $u(a_1,\theta_3)>u(a_1,\theta_2)$, and $u(a_1,\theta_2)\geq 0$, and the second inequality holds by $V(a_1,\theta_2)>V(a_2,\theta_2)$, $u(a_1,\theta_3)>u(a_1,\theta_2)$, and $u(a_1,\theta_1)<0$.

(2) $u(a_1,\theta_2)\leq 0$. In this case, we have  
\[
\frac{u(a_1,\theta_2)}{V(a_1,\theta_2)-V(a_2,\theta_2)}<\frac{u(a_1,\theta_3)}{V(a_1,\theta_3)-V(a_2,\theta_3)},
\]
by \eqref{e:dV>0} and $u(a_1,\theta_2)\leq 0< u(a_1,\theta_3)$, and
\[
\frac{u(a_1,\theta_1)}{V(a_1,\theta_1)-V(a_2,\theta_1)}<\frac{u(a_1,\theta_2)}{V(a_1,\theta_2)-V(a_2,\theta_2)},
\]
by
\begin{align*}
 -u(a_1,\theta_2)(V(a_1,\theta_1)-V(a_2,\theta_1))
\leq & -u(a_1,\theta_2)\frac{u(a_1,\theta_3)-u(a_1,\theta_1)}{u(a_1,\theta_3)-u(a_1,\theta_2)}(V(a_1,\theta_2)-V(a_2,\theta_2)) \\
<&-u(a_1,\theta_1)(V(a_1,\theta_2)-V(a_2,\theta_2)),
\end{align*}
where the first inequality holds by \eqref{d:dV/du}, $V(a_1,\theta_3)>V(a_2,\theta_3)$, $u(a_1,\theta_3)>u(a_1,\theta_2)$, and $u(a_1,\theta_2)\leq 0$, and the second inequality holds by $V(a_1,\theta_2)>V(a_2,\theta_2)$, $u(a_1,\theta_2)>u(a_1,\theta_1)$, and $u(a_1,\theta_3)>0$.
\end{proof}

\subsection{Proof of Theorem \ref{t:brenier}}
Since $\Gamma^\star$ is strictly single-dipped, we have $|\Gamma^\star_a|\leq 2$ for all $a\in A_\Gamma$, so $\Gamma^\star_a=\{t_1(a),t_2(a)\}$ with $t_1(a)=\min \Gamma^\star_a\leq \theta^\star(a)\leq \max \Gamma_a^\star = t_2(a)$ for all $a\in A_\Gamma$. Since $\Gamma$ is compact, and $\Gamma^\star$ is constructed from $\Gamma$ using a continuous function $\theta^\star(a)$, the functions $t_1$ and $t_2$ are measurable. Since $\Gamma^\star$ is single-dipped, for all $a<a'$ in $A_\Gamma$, we have $t_2(a)\leq t_2(a')$, as otherwise $(a,t_1(a))$, $(a',t_2(a'))$, $(a,t_2(a))$ would be a strictly single-peaked triple in $\Gamma^\star$; and $t_1(a')\notin (t_1(a),t_2(a))$, as otherwise $(a,t_1(a))$, $(a',t_1(a'))$, $(a,t_2(a))$ would be a strictly single-peaked triple in $\Gamma^\star$.

Suppose now that the set $\{a\in A_\Gamma:t_1(a)<t_2(a)\}$ is the union of finitely many intervals. We claim that for each $a\in A_\Gamma$ there exists $\varepsilon>0$ such that, for all $\tilde a_1,\tilde a_2\in [a-\varepsilon,a]\cap A_\Gamma$ with $t_1(\tilde a_1)\neq t_2(\tilde a_1)$ and $t_1(\tilde a_2)\neq t_2(\tilde a_2)$, we have $t_1(\tilde a_1)<t_2(\tilde a_2)$.
This claim is obvious if there does not exist a sequence $a_n\in A_\Gamma$ such that $a_n\uparrow a$, so suppose that such a sequence $a_n$ exists. By monotonicity of $t_2$, the sequence $t_2(a_n)$ converges to some $t_2(a_-)\leq t_2(a)$. In fact, we must have $t_2(a_-)=t_2(a)$, meaning that $t_2$ is left-continuous at $a$. First, if $t_1(a)<\theta^\star(a)<t_2(a)$, then $t_2(a_-)=t_2(a)$, as otherwise $\Gamma^\star_a$ would contain at least three distinct states $t_1(a)$, $t_2(a_-)$, and $t_2(a)$, by compactness of $\Gamma$, contradicting that $|\Gamma^\star_a|\leq 2$. Second, if $t_1(a)=\theta^\star(a)=t_2(a),$ then $t_2(a_-)=t_2(a)$, as otherwise there would exist $a_n\in A_\Gamma$ such that $t_2(a_n)<\theta^\star(a_n)$, contradicting that $t_1(a_n)\leq \theta^\star (a_n)\leq t_2(a_n)$. We will now show that there exists $\varepsilon>0$ with the required property. If $t_1(a)<\theta^\star (a)<t_2(a)$, then, by left-continuity of $t_2$, there exists $\varepsilon>0$ such that, for all $\tilde a_2\in [a-\varepsilon,a]\cap A_\Gamma$, we have $\theta^\star(a)< t_2(\tilde a_2)$, and thus, for all $\tilde a_1\in [a-\varepsilon,a]\cap A_\Gamma$, we have $t_1(\tilde a_1)\leq \theta^\star (\tilde a_1)\leq \theta^\star (a)<t_2(\tilde a_2).$ 
If $t_1(a)=\theta^\star (a)=t_2(a)$, then, by left-continuity of $t_2$ and the regularity condition, there exists $\varepsilon>0$ such that $t_2$ is continuous on $[a-\varepsilon,a]\cap A_\Gamma$ and either (i) $t_1(\tilde a)=\theta^\star(\tilde a)=t_2(\tilde a)$ for all $\tilde a\in [a-\varepsilon,a]\cap A_\Gamma$, or (ii) $t_1(\tilde a)<\theta^\star(\tilde a)<t_2(\tilde a)$ for all $\tilde a\in [a-\varepsilon,a)\subset  A_\Gamma$, in which case $t_1(\tilde a_1)\leq t_1(a-\varepsilon)<t_2(a-\varepsilon)\leq t_2(\tilde a_2)$ for all $\tilde a_1,\tilde a_2\in [a-\varepsilon,a)$. In particular, the inequality $t_1(\tilde a_1)\leq t_1(a-\varepsilon)$ holds because $t_1(\tilde a_1)\notin (t_1(a-\varepsilon),t_2(a-\varepsilon))$, as shown in the first paragraph,  and  $t_1(\tilde a_1)\notin [t_2(a-\varepsilon),\theta^\star(a))$, as otherwise $(a^\star(\delta_{t_1(\tilde a_1)}),t_1(a^\star(\delta_{t_1(\tilde a_1)}))$, $(\tilde a_1,t_1(\tilde a_1))$, $(a^\star(\delta_{t_1(\tilde a_1)}),t_2(a^\star(\delta_{t_1(\tilde a_1)}))$ would be a strictly single-peaked triple in $\Gamma^\star$. Thus, in both cases (i) and (ii), there exists $\varepsilon>0$ with the required property.

Suppose now that $\phi $ has a density. Suppose for contradiction that there exist two distinct optimal outcomes $\pi$ and $\pi'$. Recall that, because $|\Gamma_a^\star|\leq 2$ for all $a$, we have  $\pi_a=\pi_a'=\rho_a \delta_{t_1(a)}+(1-\rho_a)\delta_{t_2(a)}$ for all $a\in A_\Gamma$ where $1-\rho_a>0$ is given by
\[
1-\rho_a=
\begin{cases}
\frac{-u(a,t_1(a))}{u(a,t_2(a))-u(a,t_1(a))}, &t_1(a)<t_2(a),\\
1, &t_1(a)=t_2(a).
\end{cases}
\]
Thus, $\alpha_\pi\neq \alpha_{\pi'}$. Define $\hat a=\sup \{a\in A:\alpha_\pi([0,a])\neq\alpha_{\pi'}([0,a])\}\in A_\Gamma$, where the inclusion follows from $\alpha_\pi\neq \alpha_{\pi'}$ and $\alpha_\pi(A_\Gamma)=\alpha_{\pi'}(A_\Gamma)=1$. As shown above, there exists $\varepsilon>0$ such that, for all $\tilde a_1,\tilde a_2\in [\hat a-\varepsilon,\hat a]\cap A_\Gamma$ with $t_1(\tilde a_1)\neq t_2(\tilde a_1)$ and $t_1(\tilde a_2)\neq t_2(\tilde a_2)$, we have $t_1(\tilde a_1)<t_2(\tilde a_2)$. We will now show that $\alpha_\pi([0,\tilde a])=\alpha_{\pi'}([0,\tilde a])$ for all $\tilde a\in [\hat a-\varepsilon,\hat a]$ contradicting the definition of $\hat a$. 

By (P1), the marginals of $\pi$ and $\pi'$ on $\Theta$ are both equal to $\phi$. Since $t_2$ is increasing in $a$, states $\theta>t_2(\tilde a)$ can only induce actions $a>\tilde a$.  Thus, since $\alpha_{\pi'}([0,a])=\alpha_\pi([0,a])$ for all $a\geq \hat a$, and since $t_1(\tilde a_1)<t_2(\tilde a_2)$ for all $\tilde a_1,\tilde a_2\in [\hat a-\varepsilon,\hat a]\cap A_\Gamma$ with $t_1(\tilde a_1)\neq t_2(\tilde a_1)$ and $t_1(\tilde a_2)\neq t_2(\tilde a_2)$, it follows that, for all $\tilde a\in [\hat a-\varepsilon,\hat a]\cap A_\Gamma$, we have
\begin{align*}
\phi((t_2(\tilde a),1])-\phi([t_2(\tilde a),1])&\leq  \int_{[\tilde a,\hat a]} (1-\rho_a)\df \alpha_{\pi'}(a) -\int_{[\tilde a,\hat a]} (1-\rho_a)\df \alpha_{\pi}(a)\\ &\leq \phi([t_2(\tilde a),1])-\phi((t_2(\tilde a),1]).
\end{align*}
Moreover, since $\phi$ has a density, we have $\phi((t_2(\tilde a),1])=\phi([t_2(\tilde a),1])$, and hence
\[
\int_{[\tilde a,\hat a]} (1-\rho_a)\df \alpha_{\pi'}(a) =\int_{[\tilde a,\hat a]} (1-\rho_a)\df \alpha_{\pi}(a).
\]
Then, since $1-\rho_a>0$ for all $a\in A_\Gamma$, and since $\supp (\alpha_{\pi'})\subset A_\Gamma$ and $\supp (\alpha_{\pi})\subset A_\Gamma$, it follows that $\alpha_\pi([\tilde a,\hat a])=\alpha_\pi([\tilde a,\hat a])$ for all $\tilde a\in [\hat a-\varepsilon,\hat a]$. Thus, since  $\alpha_{\pi'}([0,a])=\alpha_\pi([0,a])$ for all $a\geq \hat a$, it follows that $\alpha_\pi([0,\tilde a])=\alpha_{\pi'}([0,\tilde a])$ for all $\tilde a\in [\hat a-\varepsilon,\hat a]$.

\subsection{Proof of Theorem \ref{t:FDu}} By $\Theta=[0,1]$ and Assumptions \ref{a:smooth}--\ref{a:sc}, $\theta^\star (a)$ is a strictly increasing, continuous function from $A$ onto $\ol \Theta=\Theta$. Since the range of $\theta^\star$ is $\Theta$ and full disclosure is optimal, Theorem \ref{t:contact} implies that $\theta^\star (a)\in \Gamma^\star_a$ for all $a$. Thus, since the contact set is pairwise (i.e., $|\Gamma^\star_a|\leq 2$) and $\min \Gamma^\star_a<\theta^\star(a)<\max \Gamma^\star_a$ whenever $\Gamma^\star_a$ is multivalued (by the definition of $\Gamma^\star$), it follows that $\Gamma^\star_a=\{\theta^\star (a)\}$ for all $a$, as otherwise $\min \Gamma^\star_a$, $\theta^\star(a)$, and $\max \Gamma^\star_a$ would be three distinct elements in $\Gamma^\star_a$. Hence, $\Gamma^\star=\cup _{\theta\in \Theta}(a^\star(\delta_\theta),\theta)$, so full disclosure is optimal.

\subsection{Proof of Theorem \ref{t:NAD}}

We give the proof for the single-dipped case. Since for all $\theta_1<\theta_2$ there exists $p\in (0,1)$ such that \eqref{e:nd} holds, it follows that there do not exist $\theta_1<\theta_2$ such that $(a^\star(\delta_{\theta_1}),\theta_1)$ and $(a^\star(\delta_{\theta_2}),\theta_2)$ are in $\Gamma$. Suppose by contradiction that such $\theta_1$ and $\theta_2$ exist. For any $\mu=\rho  \delta_{\theta_1}+(1-\rho )\delta_{\theta_2}$ with $\rho \in (0,1)$, we have 
\begin{align*}
	p(\theta_1)&=V(a^\star(\theta_1),\theta_1)\geq V(a^\star(\mu),\theta_1)+q(a^\star(\mu))u(a^\star(\mu),\theta_1),\\
	p(\theta_2)&=V(a^\star(\theta_2),\theta_2)\geq V(a^\star(\mu),\theta_2)+q(a^\star(\mu))u(a^\star(\mu),\theta_2),
\end{align*}
by (D1) and the definition of $\Gamma$. Adding the first inequality multiplied by $\rho $ and the second inequality multiplied by $1-\rho $, we obtain that \eqref{e:nd} fails for all $\rho \in (0,1)$, yielding a contradiction.

Since $\Theta=[0,1]$, $\Gamma^*$ is strictly single-dipped, and for all $\theta_1<\theta_2$ there exists $p\in (0,1)$ such that \eqref{e:nd} holds, it follows that $ t_1(a_2)\leq t_1(a_1)$ for all $a_1<a_2$ in $A_\Gamma$, and thus $\Gamma^\star $ is single-dipped negative assortative. (Recall that, by Theorem \ref{t:brenier}, $\Gamma^\star(a)=\{t_1(a),t_2(a)\}$ for all $a\in A_\Gamma$ where $t_2(a)$ is increasing in $a$.) Suppose by contradiction that there exist $a_1<a_2$ in $A_\Gamma$ such that $t_1(a_2)>t_1(a_1)$. Then $t_1(a_2)\geq t_2(a_1)$, as otherwise $(a_1,t_1(a_1))$, $(a_2,t_1(a_2))$, $(a_1,t_2(a_1))$ is a strictly single-peaked triple in $\Gamma^\star$. Define
\begin{gather*}
\ul a_i=\inf \{a\in A_\Gamma:t_1(a_i)\leq t_1(a)\leq t_2(a)\leq t_2(a_i)\}\leq a_i, \quad \text{for $i=1,2$}.
\end{gather*}
Since $A_{\Gamma^\star}=A_\Gamma$ and $A_\Gamma$ is compact, we have $\ul a_1,\ul a_2\in A_{\Gamma^\star}$. We claim that $\Gamma^\star_{\ul a_i}=\{\theta^\star(\ul a_i)\}$ for $i=1,2$. Suppose by contradiction that $\Gamma_{\ul a_i}^\star\neq \{\theta^\star(\ul a_i)\}$, so $\Gamma_{\ul a_i}^\star=\{t_1(\ul a_i),t_2(\ul a_i)\}$ with $t_1(\ul a_i)<\theta^\star (\ul a_i)<t_2 (\ul a_i)$. Let $\Theta_{\Gamma^\star}$ be the projection of $\Gamma^\star$ on $\Theta$. Since $\pi(\Gamma^\star)=1$ for an optimal $\pi$, we have $\phi(\Theta_{\Gamma^\star})=1$ by (P1), and the closure of $\Theta_{\Gamma^\star}$ is $\Theta=[0,1]$. Thus, there exists $(a,\theta)\in \Gamma^\star$ with $t_1(\ul a_i)<\theta<t_2 (\ul a_i)$. Since $\Gamma^\star$ is strictly single-dipped, it follows that $a<\ul a$ (otherwise $(\ul a,t_1(\ul a))$, $(a,\theta)$, $(\ul a,t_2 (\ul a))$ is a single-peaked triple in $\Gamma^\star$) and $t_1(\ul a_i)\leq t_1(a)\leq t_2(a)\leq t_2(\ul a_i)$ (otherwise either $(a,t_1(a))$, $(\ul a_i,t_1(\ul a_i)$, $(a,\theta)$ or $(a,\theta)$, $(\ul a_i,t_2(\ul a_i)$, $(a,t_2(a))$ is a strictly single-peaked triple in $\Gamma^\star$), contradicting the definition of $\ul a$. Hence, $(\ul a_1,\theta^\star (\ul a_1))$ and $(\ul a_2,\theta^\star (\ul a_2))$ are in $\Gamma$, so by the second step of the proof we must have $\ul a_1 = \ul a_2$. But, by construction, $t_1(a_1)\leq \theta^\star (\ul a_1)\leq t_2(a_1)\leq t_1(a_2)\leq \theta^\star (\ul a_2)\leq t_2(a_2)$, and $\ul a_1 = \ul a_2$ implies that these inequalities all hold with equality, contradicting $t_1(a_2)>t_1(a_1)$.

Now suppose that $\phi$ has a density $f$ and $\Gamma^\star $ is single-dipped negative assortative. Finally, we show that the functions $t_1$ and $t_2$ are continuous and satisfy the differential equations \eqref{e:obed}--\eqref{e:q'} and the boundary condition \eqref{e:boundary}. Since the closure of the projection $\Theta_{\Gamma^\star}$ of $\Gamma^\star $ on $\Theta$ is $\Theta$, it follows that the the closure of the image of the functions $t_1$ and $t_2$ must also be equal to $\Theta$. 
Since $t_1$ is decreasing and $t_2$ is increasing on the compact domain $A_\Gamma$, and since $t_1(a)\leq \theta^\star (a)\leq t_2(a)$ for all $a\in A_\Gamma$, it follows that $t_1$ and $t_2$ are continuous functions such that $t_1(\ul a)=\theta^\star (a)=t_2(\ul a)$, $t_1(a)<\theta^\star (a)<t_2(a)$ for all $a>\ul a$, $t_1(\ol a)=0$, $t_2(\ol a)=1$, and $(t_1(\ul b_i),t_2(\ul b_i)) =(t_1(\ol b_i), t_2(\ol b_i))$ for all $i$, where $\{(\ul b_i,\ol b_i)\}_i$ is an at most countable set of disjoint open intervals comprising the set $[\ul a,\ol a]\setminus A_\Gamma$. Since $\phi $ has a density, the measure of the endpoints of these intervals is zero, and hence the set of optimal outcomes is unaffected if we redefine $A_\Gamma$ as $[\ul a,\ol a]$ and extend the domain of $t_1$ and $t_2$ to $[\ul a,\ol a]$ by setting $t_1(a)=t_1(\ul b_i)=t_1(\ol b_i)$ and $t_2(a)=t_2(\ul b_i)=t_2(\ol b_i)$ for all $a\in (\ul b_i,\ol b_i)$. In sum, without loss of generality, we can assume that $t_1$ and $t_2$ are continuous monotone functions defined on $[\ul a,\ol a]$ that satisfy \eqref{e:boundary} and $t_1(a)<\theta^\star (a)<t_2(a)$ for all $a\in (\ul a,\ol a]$.

Since $\phi$ has a density and $\Gamma^\star_a=\{t_1(a),t_2(a)\}$ for all $a\in [\ul a,\ol a]$, where $t_1$ is continuously decreasing and $t_2$ is continuously increasing, we can rewrite (P2) for $\tilde A=[a,a']$, with $\ul a\leq a< a'\leq \ol a$, as
\[
\int_{a}^{a'} u(\tilde a,t_1(\tilde a))(-\df \phi ([0,t_1(\tilde a)]))+\int _{a}^{a'} u(\tilde a,t_2(\tilde a))\df \phi ([0,t_2(\tilde a)])=0.
\]
Taking the limit $a'\downarrow a$, we obtain  \eqref{e:obed} for all $a\in [\ul a,\ol a]$.

Since $\Gamma_a^\star=\{t_1(a),t_2(a)\}$ for all $a\in [\ul a,\ol a]$, Theorem \ref{t:contact} gives the FOC, for all $a\in (\ul a,\ol a]$,
\begin{gather*}
v(a,t_1(a))+ q(a)u_a(a,t_1(a))+q'(a)u(a,t_1(a))=0,\\
v(a,t_2(a))+ q(a)u_a(a,t_2(a))+q'(a)u(a,t_2(a))=0.
\end{gather*}
Solving for $q(a)$ and $q'(a)$, we get, for all $a\in (\ul a,\ol a]$,
\begin{gather*}
q(a)= \frac{v(a,t_1(a))u(a,t_2(a))-v(a,t_2(a))u(a,t_1(a))}{u(a,t_1(a))u_a(a,t_2(a))-u(a,t_2(a))u_a(a,t_1(a))},\\
q'(a)=\frac{v(a,t_1(a))u_a(a,t_2(a))-v(a,t_2(a))u_a(a,t_1(a))}{u_a(a,t_1(a))u(a,t_2(a))-u_a(a,t_2(a))u(a,t_1(a))},
\end{gather*}
where the denominators in the expressions for $q(a)$ and $q'(a)$ are not equal to $0$, by Assumption \ref{a:qc}.
Recalling that $q'$ is the derivative of $q$, we obtain \eqref{e:q'} for all $a\in (\ul a,\ol a]$.

\subsection{Proof of Corollary \ref{c:ndSDD}}
We give the proof for the single-dipped case. Noting that $\rho u(a^\star (\mu),\theta_1)+(1-\rho )u(a^\star(\mu),\theta_2)=0$ and denoting $a=a^\star (\mu)$, we infer that \eqref{e:nd} fails if there exist $\theta_1<\theta_2$ such that for all $a\in (a^\star (\delta_{\theta_1}),a^\star (\delta_{\theta_2}))$, we have
\[
u(a,\theta_2)(V(a,\theta_1)-V(a^\star(\delta_{\theta_1}),\theta_1))-u(a,\theta_1)(V(a,\theta_2)-V(a^\star(\delta_{\theta_2}),\theta_2))\leq 0.
\]
By Taylor's theorem and some algebra, we get
\begin{gather*}
u(a,\theta_2)(V(a,\theta_1)-V(a^\star(\delta_{\theta_1}),\theta_1))-u(a,\theta_1)(V(a,\theta_2)-V(a^\star(\delta_{\theta_2}),\theta_2))\\
=\frac 12u_a(a,\theta^\star(a))\left( v_a(a,\theta^\star (a)) - \frac{v(a,\theta^\star (a)) u_{aa}(a,\theta^\star (a))}{u_a(a,\theta^\star (a))} \right.\\
\left.-2\frac{v_\theta(a,\theta^\star (a)) u_a(a,\theta^\star (a)) - v(a,\theta^\star (a)) u_{a\theta}(a,\theta^\star (a))}{u_\theta(a,\theta^\star (a))}\right)\\
\cdot (a-a^\star(\delta_{\theta_1}))(a^\star(\delta_{\theta_2})-a)(a^\star(\delta_{\theta_2})-a^\star(\delta_{\theta_1})) \\
+o((a-a^\star(\delta_{\theta_1}))(a^\star(\delta_{\theta_2})-a)(a^\star(\delta_{\theta_2})-a^\star(\delta_{\theta_1}))).
\end{gather*}
Hence, if \eqref{e:ndSDD} fails at some $a$, then there exist $\theta_2>\theta_1$ with $a^\star(\delta_{\theta_2})-a>0$ and $a-a^\star(\delta_{\theta_1})>0$ small enough such that \eqref{e:nd} fails for all $\rho \in (0,1)$.

Note that ${\df\theta^\star (a)}/{\df a}=-{u_a(a,\theta^\star(a))}/{u_\theta(a,\theta^\star(a))}$, by the implicit function theorem applied to $u(a,\theta^\star (a))=0$.
Thus, denoting the partial derivatives of $v$ and $u_a$ in $a$ by $v_a$ and $u_{aa}$, we get that the derivative of $q(a)=-{v(a,\theta^\star(a))}/{u_a(a,\theta^\star(a))}$ is given by
\[
q'(a)=-\frac{v_a(a,\theta^\star(a))}{u_a(a,\theta^\star(a))}+\frac{v_\theta(a,\theta^\star(a))}{u_\theta(a,\theta^\star(a))}+\frac{v(a,\theta^\star(a)) u_{aa}(a,\theta^\star(a))}{(u_a(a,\theta^\star(a)))^2}-\frac{v(a,\theta^\star(a)) u_{a\theta}(a,\theta^\star(a))}{u_a(a,\theta^\star(a))u_\theta(a,\theta^\star(a))}.
\]

Conversely, suppose that \eqref{e:ndSDD}, together with all other assumptions of the corollary, holds. Then, for $a>a^\star (\delta_\theta)$, we have
\begin{align*}
	&V(a,\theta)-\frac{v(a,\theta^\star(a))}{u_a(a,\theta^\star(a))}u(a,\theta)-V(a^\star(\delta_\theta),\theta)\\
	=&(V(\tilde a,\theta) +q(\tilde a)(\theta-\tilde a))|^{a}_{a^\star(\delta_\theta)}\\
	=&\int^{a}_{a^\star(\delta_\theta)} [v(\tilde a,\theta)+q(\tilde a)u_a(\tilde a,\theta)+q'(\tilde a)u(\tilde a,\theta)]\df \tilde a\\
	\geq &\int^{a}_{a^\star(\delta_\theta)} \left [v(\tilde a,\theta)-\frac{v(\tilde a,\theta^\star(\tilde a))}{u_a(\tilde a,\theta^\star(\tilde a))} u_a(\tilde a,\theta)\right ]\df \tilde a\\
	&+\int^{a}_{a^\star(\delta_\theta)} \left [\frac{v(\tilde a,\theta^\star (\tilde a)) u_{a\theta}(\tilde a,\theta^\star (\tilde a))}{u_\theta(\tilde a,\theta^\star (\tilde a))} - \frac{v_\theta (\tilde a,\theta^\star (\tilde a))}{u_\theta (\tilde a,\theta^\star (\tilde a))}\right ]u(\tilde a,\theta)\df \tilde a\\
	=&\int^{a}_{a^\star(\delta_\theta)}\int ^{\theta^\star (\tilde a)}_{\theta} \left [\frac{v(\tilde a,\theta^\star(\tilde a))}{u_a(\tilde a,\theta^\star(\tilde a))} u_{a\theta}(\tilde a,\tilde \theta)-v_\theta (\tilde a,\tilde \theta)\right ]\df \tilde \theta \df \tilde a\\
	&+\int^{a}_{a^\star(\delta_\theta)}\int ^{\theta^\star (\tilde a)}_{\theta} \left [ \frac{v_\theta (\tilde a,\theta^\star (\tilde a))}{u_\theta (\tilde a,\theta^\star (\tilde a))} - \frac{v(\tilde a,\theta^\star (\tilde a)) u_{a\theta}(\tilde a,\theta^\star (\tilde a))}{u_\theta(\tilde a,\theta^\star (\tilde a))} \right ]u_\theta(\tilde a,\tilde \theta)\df \tilde \theta \df \tilde a\\
	= &\int^{a}_{a^\star(\delta_\theta)}\int ^{\theta^\star (\tilde a)}_{\theta} \left [\frac{v_\theta (\tilde a,\theta^\star (\tilde a))}{u_\theta (\tilde a,\theta^\star (\tilde a))} -\frac{v_\theta (\tilde a,\tilde \theta)}{u_\theta(\tilde a,\tilde \theta)}\right ]u_\theta(\tilde a,\tilde \theta)\df \tilde \theta \df \tilde a\\
	&+\int^{a}_{a^\star(\delta_\theta)}\int ^{\theta^\star (\tilde a)}_{\theta} \frac{v(\tilde a,\theta^\star(\tilde a))}{-u_a(\tilde a,\theta^\star(\tilde a))} \left [\frac{u_{a\theta}(\tilde a,\theta^\star(\tilde a))}{u_\theta (\tilde a,\theta^\star(\tilde a))}-\frac{u_{a\theta}(\tilde a,\tilde \theta)}{u_\theta (\tilde a,\tilde \theta)} \right]u_\theta(\tilde a,\tilde \theta)\df \tilde \theta \df \tilde a>0,
\end{align*}
where the first and last equalities are by rearrangement, the second and third equalities are by the fundamental theorem of calculus, the first inequality is by \eqref{e:ndSDD} and substitution of $q(\tilde a)$ and $q'(\tilde a)$, and the last inequality is by our assumptions imposed in the corollary.

By Taylor's theorem, we have, for $\theta_1<\theta_2$ and $a\in (a^\star(\delta_{\theta_1}),a^\star(\delta_{\theta_2}))$,
\begin{gather*}
u(a,\theta_2)(V(a,\theta_1)-V(a^\star(\delta_{\theta_1}),\theta_1))-u(a,\theta_1)(V(a,\theta_2)-V(a^\star(\delta_{\theta_2}),\theta_2))\\
=\left[V(a^\star(\theta_2),\theta_1)-\frac{v(a^\star(\theta_2),\theta_2)}{u_a(a^\star(\theta_2),\theta_2)}u(a^\star(\theta_2),\theta_1)-V(a^\star(\delta_{\theta_1}),\theta_1)\right]\\
\cdot(-u_a(a^\star (\delta_{\theta_2}),\theta_2))(a^\star(\delta_{\theta_2})-a)+o(a^\star(\delta_{\theta_2})-a).
\end{gather*}
Hence \eqref{e:nd} holds for sufficiently small $\rho >0$.

\newpage

\begin{center}
\Large{Online Appendix}
\end{center}

\renewcommand{\thesection}{C}

\section{Applications and Extensions} \label{s:other}

\subsection{Contests}\label{s:ZZ}\mbox{} \citet{ZZ} study information disclosure in contests. In their model, two contestants, $A$ and $B$, compete for a prize by exerting efforts $x_A$ and $x_B$. The probability that contestant $i=A,B$ wins is $x_i/(x_A+x_B)$. Everyone knows contestant A's value $v_A=1$. Contestant B's value $v_B$ is known to contestant B and the designer. The sender designs a signal about $v_B$ to maximize expected total effort.

It is convenient to parameterize $\theta = 1/\sqrt{v_{B}}$ and $a=\sqrt{x_{A}}$. With this parameterization, \citeauthor{ZZ}'s Proposition 1 shows that, given a posterior $\mu$, contestant A exerts effort $x_A^\star=a^\star (\mu)^2$ determined by $\E_{\mu} \left [\theta-\left(1+\theta^2 \right)a^\star (\mu) \right] =0$, and contestant B (who knows $\theta$) exerts effort $x^\star(\theta)=a^\star (\mu)/\theta-a^\star (\mu)^2$, so the sender's expected utility is $x_A^\star+\E_\mu \left[x^\star(v_B)\right]=\E_\mu \left[a^\star (\mu)/\theta \right]$. We thus recover our model with $V(a,\theta)=a/\theta$ and $u(a,\theta)=\theta-(1+\theta^2)a$. 

\citeauthor{ZZ} give results on optimality of pairwise disclosure, full-disclosure, and no-disclosure. Our approach easily yields the following result, which additionally gives conditions for optimality of single-dipped/-peaked disclosure and negative assortative disclosure (which were not considered by \citeauthor{ZZ}).\footnote{\citeauthor{ZZ} assume that $\phi$ is discrete; we instead assume that $\phi$ is continuous.}

\begin{proposition}\label{p:ZZ} Let $\phi$ have a density on $\Theta=[\ul \theta ,\ol \theta]$, where $0<\ul \theta<\ol \theta$. If $\ul \theta\geq 1$, then the unique optimal outcome is full disclosure. If  $\ol \theta \leq 1/\sqrt{3}$ ($1/\sqrt{3}\leq \ul \theta<\ul \theta\leq 1$), then the unique optimal outcome is single-dipped (-peaked) negative assortative disclosure.
\end{proposition}

The proof of single-dippedness/-peakedness uses Theorem \ref{l:ssdd} with a perturbation that fixes both actions. In contrast, directly applying Theorem \ref{t:SDPD} would yield only the weaker result that single-peaked negative assortative disclosure is optimal if $1/\sqrt 2\leq \ul \theta<\ol \theta< 1$.\footnote{To see this, suppose $\ol \theta< 1$. Then $u_\theta(a,\theta)=1-2\theta a>0$ for $a\leq \ol \theta/(1+\ol \theta\,\!^2)=\max A$. Moreover, $u_{a\theta}(a,\theta)/u_\theta(a,\theta)=-2\theta/(1-2\theta a)$ is always decreasing in $\theta$, while $v_\theta(a_2,\theta)/u_\theta(a_1,\theta)=-1/(\theta^2-2\theta^3 a_1)$ is decreasing in $\theta $ iff $3\ul \theta \min A= 3\ul \theta^2/(1+\ul\theta\,\!^2)\geq  1$, or equivalently $\ul \theta \geq 1/\sqrt 2$.}

\subsection{Affiliated Information}\label{s:GS}\mbox{} \citet{GS2018} consider a persuasion model with a privately informed receiver, where it is commonly known that the receiver wishes to accept a proposal iff $\theta$ exceeds a threshold $\theta_0$, and the receiver's type $t$ is his private signal of $\theta$. Letting $G(t|\theta)$ denote the distribution of $t$ conditional on $\theta$, with corresponding density $g(t|\theta)$, this setup maps to our model with $V(a,\theta)=G(a|\theta)$, $u(a,\theta)=(\theta-\theta_{0})g(a|\theta)$, and $g(t|\theta)$ strictly log-submodular in $(t,\theta)$.\footnote{The ordering convention here is that high $t$ is bad news about $\theta$. This ordering is opposite to \citeauthor{GS2018}'s, but follows our convention that the receiver accepts for types below a cutoff.}$^,$\footnote{\citeasnoun{IP20} study robust stress test design in a setting with multiple receivers with coordination motives. As they note, the single-receiver version of their model is a special case of \citeasnoun{GS2018}.} These preferences satisfy
Assumptions \ref{a:smooth}, \ref{a:qc} (see Lemma \ref{l:ASC}), \ref{a:sc}, and \ref{a:v>0}, but not Assumption \ref{a:int}, as $u(a,\theta)>0$ for all $a$ when $\theta>\theta_0$. Nonetheless, assuming that the receiver breaks ties in the sender's favor, we have $a^\star (\mu)=\max\{a:\int_\Theta u(a,\theta)\df \mu\geq 0\}$.

Let us take for granted that Theorem \ref{l:ssdd} holds even though Assumption \ref{a:int} is violated (e.g., this is clearly true if $\Theta$ is finite). Applying Theorem \ref{l:ssdd} with a perturbation that fixes one action while increasing the other action and the sender's expected utility (for fixed actions), we obtain the following result, which reproduces \citeauthor{GS2018}'s main qualitative insight.%

\begin{proposition}\label{c:spd1}
Every optimal outcome is single-peaked.
\end{proposition}

Notice that when Assumption \ref{a:int} fails, condition \eqref{e:nd} cannot hold for all $\mu$, because there exist states $\theta_1\neq \theta_2$ such that either (i) $u(a,\theta_1)>0$ and $u(a,\theta_2)>0$ for all $a$, so that $a^\star(\rho\delta_{\theta_1}+(1-\rho)\delta_{\theta_2})=1$ for all $\rho\in [0,1]$, or (ii) $u(a,\theta_1)<0$ and $u(a,\theta_2)<0$ for all $a$, so that $a^\star (\rho \delta_{\theta_1}+(1-\rho)\delta_{\theta_2})=0$ for all $\rho\in [0,1]$. In both cases, we obviously have, for all $\rho\in [0,1]$,
\[
\rho V(a ^{\star }(\mu),\theta_1) +(1-\rho )V(a ^{\star }(\mu),\theta_2)= \rho V(a ^{\star }(\delta_{\theta_1}),\theta_1) +(1-\rho )V(a ^{\star }(\delta_{\theta_2}),\theta_2),
\] 
so \eqref{e:nd} necessarily fails. This suggests the following adjusted requirement when Assumption \ref{a:int} fails: for all $\theta_1,\theta_2$ with $u(1,\theta_1)<0<u(0,\theta_2)$, condition \eqref{e:nd} holds for some $\rho\in(0,1)$. This requirement is clearly satisfied in \citet{GS2018}, as then $\theta_1<\theta_0<\theta_2$, so \eqref{e:nd} holds for $\rho$ sufficiently small so that $a^\star (\rho \delta_{\theta_1}+(1-\rho)\delta_{\theta_2})=1$. For the case where $\phi$ has a density on $\Theta=[0,1]$, Theorem 3.1 in \citet{GS2018} implies that the optimal outcome is single-peaked negative assortative, in the sense that there exist an increasing function $t_1(a)$ and a decreasing function $t_2(a)$ such that $t_1(a)\leq\theta_0\leq t_2(a)$ and $\supp (\pi_a)=\{t_1(a),t_2(a)\}$ for all $a>0$, and $\supp(\pi_0)=[0,t_1(0)]$. %

\subsection{Stress Tests}\label{s:GL}\mbox{}
\citeasnoun{GL} consider a model of optimal stress tests. The sender is a bank regulator and the receiver is a perfectly competitive market. The bank has an asset that yields a random cash flow. The asset's quality is $\theta$, which is observed by the bank and the regulator but not the market, and is normalized to equal the asset's expected cash flow.\footnote{This is the model in Section 5 of their paper, where the bank observes $\theta$.} The regulator designs a test to reveal information about $\theta$. After observing the test result, the market offers a competitive price $a$ for the asset. Finally, the bank decides whether to keep the asset and receive the random cash flow, or sell it at price $a$. Letting $z$ denote the bank's final cash holding (equal to either the random cash flow or $a$), the bank's payoff equals $z+\1 \{z \geq \theta_{0}\}$, where $\theta_{0}$ is a constant. An interpretation is that the bank faces a run if its cash holding falls below $\theta_{0}$. The regulator designs the test to maximize expected social welfare, or equivalently to minimize the probability of a run.

\citeauthor{GL} show that a bank with a type-$\theta$ asset is willing to sell at a price $a$ iff $a$ exceeds a reservation price $\tilde{\sigma} (\theta)$ that satisfies $\tilde{\sigma} (\theta) >\theta $ if $\theta < \theta_0$, $\tilde{\sigma} (\theta) <\theta $ if $\theta > \theta_0$, and $\tilde{\sigma}' (\theta)\geq 0$. Intuitively, if $\theta < \theta_0$ then the bank demands a premium to forego the chance that a lucky cash flow shock pushes its holdings above $\theta_0$, while if $\theta > \theta_0$ then the bank desires insurance against bad cash flow shocks that push its holdings below $\theta_0$. However, the value of the regulator's problem is unaffected if the reservation price is re-defined as $\sigma (\theta) =\theta$ if $\theta \leq \theta_0$ and $\sigma (\theta) =\tilde{\sigma} (\theta)$ if $\theta > \theta_0$, because it is suboptimal for the regulator to induce a bank to sell at a price below $\theta_0$. It is more convenient to work with the normalized reservation price $\sigma (\theta)$. 

It is also convenient to restrict attention to tests that, for each $\theta$, either induce the bank to sell or fully disclose the bank's value: this is without loss because if the regulator pools two asset types that do not sell, then it would be weakly better to disclose these types. Note that for such a test, the price induced by any posterior $\mu$ is $a^\star (\mu)=\mathbb{E}_{\mu}[\theta]$, so we are in the linear receiver case. We can capture the requirement that the bank always sells if $a\neq \theta$ by setting $V(a,\theta)=- \infty$ if $a<\sigma(\theta)$. Finally, letting $w(\theta)>0$ equal the social gain when a bank sells a type-$\theta$ asset at a price above $\theta_0$ (which equals the probability that a type-$\theta$ asset yields a cash flow below $\theta_0$), we obtain the linear receiver case of our model with
\[
V(a,\theta)=
\begin{cases}
w(\theta)\1\{a\geq \theta_0\}, &\text{if $a\geq \sigma(\theta)$,}\\
-\infty, &\text{otherwise.}
\end{cases}
\]
Note that $V$ violates Assumptions \ref{a:smooth} and \ref{a:v>0}, as it is discontinuous and only weakly increasing in $a$. Nonetheless, if we assume that $\Theta$ is a finite set (as do \citeauthor{GL}), we recover their main qualitative insight.

\begin{proposition}\label{p:GL} 
Let $\Theta$ be finite. There exists an optimal single-dipped outcome.
\end{proposition}
To prove the proposition, we use a perturbation that fixes both actions. Since $V$ is only weakly increasing, this perturbation now only weakly increases the sender's expected utility. Nonetheless, when $\Theta$ is finite, repeatedly apply such perturbations eventually yields a single-dipped outcome. We also note that, as \citeauthor{GL} show, if $\mathbb{E}_{\phi}[\theta]<\theta_{0}$---so that no-disclosure does not attain the sender's first-best outcome---then every optimal outcome is single-dipped.\footnote{A related model by \citet{GT} studies optimal information disclosure to facilitate trade in an insurance market with adverse selection. Their model can be mapped to the linear receiver case with $V(a,\theta)=\nu(a)$ if $a\geq \sigma (\theta)$ and $V(a,\theta)=-\infty$ otherwise, where $\nu(a)$ is a strictly increasing, strictly concave function, and $\sigma$ is a continuous, strictly increasing function that satisfies $\sigma(\theta)<\theta$. Considering a similar perturbation as in \citeauthor{GL} shows that single-dipped negative assortative disclosure is optimal in their model. We also mention \citet{LW}, where a bank regulator discloses information about the design of a stress test to induce banks to make socially desirable investments. In this model, single-peaked disclosure is optimal.}

\renewcommand{\thesection}{D}
\section{Additional Examples}
\label{s:examples}

\begin{example}[$\Gamma^\star$ might not be compact; with the ``wrong'' selection from $Q$, $\Gamma$ might not be compact either.] \label{e:1}
Consider the linear case with $V(a)=0$ if $a<1/2$ and $V(a)=(a-1/2)^2$ otherwise. Let $\phi $ be uniform on $\Theta=\{0,1/2,1\}$. Note that $p(\theta)=V(\theta,\theta)$ solves (D). Moreover, $Q(a)=0$ if $a<1/2$ and $Q(a)=[a-1/2,a]$ otherwise. Our selection from $Q$ is given by $q(a)=0$ if $a<1/2$ and $q(a)=2a-1$ otherwise. Note that this selection is from the interior of $Q(a)$ for all $a\in (1/2,1)$.

With our selection, the contact set  $\Gamma = ([0,1/2]\times \{0,1/2\})\cup \{1,1\}$ is compact, but $\Gamma^\star =\Gamma\setminus (\{0,1/2\}\cup\{1/2,0\})$ is not compact. Note also that there exists an optimal outcome with $\supp(\pi)=\Gamma$ (e.g., the outcome that induces action $1$ with certainty if $\theta=1$, and induces action $a\in [0,1/2]$ with densities $4-8a$ and $8a$ if $\theta=0$ and $\theta=1/2$, respectively.) However, for any such an outcome there exists a conditional probability $\pi_a$ such that $\supp (\pi_a)=\Gamma^\star _a$ for all $a\in A_\Gamma=[0,1/2]\cup \{1\}$ (i.e., $\pi_1=\delta_1$ and $\pi_a=(1-2a)\delta_0 +2a\delta_{1/2}$ for all $a\in [0,1/2]$).

In contrast, consider an alternative selection from $Q$ given by $\tilde q(a)=0$ if $a<1/2$ and $\tilde q(a)=a$ otherwise. The associated contact set $\tilde \Gamma = \Gamma\cup ([1/2,1)\times \{1\})\setminus \{(1/2,0)\}$ is not compact because $(1/2,0)\notin \tilde \Gamma$, and $A_{\tilde\Gamma}$ contains redundant actions $a\in (1/2,1)$ that are not induced by any optimal outcome.
\end{example}

\begin{example}[The FOC \eqref{e:FOC} might not hold on all of $\Gamma$.]\label{ex:FOC}
Consider the linear receiver case. Let $\phi$ be uniform on $\Theta=\{0,1/3,1\}$, and $V(a,\theta)=-a^2$ if $\theta=0$ and $V(a,\theta)=-a/3+a^2-3a^3/4$ if $\theta\in \{1/3,1\}$. Since $V(a,\theta)\leq 0$ for all $(a,\theta)$ with equality on $\Gamma=\{0,0\}\cup (\{0,2/3\}\times \{1/3,1\})$ and strict inequality elsewhere, the unique optimal outcome reveals state $0$ (which induces action $0$) and pools states $1/3$ and $1$ (which induces action $2/3$). The contact set is $\Gamma$, so $\Gamma_0=\Theta$. But \eqref{e:FOC} cannot hold on $\Gamma$, because the following system of equations does not have a solution $(q(0),q'(0))$,
\[
\begin{cases}
0-q(0)+q'(0)0=0,\\
-\frac 13 -q(0)+q'(0)\frac 13 =0,\\
-\frac 13 -q(0)+q'(0)1 =0.
\end{cases}
\]
Intuitively, $\theta^\star (a)\in (\min \Gamma_a,\max \Gamma_a)$ is an interior case, so the FOC is valid on $\Gamma^\star_a=\Gamma_a$; while $\theta^\star (a)\in \{\min \Gamma_a,\max \Gamma_a\}$ is a boundary case, so the FOC may be invalid on $\Gamma_a$, but it is still valid on $\Gamma^\star_a=\{\theta^\star (a)\}$ given our selection $q(a)=v(a,\theta^\star(a))/(-u_a(a,\theta^\star(a))$.
\end{example}

\begin{example}[Without Assumption \ref{a:sc}, $\Gamma$ might not be compact and the FOC might fail.]\label{ex:SC}
Let $\phi$ be uniform on $\Theta=\{0,1/3,2/3,1\}$; $u(a,0)=-a$, $u(a,1/3)=u(a,2/3)=1/2-a$, and $u(a,1)=1-a$; and $V(a,0)=V(a,1/3)=0$, $V(a,2/3)=a-1/2$, and $V(a,1)=a-1$. Note that $p=0$ solves (D). Moreover, $Q(a)=0$ if $a<1/2$, $Q(a)=1$ if $a>1/2$, and $Q(1/2)=[0,1]$. 

For any selection $\tilde q$ from $Q$, the associated contact set $\tilde \Gamma$ satisfies $\tilde \Gamma_a=\{0,1/3\}$ if $a<1/2$, $\tilde \Gamma_a=\{2/3,1\} $ if $a>1/2$, and $\tilde \Gamma_{1/2}=\{1/3,2/3\}\cup (\{\tilde q(1/2)\}\cap \Theta)$. The set $\tilde \Gamma$ is not compact because $(1/2,0)\notin \tilde \Gamma$ if $\tilde q(1/2)\neq 0$ and $(1/2,1)\notin \tilde \Gamma$ if $\tilde q(1/2)\neq 1$. Moreover, there does not exist a full measure set where the FOC holds: since the full-disclosure outcome $\pi=(\delta_{(0,0)}+ \delta_{(1/2,1/3)}+\delta_{(1/2,2/3)}+\delta_{(1,1)})/4$ is supported on $\tilde \Gamma$, it is optimal, but the FOC does not hold at $(1/2,1/3)$ if $\tilde q(a)\neq 0$ and at $(1/2,2/3)$ if $\tilde q(a)\neq 1$.
\end{example}

\begin{example}[The Hausdorff limit of single-dipped sets might not be single-dipped.]\label{ex:stab}
Consider the linear receiver case. Let $\Theta=\{0,1/2,1\}$ and $\Gamma^n$ be given by
\[
\Gamma_{a}^n=
\begin{cases}
\{0,\frac 12\}, &a\in \left[\frac 1 4,\frac 1 2 - \frac {1}{4n}\right],\\
\{\frac 12,1\}, &a\in \left[\frac 12 +\frac{1}{4n}, \frac{3}{4}\right],\\
\emptyset, &\text{otherwise}.
\end{cases}
\]
Clearly, $\Gamma^n$ is single-dipped for each $n$, but the limit set $\Gamma$ given by
\[
\Gamma_{a} = 
\begin{cases}
\{0,\frac 12\}, &a\in [0,\frac 12),\\
\{0,\frac 12,1\}, &a=\frac 12,\\
\{\frac 12,1\}, &a\in (\frac 12,1],
\end{cases} 
\]
is not single-dipped, as it contains the strictly single-peaked triple $(1/2,0)$, $(3/4,  1/2)$, and $(1/2, 1)$. Note that for any convergent sequence of optimal outcomes $\pi^n\rightarrow \pi$ with $\supp(\pi^n)=\Gamma^n$, we have $\supp (\pi)=\Gamma$. Nevertheless, since each $\pi^n$ is supported on $\Gamma^n$ and satisfies (P2), it follows that $\pi (\{1/2\}\times \{0,1\})=0$, and hence $\pi$ is concentrated on the single-dipped set $\Gamma^\dagger =\Gamma\setminus (\{1/2\}\times \{0,1\})$.
\end{example}

\begin{example}[If $\Theta \neq \ol \Theta$, $\Gamma^\star$  might not be negative assortative.]\label{ex:NAD}
Consider the linear case with $V(a,\theta)=\sin (3\pi a)$. Let $\phi$ be uniform on $\Theta=\{0,1/2,1\}$. Then $\Gamma^\star =\{(1/6,0),(1/6,1/2),(5/6,1/2),(5/6,1)\}$,  so the unique optimal outcome induces action $1/6$ at state $0$, action $5/6$ at state $1$, and randomizes between actions $1/6$ and $5/6$ with equal probabilities at state $1/2$. Clearly, $\Gamma^\star$ is both strictly single-dipped and strictly single-peaked, and \eqref{e:nd} holds (e.g., at $\rho=2/3$ for $(\theta_1,\theta_2)=(0,1/2)$, at $\rho=1/3$ for $(\theta_1,\theta_2)=(1/2,1)$, and at $\rho=5/6$ for $(\theta_1,\theta_2)=(0,1)$), but $\Gamma^\star$ is not negative assortative.
\end{example}

\renewcommand{\thesection}{E}

\section{Additional Proofs}

\subsection{Proof of Lemma \ref{l:ASC}} \label{proof:ASC}

$(1)\implies (2).$ It is easy to see that Assumption \ref{a:qc} for $\mu=\delta_\theta$ such that $u(a,\theta)=0$ yields \eqref{1}. Similarly, Assumption \ref{a:qc} for $\mu=\rho \delta_\theta+(1-\rho )\delta_\theta$ such that $u(a,\theta)<0<u(a,\theta')$ and $\rho u(a,\theta)+(1-\rho )u(a,\theta')=0$  yields \eqref{2}.

$(2)\implies (1).$ By Lemma \ref{l:Choquet}, for any $a\in A$ and $\mu\in \Delta(\Theta)$ such that $\int u(a,\theta)\df\mu=0$, there exists $\lambda_\mu \in \Delta (\Delta(\Theta))$ such that $\int \eta \df\lambda_\mu =\mu$, and for each $\eta\in \supp (\lambda_\mu)$ there exist $\theta,\theta'\in \Theta$ and $\rho \in [0,1]$ such that $\eta = \rho \delta_\theta +(1-\rho )\delta_{\theta'}$ and
\begin{equation}\label{Eu=0}
\rho u(a,\theta)+(1-\rho )u(a,\theta')=0.
\end{equation}
It suffices to show that
\begin{equation}\label{3}
\rho u_a(a,\theta)+(1-\rho )u_a(a,\theta')<0.
\end{equation}
There are two cases to consider. First, if $\rho u(a,\theta)=0$, then \eqref{3} follows from \eqref{1} and \eqref{Eu=0}. Second, if $\rho u(a,\theta)\neq 0$, then \eqref{3} follows from \eqref{2} and \eqref{Eu=0}.

$(3)\implies (1).$ 
Notice that
\begin{equation*}
\int u(a,\theta)\df \mu =0 \iff \int \tilde u(a,\theta)\df \mu=0.	
\end{equation*}
Hence, if $\tilde u_a(a,\theta)<0$ for all $(a,\theta)$ and $\int u(a,\theta)\df \mu=0$, then 
\begin{equation*}
\int u_a(a,\theta)\df \mu = g(a)\int \tilde u_a(a,\theta)\df \mu+  g'(a) \int \tilde u(a,\theta)\df \mu=g(a)\int \tilde u_a(a,\theta)\df \mu<0,
\end{equation*}
yielding Assumption \ref{a:qc}.

$(1)\implies (3).$ We rely on the following lemma.
\begin{lemma}\label{l:ASC*}
If Assumptions \ref{a:smooth} and \ref{a:qc} hold, then there exists a continuous function $\gamma(a)$ such that
\begin{equation}\label{e:ASC*}
u_a(a,\theta)+\gamma (a) u(a,\theta) <0, \quad \text{for all } (a,\theta)\in A\times \Theta.
\end{equation}
\end{lemma}
Given this lemma, the required $g$ is given by
\[
g(a)=e^{-\int_0^a \gamma (\tilde a)\df \tilde a},
\]
as follows from
\[
\tilde u_a(a,\theta)= \frac{\partial }{\partial a}\left(\frac{u(a,\theta)}{e^{-\int_0^a \gamma (\tilde a)\df \tilde a}}\right) =\frac{u_a(a,\theta)+\gamma(a)u(a,\theta)}{e^{-\int_0^a \gamma (\tilde a)\df \tilde a}}<0.
\]

\begin{proof}[Proof of Lemma \ref{l:ASC*}] Fix $a\in [0,1]$. Let $M_+([0,1])$ be the set of positive Borel measures on $[0,1]$. Define the set $C\subset \R^3$ as follows 
\[C=\left\{\left(\int u(a,\theta)\df \mu,\int u_a(a,\theta)\df \mu -z,\int \df \mu\right)\ \big|\ \mu\in M_+([0,1]), \ z\geq 0 \right\}.\]
Clearly, $C$ is a convex cone. 

Moreover, $C$ is closed, because $u(a,\theta)$ and $u_a(a,\theta)$ are continuous in $\theta$. To see this, let sequences $\mu_n\in M_{+}([0,1])$ and $z_n\in \R_+^n$ be such that
\[
\int u(a,\theta) \df \mu_n\rightarrow c_1,\ \int u_a(a,\theta)\df \mu_n -z_n\rightarrow c_2,\ \int \df \mu_n\rightarrow c_3
\] 
for some $(c_1,c_2,c_3)\in \R^3$. It follows from $\int \df \mu_n\rightarrow c_3$ that all $\mu_n$ belong to a compact subset of positive measures whose total variation is bounded by $\sup_n \int \df \mu_n$, and hence, up to extraction of a subsequence, $\mu_n\rightarrow \mu\in M_+([0,1])$, with $\int \df \mu =c_3$. 
Since $u(a,\theta)$ and $u_a(a,\theta)$ are continuous in $\theta$, we get $\int u(a,\theta)\df \mu_n\rightarrow \int u(a,\theta)\df \mu=c_1$ and $\int u_a(a,\theta)\df \mu_n\rightarrow \int u_a(a,\theta)\df \mu$. Hence, $z_n\rightarrow \int u_a(a,\theta)\df \mu-c_2=z\geq 0$. In sum,
\[
\int u(a,\theta) \df \mu=c_1,\ \int u_a(a,\theta)\df \mu -z= c_2,\ \int \df \mu= c_3,
\]
showing that $C$ is closed.

Next, notice that Assumption \ref{a:qc} implies that $(0,0,1)\notin C.$ Thus, by the separation theorem (e.g., Corollary 5.84 in \citealt{aliprantis2006}), there exists $y\in \R^3$ such that, for all $\mu\in M_+([0,1])$ and $z\geq 0$,
\begin{align*}
0y_1+0y_2+1y_3<0&\leq \left(\int u(a,\theta)\df \mu\right)y_1+\left(\int u_a(a,\theta)\df \mu -z\right)y_2+\left(\int \df \mu\right)y_3,
\end{align*}
or equivalently
\begin{equation}\label{Cy}
\begin{aligned}
u(a,\theta) y_1 + u_a(a,\theta) y_2 +y_3 &\geq 0, \quad \text{for all}\ \theta\in [0,1],\\
-y_2 &\geq 0,\\
y_3 &<0.
\end{aligned}
\end{equation}
We now show that there exists a scalar $\gamma(a) \in \R$ satisfying
\begin{equation}\label{e:gamma}
u_a(a,\theta)+\gamma(a) u(a,\theta)<0, \quad \text{for all} \ \theta\in [0,1].	
\end{equation}
There are two cases. First, if $y_2<0$ then $\gamma(a) = y_1/y_2\in \R$ satisfies \eqref{e:gamma}. Second, if $y_2=0$ then \eqref{Cy} implies that
\[
u(a,\theta)y_1\geq -y_3>0, \quad \text{for all } \theta\in [0,1].
\]
Thus, we have either (i) $u(a,\theta)>0$ for all $\theta\in [0,1]$, so, taking into account continuity of $u(a,\theta)$ and $u_a(a,\theta)$ in $\theta$,  
\[\gamma(a)=\min_{\theta\in [0,1]}\left\{-\frac{u_a(a,\theta)}{u(a,\theta)}\right\}-1\in \R\] 
satisfies \eqref{e:gamma}; or (ii) $u(a,\theta)<0$ for all $\theta\in [0,1]$, so  
\[\gamma(a) =\max_{\theta\in [0,1]}\left\{-\frac{u_a(a,\theta)}{u(a,\theta)}\right\}+1\in \R\] 
satisfies \eqref{e:gamma}. 

It remains to show that if for all $a\in [0,1]$ there exists $\gamma(a)\in \R$ satisfying \eqref{e:gamma}, then there exists a continuous function $\tilde \gamma:[0,1]\rightarrow \R$ satisfying \eqref{e:gamma}. Define a correspondence $\varphi:[0,1]\rightrightarrows \R$,
\[
\varphi (a)=\{r\in \R: u_a(a,\theta)+ru(a,\theta)<0, \quad  \text{for all } \theta\in [0,1]\}.
\]
Note that $\varphi$ is nonempty valued by assumption, and is clearly convex valued. In addition, $\varphi$ has open lower sections, because for each $r\in \R$ the set 
\[\{a\in [0,1]:u_a(a,\theta)+r u(a,\theta)<0,\quad \text{for all }\theta\in [0,1]\}\]
is open, since $u_a$ and $u$ are continuous on the compact set $[0,1]\times [0,1]$. Thus, by Browder's Selection Theorem (Theorem 17.63 in \citealt{aliprantis2006}), $\varphi$ admits a continuous selection $\tilde \gamma$, which by construction satisfies \eqref{e:gamma}.
\end{proof}

\subsection{Proof of Lemma 1, Points 1 and 3} \hfill

\textbf{Point 1. }The set of feasible solutions to (P) is clearly nonempty, as $\pi (a,\theta )=\phi (\theta
)\delta _{a^{\ast }(\phi )}(a)$ (i.e., no disclosure) is feasible. Since the set $A\times \Theta $ is compact, the set of probability measures $\Delta
(A\times \Theta )$ is also compact (in the weak* topology), by Prokhorov's
theorem. The constraint map in (P1) is continuous because it is a projection,
and the constraint map in (P2) is continuous because $u(a,\theta )$ is
continuous in $(a,\theta )$; so the set of feasible solutions is a closed
subset of the compact set $\Delta (A\times \Theta )$, and is thus itself
compact. Since $V(a,\theta )$ is continuous, the objective function is
continuous, and thus attains its maximum on the compact set of feasible
solutions.

\textbf{Point 3. }Consider a tightened dual problem in
which $(p,q)\in C(\Theta )\times C(A)$, and let $F_{D}$ be the set of
feasible solutions of the original dual problem: $(p,q)\in C(\Theta )\times
B(A)$ satisfying (D1). Let $F_{P}$ be the set of feasible
solutions of the primal problem: $\pi \in \Delta (A\times \Theta )$
satisfying (P1) and (P2). 
Weak duality follows easily: 
\begin{equation}
\begin{aligned} \inf_{(p,q)\in F_D,\ q\in C(A)} \int_\Theta p(\theta)\df \phi
(\theta) &\geq \inf_{(p,q)\in F_D} \int_\Theta p(\theta)\df \phi (\theta) \\
&=\inf_{(p,q)\in F_D,\ \pi\in F_P} \int_{A\times \Theta} p(\theta)
\df \pi(a,\theta) \\ &\geq \sup_{(p,q)\in F_D,\ \pi\in F_P} \int_{A\times
\Theta} (V(a,\theta)+q(a)u(a,\theta))\df \pi(a,\theta) \\ &=\sup_{\pi \in F_P}
\int_{A\times\Theta} V(a,\theta)\df \pi(a,\theta), \end{aligned}  \label{e:WD}
\end{equation}%
where the first inequality holds because the original dual problem is more
relaxed than the tightened dual problem, the first equality holds by (P1),
the second inequality holds by (D1), and the second equality holds by (P2).

By the Riesz representation theorem, the space $M_{r}(A\times \Theta )$ of
regular, signed Borel measures on the compact set $A\times \Theta $ with the
total variation norm is the topological dual of the space $C(A\times \Theta )
$ of continuous functions on $A\times \Theta $ with the supremum norm.
Moreover, the set of (positive) measures in $M_{r}(A\times \Theta )$, $%
M_{r}(\Theta )$, and $M_{r}(A)$ are all weak* closed, so the positive cones
in the primal variable space and the primal constraint space are closed. 

The
tightened dual problem has a finite value, since it is bounded below by the
value of the primal problem and is bounded above by $\overline{V}:=\max_{a,\theta}V(a,\theta)$, as $%
(p,q)=(\overline{V},0)$ is feasible. Moreover, since $u,V\in C(A\times
\Theta )$, there is an interior feasible solution $(p,q)=(1+\overline{V},0)$
of the tightened dual problem, as the function $p(\theta
)-q(a)u(a,\theta )-V(a,\theta )=1+\overline{V}-V(a,\theta )$ lies in the
interior of the positive cone of $C(A\times \Omega )$. Together with the closedness properties established in the previous paragraph, this implies that the
(generalized) Slater condition is satisfied for the tightened dual problem, so
there is no duality gap by Corollary 3.14 in \citet{anderson1987}: that is,
\begin{equation*}
\inf_{(p,q)\in F_{D},\ q\in C(A)}\int_{\Theta }p(\theta )\df \phi (\theta
)=\sup_{\pi \in F_{P}}\int_{A\times \Theta }V(a,\theta )\df \pi (a,\theta ),
\end{equation*}%
It follows that all inequalities in \eqref{e:WD} hold with equality. Finally, as the original dual and primal problems admit
solutions, we have 
\begin{align*}
\min_{(p,q)\in F_{D}}\int_{\Theta }p(\theta )\df \phi (\theta )=\max_{\pi \in F_{P}}\int_{A\times \Theta }V(a,\theta )\df \pi (a,\theta ).
\end{align*}

\subsection{Proof of Lemma \ref{l:stab}}
We give the theorem for the single-dipped case. Let  $\pi^n$ be any optimal outcome, so that $\supp (\pi^n)\subset \Gamma^n$. Since the set of compact subsets of a compact set is compact (in the Hausdorff topology), taking a subsequence if necessary, $\Gamma^n$ converges to some compact set $\ol\Gamma\subset A\times \Theta$. Since the set of implementable outcomes is compact (in the weak* topology), taking a subsequence if necessary, $\pi^n$ converges weakly to some implementable outcome $ \pi$. Finally, since $\Gamma^n\rightarrow \ol\Gamma$, $\pi^n\rightarrow \pi$, and $\supp( \pi^n)\subset \Gamma^n$, it follows that $\supp (\pi) \subset \ol \Gamma$, by  Box 1.13 in \citet{santambrogio}. 

We claim that $\pi$ is optimal under $v$. Since $v^n$ converges uniformly to $v$, for each $\delta>0$ there exists $n_\delta\in \N$ such that, for all $n\geq n_\delta$, we have $|v_n(a,\theta)-v(a,\theta)|\leq \delta$ for all $(a,\theta)$. Since $\pi^n$ is optimal under $v^n$, for each implementable outcome $\tilde\pi$ we have
\begin{align*}
\int_{A\times \Theta} \int_0^a v (\tilde a,\theta)\df \pi^n(a,\theta) &\geq \int_{A\times \Theta} \int_0^a v^n (\tilde a,\theta)\df \pi^n(a,\theta)-\delta\\
&\geq \int_{A\times \Theta} \int_0^a v^n (\tilde a,\theta)\df \tilde \pi(a,\theta)	-\delta\\
&\geq \int_{A\times \Theta} \int_0^a v(\tilde a,\theta)\df \pi(a,\theta)-2\delta.
\end{align*}
Passing to the limit as $\delta\rightarrow 0$ and $n\rightarrow \infty$ establishes the optimality of $\pi$  under $v$. 

Let $\Phi$ be a subset of $A$ such that $a\in \Phi$ iff there exists a strictly single-peaked triple $(a_1,\theta_1)$, $(a_2,\theta_2)$, $(a_1,\theta_3)$ in  $\ol \Gamma$  with $a= a_1$. Define \[\Gamma^\dagger =\ol \Gamma \setminus \cup_{a\in \Phi}\left( \{a\}\times \Theta\setminus \{\theta^\star (a)\}\right).\]  
We show that $\Gamma^\dagger $ is a Borel single-dipped set satisfying $\pi(\Gamma^\dagger)=1$, and hence $\pi$ is single-dipped.

First, we show that $\Gamma^\dagger$ is single-dipped. For each strictly single-peaked triple $(a_1,\theta_1)$, $(a_2,\theta_2)$, $(a_1,\theta_3)$ in $\ol \Gamma$, we have $a_1\in \Phi$, and thus $(a_1,\theta)\in \Gamma^\dagger$ only if $\theta=\theta^\star (a_1)$. But then $(a_1,\theta_1)$ and $(a_1,\theta_3)$ cannot both be in $\Gamma^\dagger$, as $\theta_1\neq \theta_3$.

Second, we show that for each strictly single-peaked triple $(a_1,\theta_1)$, $(a_2,\theta_2)$, $(a_1,\theta_3)$ in $\ol \Gamma$, we have $\theta^\star (a_1)=\theta_2$. Fix such a triple. Since $\Gamma^n\rightarrow\ol \Gamma$ and $\theta^\star (a)$ is uniformly continuous on $A$, for each $\delta>0$ there exist $n\in \N$ and a triple $(a_1^n,\theta_1^n)$, $(a_2^n,\theta_2^n)$, $(a_3^n,\theta_3^n)$ in $\Gamma^n$ such that $\theta_1^n<\theta_2^n<\theta_3^n$, $a_1^n<a_2^n$, $a_3^n<a_2^n$, $|\theta_i^n-\theta_i|\leq \delta$, and $|\theta^\star (a_i^n)-\theta^\star(a_i)|\leq \delta$ for all $i\in \{1,2,3\}$ (where $a_3 =a_1$). Hence,
\[
\theta^\star(a_1)-\delta\leq \theta^\star(a_1^n) \leq \theta_2^n\leq \theta_2+\delta \implies \theta^\star(a_1)\leq \theta_2+2\delta.
\]
To understand the middle inequality, suppose by contradiction that $\theta^\star(a_1^n)>\theta_2^n$. Recall that, by Theorem \ref{t:contact}, each contact set $\Gamma^n$ satisfies $\min \Gamma^n_{a_1^n}\leq \theta^\star(a_1^n)\leq\max \Gamma^n_{a_1^n}$. Hence, there exists $\hat \theta_1^n\in \Gamma^n_{a_1^n}$ with $\hat \theta_1^n\geq \theta^\star(a_1^n)>\theta_2^n$ (for example, $\hat \theta_1^n=\max \Gamma^n_{a_1^n}$). But then $\Gamma^n$ cannot be single-dipped, as it contains the strictly single-peaked triple $(a_1^n,\theta_1^n), (a_2^n,\theta_2^n),(a_1^n,\hat \theta_1^n)$. By an analogous argument, we get
\[
\theta_2-\delta\leq \theta_2^n\leq \theta^\star(a_3^n)\leq \theta^\star(a_1)+\delta \implies \theta^\star(a_1)\geq \theta_2-2\delta. 
\]
Since $\delta>0$ is arbitrary, we get $\theta^\star(a_1)=\theta_2$.

Third, we show that for any two strictly single-peaked triples $(a_1,\theta_1)$, $(a_2,\theta_2)$, $(a_1,\theta_3)$ and $(\tilde a_1,\tilde \theta_1)$, $(\tilde a_2,\tilde \theta_2)$, $(\tilde a_1,\tilde \theta_3)$ in $\ol \Gamma$, we have $\tilde \theta_2\notin (\theta_2,\theta_3)$. Suppose by contradiction that $\tilde \theta_2\in (\theta_2,\theta_3)$. By the previous paragraph, $\theta^\star(a_1)=\theta_2$ and $\theta^\star(\tilde a_1)=\tilde \theta_2$. Moreover, for each $\delta>0$, there exist $n\in \N$ and two triples $(a_1^n,\theta_1^n)$, $(a_2^n,\theta_2^n)$, $(a_3^n,\theta_3^n)$ and $(\tilde a_1^n,\tilde \theta_1^n)$, $(\tilde a_2^n,\tilde \theta_2^n)$, $(\tilde a_3^n,\tilde \theta_3^n)$ in $\Gamma^n$ such that $|\theta_i^n-\theta_i|\leq \delta$, $|\theta^\star(a_i^n)-\theta^\star(a_i)|\leq \delta$, $|\tilde \theta_i^n-\tilde \theta_i|\leq \delta$, and $|\theta^\star(\tilde a_i^n)-\theta^\star(\tilde a_i)|\leq \delta$ for all $i\in \{1,2,3\}$ (where $a_3 =a_1$ and $\tilde a_3 =\tilde a_1$). Next, since $\min \Gamma^n_{a_3^n}\leq \theta^\star(a_3^n)\leq \max \Gamma^n _{a_3^n}$, %
there exists $\hat \theta_3^n\in \Gamma^n_{a_3^n}$ %
such that $\hat \theta_3^n\leq \theta^\star(a_3^n)$ %
(for example, $\hat \theta_3^n =\min \Gamma^n_{a_3^n}$). %
Since $\Gamma^n$ is single-dipped, to reach a contradiction it suffices to show that the triple $(a_3^n,\hat \theta_3^n)$, $(\tilde a_2^n,\tilde \theta_2^n)$, $(a_3^n,\theta_3^n)$ (which is in $\Gamma^n$ by construction) is strictly single-peaked for small enough $\delta>0$. To see this, notice that we  have
\begin{align*}
\theta^\star(a_3^n) &\leq \theta^\star(a_1)+\delta =\theta_2+\delta,\\
\theta^\star(\tilde a_2^n) &\geq \theta^\star(\tilde a_2)-\delta> \theta^\star(\tilde a_1	)-\delta =\tilde \theta_2-\delta,\\
\hat \theta_3^n &\leq \theta^\star(a_3^n)\leq \theta^\star(a_1)+\delta =\theta_2+\delta,\\
\tilde \theta_2^n &\in [\tilde \theta_2-\delta,\tilde \theta_2+\delta],\\
\theta_3^n &\geq \theta_3-\delta.  
\end{align*}
Thus, if $\delta\in (0, \min \{(\tilde \theta_2-\theta_2)/2, (\theta_3-\tilde \theta_2)/2\})$, then $a_3^n<\tilde a_2^n$ and $\hat \theta_3^n<\tilde \theta_2^n<\theta_3^n$, so the triple $(a_3^n,\hat \theta_3^n)$, $(\tilde a_2^n ,\tilde \theta_2^n)$, $(a_3^n,\theta_3^n)$  is strictly single-peaked.

Fourth, we show that the set $\Phi$ is countable, and thus $\Gamma^\dagger$ is Borel. If $a_1\in \Phi$, then there exists a strictly single-peaked triple $(a_1,\theta_1)$, $(a_2,\theta_2)$, $(a_1,\theta_3)$ in $\ol \Gamma$ with $\theta^\star(a_1)=\theta_2$. Let us associate with each such $a_1$ some rational number $r(a_1)\in (\theta_2,\theta_3)$. Since for any other strictly single-peaked triple $(\tilde a_1,\tilde \theta_1)$, $(\tilde a_2,\tilde \theta_2)$, $(\tilde a_1,\tilde \theta_3)$ in $\ol \Gamma$, we have $\tilde \theta_2\notin (\theta_2,\theta_3)$ and, by symmetry, $\theta_2\notin (\tilde \theta_2,\tilde \theta_3)$, we see that $(\theta_2,\theta_3)\cap (\tilde \theta_2,\tilde \theta_3)=\emptyset$ if $\theta_2\neq \tilde \theta_2$. Consequently, $r(a_1)\neq r(\tilde a_1)$ if $a_1,\tilde a_1\in \Phi$ and $a_1\neq \tilde a_1$. Thus, $r$ is a one-to-one mapping of $\Phi$ into a subset of the set of rational numbers, so $\Phi$ is countable.
 
Finally, we show that $\pi(\Gamma^\dagger)=1$. Since $\Phi$ is countable and probability measures are countably additive, it suffices to show that $\pi(\{a_1\}\times\Theta \setminus \{\theta^\star(a_1)\})=0$ for each $a_1\in \Phi$. In turn, this follows if for each $\varepsilon_\theta>0$, we have 
\[\pi ((a_1-\varepsilon,a_1+\varepsilon)\times \Theta\setminus [\theta^\star(a_1)-\varepsilon_\theta,\theta^\star(a_1)+\varepsilon_\theta])\rightarrow 0\text{ as }\varepsilon\rightarrow 0.\] 
Fix $a_1\in \Phi$ and a strictly single-peaked triple $(a_1,\theta_1)$, $(a_2,\theta^\star(a_1))$, $(a_1,\theta_3)$ in $\ol \Gamma$. Let $\varepsilon\in (0,a_2-a_1)$.  Since  $\Gamma^n\rightarrow \ol \Gamma$,  $(a_2,\theta^\star (a_1))\in \ol \Gamma$, and $a_2>a_1$ (and thus $\theta^\star (a_2)>\theta^\star (a_1)$), there exists $n\in \N$ and $(a_2^n,\theta_2^n)\in \Gamma^n$ with $\theta^\star (a_1-\varepsilon)<\theta_2^n<\theta^\star(a_1+\varepsilon)<\theta^\star (a_2^n)$.

Since $\Gamma^n$ is a compact single-dipped set with $\min \Gamma^n_a\leq \theta^\star(a)\leq \max \Gamma^n_a$ for all $a\in A_{\Gamma^n}$, the triple $(a,\min \Gamma^n_a)$, $(a_2^n, \theta_2^n)$, $(a,\max \Gamma^n_a)$ cannot be strictly single-peaked. Hence, we have the following implications: \\
(i) if $a\in [a_1-\varepsilon,a^\star(\delta_{\theta_2^n}))$, then $\Gamma_a^n\cap (\theta_2^n,1]=\emptyset$; \\
(ii) if $a\in (a^\star(\delta_{\theta_2^n}),\theta_2+\varepsilon]$, then $\Gamma_a^n\cap [0,\theta_2^n)=\emptyset$;\\
(iii) if $a=a^\star(\delta_{\theta_2^n})$, then $\Gamma_a^n\cap (\theta_2^n,1]=\emptyset$ or $\Gamma_a^n\cap [0,\theta_2^n)=\emptyset$. 

Let $\chi_0^n=\pi^n((a_1-\varepsilon,a_1+\varepsilon)\times [0,\theta^\star(a_1)-\varepsilon_\theta))$ and $\chi_1^n=\pi^n((a_1-\varepsilon,a_1+\varepsilon)\times (\theta^\star(a_1)+\varepsilon_\theta,1])$. By (P2) and condition (iii), we have $\pi^n(a^\star(\delta_{\theta_2^n})\times\Theta\setminus \{\theta_2^n\})=0$, and hence
\begin{align*}
\int_{\left\{a^\star(\delta_{\theta_2^n})\right\}\times [0,\theta_2^n]}u(a,\theta)\df \pi^n(a,\theta)=\int_{\left\{a^\star(\delta_{\theta_2^n})\right\}\times [\theta_2^n,1]}u(a,\theta)\df \pi^n(a,\theta)=0.
\end{align*}
Together with conditions (i) and (ii) (and again using (P2)), we have
\begin{align*}
0%
&=\int_{(a_1-\varepsilon,a^\star(\delta_{\theta_2^n})]\times [0,\theta_2^n]} u(a,\theta)\df \pi^n(a,\theta)\\
&\leq \max_{(a,\theta)\in [a_1-\varepsilon,a^\star(\delta_{\theta_2^n})]\times [0,\theta^\star(a_1)-\varepsilon_\theta]} u(a,\theta)\chi_0^n\\
&\quad +\max_{(a,\theta)\in [a_1-\varepsilon,a^\star(\delta_{\theta_2^n})]\times [\theta^\star(a_1)-\varepsilon_\theta,\theta_2^n]} u(a,\theta)(1-\chi_0^n)\\
\implies&\chi_0^n\leq \frac{\underset{(a,\theta)\in [a_1-\varepsilon,a^\star(\delta_{\theta_2^n})]\times [\theta^\star(a_1)-\varepsilon_\theta,\theta_2^n]}{\max }u(a,\theta)}{\underset{(a,\theta)\in [a_1-\varepsilon,a^\star(\delta_{\theta_2^n})]\times [\theta^\star(a_1)-\varepsilon_\theta,\theta_2^n]}{\max }u(a,\theta)-\underset{(a,\theta)\in [a_1-\varepsilon,a^\star(\delta_{\theta_2^n})]\times [0,\theta^\star(a_1)-\varepsilon_\theta]}{\max} u(a,\theta)},
\end{align*}
and 
\begin{align*}
0%
&=\int_{[a^\star(\delta_{\theta_2^n}),a_1+\varepsilon)\times [\theta_2^n,1]} u(a,\theta)\df \pi^n(a,\theta)\\
&\geq \min_{(a,\theta)\in [a^\star(\delta_{\theta_2^n}),a_1+\varepsilon]\times [\theta^\star(a_1)+\varepsilon_\theta,1]} u(a,\theta)\chi_1^n\\
&\quad +\min_{(a,\theta)\in [a^\star(\delta_{\theta_2^n}),a_1+\varepsilon]\times [\theta_2^n,\theta^\star(a_1)+\varepsilon_\theta]} u(a,\theta)(1-\chi_1^n)\\
\implies& \chi_1^n  \leq \frac{-\underset{(a,\theta)\in [a^\star(\delta_{\theta_2^n}),a_1+\varepsilon]\times [\theta_2^n,\theta^\star(a_1)+\varepsilon_\theta]}{\min}u(a,\theta)}{\underset{(a,\theta)\in [a^\star(\delta_{\theta_2^n}),a_1+\varepsilon]\times [\theta^\star(a_1)+\varepsilon_\theta,1]}{\min} u(a,\theta)-\underset{(a,\theta)\in [a^\star(\delta_{\theta_2^n}),a_1+\varepsilon]\times [\theta_2^n,\theta^\star(a_1)+\varepsilon_\theta]}{\min} u(a,\theta)},
\end{align*}
where the inequalities hold because $u(a_1-\varepsilon,\theta_2^n)>0$ and $u(a_1+\varepsilon,\theta_2^n)<0$, by Assumption \ref{a:sc} and $\theta^\star (a_1-\varepsilon)<\theta_2^n<\theta^\star(a_1+\varepsilon)$.

By Assumptions \ref{a:smooth} and \ref{a:qc}, $u_a(a,\theta)<0$ for all $(a,\theta)$ in a neighborhood of $(a,\theta^*(a))$. Hence, for sufficiently small $\varepsilon$, $u(a,\theta)$ is maximized over $a \in [a_1-\varepsilon,a^\star(\delta_{\theta_2^n})]$ at $a=a_1-\varepsilon$, and $u(a,\theta)$ is minimized over $a \in [a^\star(\delta_{\theta_2^n}),a_1+\varepsilon]$ at $a=a_1+\varepsilon$. Therefore, By the Portmanteau Theorem (Theorem 15.3 in \citealt{aliprantis2006}), for sufficiently small $\varepsilon>0$, we get
\begin{align*}
&\pi((a_1-\varepsilon,a_1+\varepsilon)\times[0,\theta^\star(a_1)-\varepsilon_\theta))\leq \liminf_n \chi_0^n\\ 
\leq &\frac{\underset{\theta\in[\theta^\star(a_1)-\varepsilon_\theta,\theta^\star(a_1)]}{\max }u(a_1-\varepsilon,\theta)}{\underset{\theta\in [\theta^\star(a_1)-\varepsilon_\theta,\theta^\star(a_1)]}{\max }u(a_1-\varepsilon,\theta)-\underset{\theta\in[0,\theta^\star(a_1)-\varepsilon_\theta]}{\max }u(a_1-\varepsilon,\theta)},\\ 
&\pi((a_1-\varepsilon,a_1+\varepsilon)\times(\theta^\star(a_1)+\varepsilon_\theta,1])\leq \liminf_n \chi_1^n \\
\leq &\frac{-\underset{\theta\in [\theta^\star(a_1),\theta^\star(a_1)+\varepsilon_\theta]}{\min}u(a_1+\varepsilon,\theta)}{\underset{\theta\in [\theta^\star(a_1)+\varepsilon_\theta,1]}{\min}u(a_1+\varepsilon,\theta)-\underset{\theta\in [\theta^\star(a_1),\theta^\star(a_1)+\varepsilon_\theta]}{\min}u(a_1+\varepsilon,\theta)}.	
\end{align*}
Next, taking into account Assumption \ref{a:sc}, we get
\begin{align*}
\lim_{\varepsilon\rightarrow 0}\underset{\theta\in[\theta^\star(a_1)-\varepsilon_\theta,\theta^\star(a_1)]}{\max }u(a_1-\varepsilon,\theta)&=0, \quad &\lim_{\varepsilon\rightarrow 0}\underset{\theta\in[0,\theta^\star(a_1)-\varepsilon_\theta]}{\max }u(a_1-\varepsilon,\theta)<0, \\
\lim_{\varepsilon\rightarrow 0}\underset{\theta\in [\theta^\star(a_1),\theta^\star(a_1)+\varepsilon_\theta]}{\min}u(a_1+\varepsilon,\theta)&=0, \quad &\lim_{\varepsilon\rightarrow 0}\underset{\theta\in [\theta^\star(a_1)+\varepsilon_\theta,1]}{\min}u(a_1+\varepsilon,\theta)>0.
\end{align*}
Consequently,  $\pi ((a_1-\varepsilon,a_1+\varepsilon)\times \Theta\setminus [\theta^\star(a_1)-\varepsilon_\theta,\theta^\star(a_1)+\varepsilon_\theta])\rightarrow 0$ as $\varepsilon\rightarrow 0$.

\subsection{Linear Receiver Case}\label{app:SR}
For the linear receiver case, we replace Theorem \ref{l:ssdd} with Theorem \ref{l:ssddSR}.
\begin{theorem}\label{l:ssddSR}
Let Assumptions \ref{a:smooth}--\ref{a:sc} hold.
Suppose that for all $\theta_1<\theta_2<\theta_3$ and all $a_2>(<)a_1$ such that $\theta_1\leq \theta^\star(a_1)\leq \theta_3$, we have $|R|>(<)0$ and
\begin{align*}
u(a_2,\theta_2)u(a_1,\theta_1)&\geq(\leq) u(a_2,\theta_1)u(a_1,\theta_2),\\
u(a_2,\theta_3)u(a_1,\theta_2)&\geq(\leq) u(a_2,\theta_2)u(a_1,\theta_3).
\end{align*}
Then $\Gamma$ is single-dipped (single-peaked).
\end{theorem}
\begin{proof}
We give the theorem for the single-dipped case. Suppose, by contradiction, that $\Gamma$ is not single-dipped. Then, as shown in the proof of Theorem \ref{l:ssdd}, there exists a strictly single-peaked triple $(a_1,\theta_1)$, $(a_2,\theta_2)$, $(a_3,\theta_3)$ in $\Gamma$ such that $\theta_1\leq \theta^\star(a_1)\leq \theta_3$. 

As shown in the proof of Theorem \ref{t:SDPD}, Assumptions \ref{a:qc}--\ref{a:sc} and $\theta_1\leq \theta^\star(a_1)\leq \theta_3$ imply
\[
u(a_2,\theta_3)u(a_1,\theta_1)> u(a_2,\theta_1)u(a_1,\theta_3).
\]
By (D1) and Theorem \ref{t:contact}, we have
\begin{align*}
V(a_1,\theta_1)+q(a_1) u(a_1,\theta_1) &\geq V(a_2,\theta_1) +q(a_2) u(a_2,\theta_1),\\
V(a_2,\theta_2)+q(a_2) u(a_2,\theta_2) &\geq V(a_1,\theta_2) +q(a_1) u(a_1,\theta_2),\\
V(a_1,\theta_3)+q(a_1) u(a_1,\theta_3) &\geq V(a_2,\theta_3) +q(a_2) u(a_2,\theta_3),
\end{align*}
Adding up the first inequality multiplied by $u(a_2,\theta_3)u(a_1,\theta_2)- u(a_2,\theta_2)u(a_1,\theta_3)>0$, the second inequality multiplied by $u(a_2,\theta_3)u(a_1,\theta_1)-u(a_2,\theta_1)u(a_1,\theta_3)\geq 0$, and the third inequality multiplied by $u(a_2,\theta_2)u(a_1,\theta_1) - u(a_2,\theta_1)u(a_1,\theta_2)\geq 0$, we get $|R|\leq 0$, leading to a contradiction.
\end{proof}
Notice that in the linear receiver case, the conditions of Theorem \ref{l:ssddSR} are satisfied if $v_\theta(a,\theta)$ is strictly increasing (decreasing) in $\theta$. Moreover, notice that, in the linear receiver case, the proofs of Lemmas \ref{l:S><0} and \ref{l:R>0} remain valid without Assumption \ref{a:v>0}, because
\begin{gather*}
\begin{vmatrix}
u(a,\theta_2)-u(a,\theta_1) &u(a,\theta_3)-u(a,\theta_2)\\
u_a(a,\theta_2)-u_a(a,\theta_1) &u_a(a,\theta_3)-u_a(a,\theta_2)
\end{vmatrix}=0,\\
\begin{vmatrix}
u(a_1,\theta_2)-u(a_1,\theta_1) &u(a_1,\theta_3)-u(a_1,\theta_2)\\
u(a_2,\theta_2)-u(a_2,\theta_1) &u(a_2,\theta_3)-u(a_2,\theta_2)
\end{vmatrix}=0.
\end{gather*}
Thus, to prove Theorem \ref{t:SDPD} in the linear receiver case without Assumption \ref{a:v>0}, we just need to replace the vector $y$ in the proof of the single-peaked case with 
\[
y=
-
\begin{pmatrix}
u(a_2,\theta_3)u(a_1,\theta_2)-u(a_2,\theta_2)u(a_1,\theta_3)\\
u(a_2,\theta_3)u(a_1,\theta_1)-u(a_2,\theta_1)u(a_1,\theta_3)\\
u(a_2,\theta_2)u(a_1,\theta_1)-u(a_2,\theta_1)u(a_1,\theta_2)
\end{pmatrix}.
\]

\subsection{Proof of Lemma \ref{p:full}} The support of the full disclosure outcome is $\cup _{\theta\in \Theta} (a^\star (\delta_\theta),\theta)$. Thus, by Lemma \ref{l:dual} and Theorem \ref{t:contact}, full disclosure is optimal iff there exists $q\in B(a)$ such that
\begin{gather*}
V(a^\star(\delta_\theta),\theta)\geq V(a,\theta) +q(a) u(a,\theta), \quad \text{for all $(a,\theta)\in A\times \Theta$},\\
\iff \frac{V(a,\theta_1)-V(a^\star(\delta_{\theta_1}),\theta_1)}{-u(a,\theta_1)}\leq q(a)\leq \frac{V(a^\star(\delta_{\theta_2}),\theta_2)-V(a,\theta_2)}{u(a,\theta_2)},
\end{gather*}
for all $a\in A$ and $\theta_1,\theta_2\in \Theta$ such that $\theta_1<\theta^\star (a)<\theta_2$.
As shown in the proof of Lemma \ref{l:D=D'}, the left-hand side and right-hand side functions are bounded on $A\times \Theta$, so full disclosure is optimal iff, for all $a\in A$ and $\theta_1,\theta_2\in \Theta$ such that $\theta_1<\theta^\star (a)<\theta_2$, we have
\begin{gather*}
\frac{V(a,\theta_1)-V(a^\star(\delta_{\theta_1}),\theta_1)}{-u(a,\theta_1)}\leq \frac{V(a^\star(\delta_{\theta_2}),\theta_2)-V(a,\theta_2)}{u(a,\theta_2)},\\
\iff u(a,\theta_2)V(a,\theta_1)-u(a,\theta_1)V(a,\theta_2)\leq u(a,\theta_2)V(a^\star(\delta_{\theta_1}),\theta_1)-u(a,\theta_1)V(a^\star(\delta_{\theta_2}),\theta_2),\\
\iff \rho V(a ^{\star }(\mu),\theta_1) +(1-\rho )V(a ^{\star }(\mu),\theta_2)\leq \rho V(a ^{\star }(\delta_{\theta_1})),\theta_1) +(1-\rho )V(a ^{\star }(\delta_{\theta_2}),\theta_2),
\end{gather*}
where $\rho =u(a,\theta_2)/(u(a,\theta_2)-u(a,\theta_1))$, $\mu=\rho \delta_{\theta_1}+(1-\rho )\delta_{\theta_2}$, and $a^\star(\mu)=a$, by the definition of $a^\star (\mu)$. To complete the proof that full disclosure is optimal iff \eqref{e:full} holds for all $\mu$, note that for all $a$ and $\theta_1,\theta_2\in \Theta$ such that $\theta_1<\theta^\star (a)<\theta_2$, we have $\rho =u(a,\theta_2)/(u(a,\theta_2)-u(a,\theta_1))\in (0,1)$; and conversely, for each $\theta_1<\theta_2$ and $\rho \in (0,1)$, there exists a unique $a\in (a^\star(\delta_{\theta_1}),a^\star(\delta_{\theta_2}))$ such that $\rho =u(a,\theta_2)/(u(a,\theta_2)-u(a,\theta_1))$.

Finally, assume that \eqref{e:full} holds with strict inequality for all $\mu$. Suppose by contradiction that full disclosure is not uniquely optimal. Then, by Theorem \ref{t:contact}, there exist $a\in A_\Gamma$ and $\theta_1,\theta_2\in \Gamma^\star_a$ such that $\theta_1<\theta^\star(a)<\theta_2$, so
\begin{align*}
V(a,\theta_1)+q(a)u(a,\theta_1)\geq V(a^\star(\delta_{\theta_1}),\theta_1),\\
V(a,\theta_2)+q(a)u(a,\theta_2)\geq V(a^\star(\delta_{\theta_2}),\theta_1).	
\end{align*}
Denote $\rho =u(a,\theta_2)/(u(a,\theta_2)-u(a,\theta_1)\in (0,1)$ and $\mu=\rho \delta_{\theta_1}+(1-\rho )\delta_{\theta_2}$. Notice that $a=a^\star (\mu)$. Adding the first inequality multiplied by $\rho$ and the second inequality multiplied by $1-\rho$ gives
\[
\rho V(a ^{\star }(\mu),\theta_1) +(1-\rho )V(a ^{\star }(\mu),\theta_2)\geq \rho V(a ^{\star }(\delta_{\theta_1}),\theta_1) +(1-\rho )V(a ^{\star }(\delta_{\theta_2}),\theta_2),
\]
contradicting that \eqref{e:full} holds with strict inequality.

\subsection{Proof of Corollary \ref{c:fd}'}
Condition \eqref{e:full} holds because 
\begin{align*}
& \rho V(p\theta_1 +(1-\rho )\theta_2,\theta_1 )+(1-\rho )V(\rho \theta_1 +(1-\rho )\theta_2,\theta_2) \\
& \leq \rho (\rho V(\theta_1 ,\theta_1 )+(1-\rho )V(\theta_2,\theta_1))+(1-\rho )(\rho V(\theta_1 ,\theta_2)+(1-\rho )V(\theta_2,\theta_2)) \\
& \leq \rho V(\theta_1 ,\theta_1 )+(1-\rho )V(\theta_2,\theta_2),
\end{align*}%
where the first inequality holds because $V(a ,\theta )$ is convex in $a$, and the second holds because $V(\theta_1 ,\theta_2)+V(\theta_2,\theta_1 )\leq V(\theta_1 ,\theta_1 )+V(\theta_2,\theta_2)$.

\subsection{Proof for Example \ref{e:quantile}} \label{s:quantileproof}
First, notice that the outcome $\pi$ is implementable. (P1) holds because, for all $a\in [\ul a,a]$
\begin{gather*}
\pi_a=\frac{\df\phi ([0,t_1(a)])}{\df\phi ([0,t_1(a)]+\df\phi ([a,1])}\delta_{t_1(a)}+\frac{\df\phi ([a,1])}{\df\phi ([0,t_1(a)]+\df\phi ([a,1])}\delta_{t_1(a)},\\
\alpha_{\pi}([a,1])=\phi([0,t_1(a)])+\phi ([a,1]),
\end{gather*}
as follows from $\kappa \phi ([0,t_1(a)])=(1-\kappa)\phi ([a,1])$, which implies that $\kappa \df\phi ([0,t_1(a)])=(1-\kappa)\df\phi ([a,1])$ and that $t_1$ is a continuous, strictly decreasing function. (P2) holds because, for all $a\in [\ul a,1]$, 
\[
\E_{\pi_a}[u(a,\theta)]=\E_{\pi_a}[\1 \{\theta\geq a\}-\kappa]=\pi_a([a,1])-\kappa=0.
\]
Consider now any other implementable outcome $\tilde\pi$. By (P2), there exists $\tilde\pi_a $ with $\tilde\pi_a([a,1])\geq \kappa$, as otherwise $\E_{\tilde \pi_a}[u(a,\theta)]<0$. Thus, by (P1), $\alpha_{\tilde \pi} ([a,1])\leq \phi([a,1])/\kappa$, as follows from
\[
\phi([a,1])=\int_A \tilde \pi_{\tilde a}([a,1])\df \alpha_{\tilde \pi} (\tilde a)\geq\int_{a}^1 \tilde \pi_{\tilde a}([a,1])\df \alpha_{\tilde \pi} (\tilde a)  \geq \kappa \alpha_{\tilde \pi} ([a,1]).
\]
Since $\alpha_\pi([a,1])=\phi([a,1])/\kappa$, it follows that $\alpha_{\pi}$ first-order stochastically dominates $\alpha_{\tilde \pi}$, and thus, for an increasing $V$, 
\[
\int_{A\times \Theta} V(a)\df \pi(a,\theta)=\int_{A} V(a)\df \alpha_{\pi}(a)\geq \int_{A} V(a)\df \alpha_{\tilde \pi}(a)=\int_{A\times \Theta} V(a)\df \pi(a,\theta),
\]
showing that $\pi$ is optimal.

\subsection{Proof for Example \ref{ex:segpair}} \label{s:segpair}
The optimal outcome $\pi$ is unique, because there is a unique implementable outcome $\pi$ with $\pi (\Gamma)=1$. To illustrate how the argument works more generally, we suppose that $\phi$ has a density on $\Theta=[\ul \theta,\ol \theta]$, and that there exists a bifurcation point $a_0$ in the interior of $A_\Gamma=[\ul a,\ol a]$ such that $\Gamma_a=\{t_1(a),t_2(a)\}$ with $t_1(a)=\theta^\star (a)=t_2(a)$ for $a\in [\ul a, a_0]$, and $t_1(a)<\theta^\star(a)<t_2(a)$ for $a\in [a_0,\ol a]$ where $t_1:(a_0,\ol a]\rightarrow[\ul \theta, \theta^\star(a_0))$ is continuous, strictly decreasing, and bijective and $t_2:(a_0,\ol a]\rightarrow(\theta^\star(m),\ol \theta ]$ is continuous, strictly increasing, and bijective. 

Define the continuous, strictly decreasing, and bijective inverse $t_1^{-1}:[\ul \theta,\theta^\star (a_0))\rightarrow (a_0,\ol a]$ by
\[
t_1^{-1}(\theta)=\{a\in (a_0,\ol a]:t_1(a)=\theta\}.
\]
Define the distribution functions $F(\theta)=\phi([-1,\theta])$ and $H(a)=\alpha_\pi([-1,a])$ representing measures $\phi$ and $\alpha_\pi$. 
Define the $\theta$-section of $\Gamma$ by $\Gamma^\theta=\{a\in A:(a,\theta)\in \Gamma\}$. 
Recall that, for $a\in (a_0,\ol a]$, $\pi_a=\rho_a\delta_{t_1(a)}+(1-\rho_a)\delta_{t_2(a)}$  with $\rho_a=u(a,t_2(a))/(u(a,t_2(a)-u(a,t_1(a))\in (0,1)$, by (P2).

For all $a\in (a_0,\ul a]$, we have $\Gamma^{t_2(a)}=\{a\}$ and $\Gamma_{a}=\{t_1(a),t_2(a)\}$, and thus, by (P1) and (P2),
\[
\df F(t_2(a))=(1-\rho_a)\df H(a).
\]
For all $a\in [\ul a,a_0)$, we have $\Gamma_{a}=\{\theta^\star (a)\}$ and $\Gamma^{\theta^\star (a)}=\{a, t_1^{-1}(\theta^\star(a))\}$, with $t_1^{-1}(\theta^\star(a))\in (a_0,\ol a]$ and thus, by (P1) and (P2),
\[
\df F(\theta^\star (a))=\df H(a) - \rho_{t_1^{-1}(\theta^\star (a))}\df H(t_1^{-1}(\theta^\star (a))),
\]
where the last term has a minus sign because $t_1^{-1}(\theta^\star (a))$ is decreasing in $a$. So,
\[
\df H(a)=
\begin{cases}
\frac{1}{1-\rho_a}\df F(t_2(a)), &a\in (a_0,\ol a],\\
\df F(\theta^\star (a))+\frac{\rho_{t_1^{-1}(\theta^\star (a))}}{1-\rho_{t_1^{-1}(\theta^\star (a))}}\df F(t_2(t_1^{-1}(\theta^\star (a)))), &a\in [\ul a,a_0).
\end{cases}
\]
Substituting $\theta^\star (a)=a$ for $a\in [\ul a,a_0)=[-1,0)$, and $\rho_a=1/2$, $t_1(a)=-a$, and $t_2(a)=3a$ for $a\in (a_0,\ol a]=(0,1]$, we obtain that $\alpha_\pi$ has the stated density $h$.

Finally, to see that the contact set is the stated set $\Gamma$, we invoke the following lemma from \citet{KW}.
\begin{lemma}\label{l:gerry}
Functions
\[
p(\theta)=
\begin{cases}
T(2\theta), &\theta\in [-1,0),\\
3 T(\frac 23 \theta), &\theta\in [0,3],
\end{cases}
\quad\text{and}\quad 
q(a)=
\begin{cases}
\frac{2T'(2a)}{T'(0)}, &a\in [-1, 0),\\
2, &a\in [0,3].
\end{cases}
\]
satisfy (D1) with equality for all $(a,\theta)\in\Gamma$ and strict inequality for all $(a,\theta)\notin \Gamma$.
\end{lemma}

\begin{proof}[Proof of Lemma \ref{l:gerry}]
Since $T$ is symmetric about 0 (i.e., $T(\theta-a)=-T(a-\theta)$) and $T'$ is strictly log-concave, it follows that $T'(0)>T'(y)$ for all $y\neq 0$ and $T(y)$ is strictly concave for $y\geq 0$. Hence, if $y_1'\leq y_1\leq y_2\leq y_2'$, $(y_1',y_2')\neq (y_1,y_2)$, and $\rho'y_1'+(1-\rho')y_2'=\rho y_1+(1-\rho)y_2$, for some $y_1,y_2,y_1',y_2'\geq 0$ and $\rho,\rho'\in (0,1)$, then $\rho'T(y_1')+(1-\rho)T(y_2')<\rho T(y_1)+(1-\rho)T(y_2)$, by Jensen's inequality.

We split the analysis into six cases.

(1) For $a\in [0,3]$ and $\theta\in [a,3]$, (D1) simplifies to
\[
3T(\tfrac 23 \theta) \geq T(2a) + 2T(\theta-a),
\]
which holds with equality for $\theta=3a=t_2(a)$ and strict inequality for $\theta\neq 3a$.

(2) For $a\in [0,3]$ and $\theta\in (0,a)$, (D1) simplifies to
\[
3T(\tfrac 23 \theta)+2T(a-\theta) \geq T(2a) + 4T(0),
\]
which always holds with strict inequality.

(3) For $a\in [0,3]$ and $\theta \in [-1,0]$, (D1) simplifies to
\[
2T(a-\theta) \geq  T(2a)+T(-2\theta),
\]
which holds with equality for $\theta=-a=t_1(a)$ and strict inequality for $\theta\neq -a$.

(4) For $a\in [-1,0)$ and $\theta\in [0,3]$, (D1) simplifies to
\[
3T(\tfrac 23 \theta)+ T(-2a) \geq q (a)T(\theta-a)+2T(0),
\]
which always holds with strict inequality because $q(a)<2$ and $T(\theta-a)>0$.

(5) For $a\in [-1,0)$ and $\theta\in (a,0)$, (D1) simplifies to
\[
T(-2a)  \geq T(-2\theta)+ q (a)T(\theta-a),
\]
which is equivalent to 
\[
\frac{T(-2a)-T(-2\theta)}{T'(-2a) (-2a+2\theta)}\geq \frac{T(\theta-a)-T(0)}{T'(0)(\theta-a)},
\]
which always holds with strict inequality because $T(y)$ is strictly concave for $y\geq 0$, and thus the left-hand side is strictly greater than 1 whereas the right-hand side is strictly less than 1.

(6) For $a\in [-1,0)$ and $\theta\in [-1,a]$, (D1) simplifies to
\[
T(-2a) +q (a)T(a-\theta) \geq T(-2\theta),
\]
which holds with equality for $\theta=a=t_1(a)$. For $\theta<a$, the inequality is equivalent to
\[
\frac{T(a-\theta)-T(0)}{T'(0)(a-\theta)}\geq \frac{T(-2\theta)-T(-2a)}{T'(-2a)(-2\theta+2a)},
\]
which always holds with strict inequality because
\begin{align*}
\frac{T(-2\theta)-T(-2a)}{T'(-2a)(2a-2\theta)}&=\frac{1}{2a-2\theta}\int_{0}^{2(a-\theta)} \frac{T'(y-2a)}{T'(-2a)}\df y\\
&<\frac{1}{2a-2\theta}\int_{0}^{2(a-\theta)} \frac{T'(y)}{T'(0)}\df y\\
& =\frac{T(2a-2\theta)-T(0)}{T'(0)(2a-2\theta)}\\
&<\frac{T(a-\theta)-T(0)}{T'(0)(a-\theta)},
\end{align*}
where the first inequality holds because $T'(x+y)/T'(x)$, with $y>0$, is strictly decreasing in $x$ for a strictly log-concave $T'$, and the second inequality holds because $T(y)$ is strictly concave for $y\geq 0$.
\end{proof}

\subsection{Proof of Proposition \ref{p:ZZ}}
Clearly, $a^\star(\mu)={\E_\mu [\theta]}/({1+\E_\mu[\theta^2]})$. To ensure that Assumption \ref{a:int} holds, we normalize $A=[\min_{\theta\in [\ul \theta,\ol \theta]} a^\star(\delta_\theta),\max_{\theta\in [\ul \theta,\ol \theta]} a^\star(\delta_\theta)]$. Assumptions \ref{a:smooth}, \ref{a:qc}, \ref{a:v>0} obviously hold. Moreover, since $a^\star(\delta_\theta)$ is strictly increasing on $[0,1]$ and strictly decreasing on $[1,+\infty)$, it follows that $u(a^\star(\delta_\theta),\theta')>0$ if $\theta<\theta'\leq 1$ and if $1\leq \theta'<\theta$. Thus, if $\ol \theta\leq 1$, then Assumption \ref{a:sc} holds, whereas, if $\ul \theta \geq 1$, Assumption \ref{a:sc} also holds once the state is redefined as $-\theta$. So Theorems \ref{t:pairwise}, \ref{l:ssdd}, \ref{t:NAD} and Lemma \ref{p:full} apply.

Lemma \ref{l:ZZ1} replicates Lemma 1 and Proposition 3 in \citet{ZZ}.
\begin{lemma}\label{l:ZZ1}
If $\theta_1<\theta_2$  and $\theta_1\theta_2> (<)1$, then  $\rho V(a^\star(\delta_{\theta_1}),\theta_1)+(1-\rho )V(a^\star(\delta_{\theta_2}),\theta_2)>(<)\rho V(a^\star(\mu),\theta_1)+(1-\rho )V(a^\star(\mu),\theta_2)$ for all $\rho\in (0,1)$.

$a^\star(\mu)\left({\rho }/{\theta_1}+{(1-\rho )}/{\theta_2}\right)$.

\end{lemma}
\begin{proof} 
For $\mu=\rho \delta_{\theta_1}+(1-\rho )\delta_{\theta_2}$, $a^\star(\mu)={(\rho \theta_1 +(1-\rho )\theta_2)}/{(1+\rho \theta_1^2+(1-\rho )\theta_2^2)}$. 
Thus, if $\theta_1<\theta_2$ and $\theta_1\theta_2>(<)1$, we have
\begin{gather*}
\frac{\df}{\df \rho }a^\star(\mu) =  \frac{(\theta_2-\theta_1)(\theta_1\theta_2-1)}{(1+\rho \theta_1^2+(1-\rho )\theta_2^2)^2}>(<)0,\\
\frac{\df^2}{\df \rho^2 }a^\star(\mu)  =  \frac{(\theta_2-\theta_1)(\theta_1\theta_2-1)(\theta_2^2-\theta_1^2)}{(1+\rho \theta_1^2+(1-\rho )\theta_2^2)^3}>(<)0.
\end{gather*}
Define $\varphi (\rho)=a^\star(\mu)\left({\rho }/{\theta_1}+{(1-\rho )}/{\theta_2}\right) $.
Thus, if $\theta_1<\theta_2$ and $\theta_1\theta_2>(<)1$, we have
\[
\varphi''(\rho )={\left(\frac{\rho }{\theta_1}+\frac{1-\rho }{\theta_2}\right)}{\frac{\df^2}{\df \rho^2 }a^\star(\mu)} +2{\left(\frac{1}{\theta_1}-\frac{1}{\theta_2}\right)}\frac{\df}{\df \rho }a^\star(\mu)>(<)0,
\]
so $\varphi$ is strictly convex (concave), and $\rho\varphi(1)+(1-\rho) \varphi(0)>(<)\varphi (\rho)$.
\end{proof}
If $\ul \theta\geq 1$, then $\theta_1\theta_2>1$ for all $\ul \theta_1\leq \theta_1<\theta_2$, so full disclosure is uniquely optimal by Lemmas \ref{p:full} and \ref{l:ZZ1}. Assume henceforth that $\ul \theta\leq 1$. 

After some algebra,  we get, for all $a$ and $\theta_1<\theta_2<\theta_3$,
\[
|S|=\frac{(\theta_3-\theta_2)(\theta_3-\theta_1)(\theta_2-\theta_1)(1-\theta_2\theta_3-\theta_1\theta_3-\theta_1\theta_2)}{\theta_1\theta_2\theta_3}
\]
If $\ol \theta\leq 1/\sqrt 3$ ($ \ul \theta\geq 1/\sqrt 3$), then $|S|>(<)0$ for all $\theta_1<\theta_2<\theta_3\leq \ol \theta$  ($\ul \theta\leq \theta_1<\theta_2<\theta_3$), so $\Gamma^\star$ is pairwise by Theorem \ref{t:pairwise}. Proposition 4 in \citet{ZZ} derives a version of this result for a finite set $\Theta$.

Moreover, if $\ol \theta \leq 1/\sqrt 3$ ($\ul \theta \geq 1/\sqrt 3$), then $\Gamma$ is single-dipped (-peaked), as follows from Theorem \ref{l:ssdd} with
\begin{gather*}
y=
\begin{pmatrix}
u(a_2,\theta_3)u(a_1,\theta_2)-u(a_2,\theta_2)u(a_1,\theta_3) \\
u(a_2,\theta_3)u(a_1,\theta_1)-u(a_2,\theta_1)u(a_1,\theta_3) \\
u(a_2,\theta_2)u(a_1,\theta_1)-u(a_2,\theta_1)u(a_1,\theta_2)
\end{pmatrix}\\
\begin{pmatrix}
y=-
\begin{pmatrix}
u(a_2,\theta_3)u(a_1,\theta_2)-u(a_2,\theta_2)u(a_1,\theta_3) \\
u(a_2,\theta_3)u(a_1,\theta_1)-u(a_2,\theta_1)u(a_1,\theta_3) \\
u(a_2,\theta_2)u(a_1,\theta_1)-u(a_2,\theta_1)u(a_1,\theta_2)
\end{pmatrix}
\end{pmatrix},	
\end{gather*}
because, for $a<a'$ and $\theta<\theta'$ with $\theta\theta'<1$, we have
\[
u(a',\theta')u(a,\theta)-u(a',\theta)u(a,\theta')=(a'-a)(\theta'-\theta)(1-\theta\theta')>0,
\]
and 
\[
Ry=
\begin{pmatrix}
(a_2-a_1)^2|S|\\
0\\
0
\end{pmatrix}
\gneq 
\begin{pmatrix}
0\\
0\\
0
\end{pmatrix}
\begin{pmatrix}
Ry=
\begin{pmatrix}
-(a_2-a_1)^2|S|\\
0\\
0
\end{pmatrix}
\gneq 
\begin{pmatrix}
0\\
0\\
0
\end{pmatrix}
\end{pmatrix}.
\]
Since $\Gamma^\star$ is pairwise and  $\Gamma$ is single-dipped (-peaked) if $\ol\theta \leq 1/\sqrt 3$ ($\ul \theta \geq 1/\sqrt 3$), it follow that $\Gamma^\star$ is single-dipped (-peaked) if $\ol\theta \leq 1/\sqrt 3$ ($\ul \theta \geq 1/\sqrt 3$). Finally, since, by Lemma \ref{l:ZZ1}, \eqref{e:nd} holds for all $p\in (0,1)$, Theorem \ref{t:NAD} yields that, if $\ol\theta \leq 1/\sqrt 3$ ($\ul \theta \geq 1/\sqrt 3$), then $\Gamma^\star$ is single-dipped (-peaked) negative assortative disclosure, and the optimal outcome is unique.

\subsection{Proof of Proposition \ref{c:spd1}}
Let
\[y=
\begin{cases}
\left(0,\frac{1}{(\theta_2-\theta_0)g(a_2|\theta_2)},\frac{1}{(\theta_2-\theta_0)g(a_2|\theta_3)}\right), &\theta_2>\theta_0,\\
\left(0,1,0\right), &\theta_2=\theta_0, \\
\left(\frac{1}{(\theta_0-\theta_1)g(a_1|\theta_1)},\frac{1}{(\theta_0-\theta_2)g(a_1|\theta_2)},0\right), &\theta_2<\theta_0,
\end{cases}
\]
where $\theta_1 <\theta_2<\theta_3$, $a_2 <a_1$, and $\theta_1\leq \theta _{0}\leq \theta_3$.
 We focus on the case $\theta _{0}<\theta_2$, as the other cases are analogous. The above perturbation increases action $a_1$, because, by strict log-submodularity of $g$, 
\begin{equation*}
u(a_1,\theta_2)y_2-u(a_1,\theta_3)y_3= \frac{g(a_1|\theta _2)}{g(a_2 |\theta_2)}-\frac{g(a_1|\theta _3)}{g(a_2 |\theta _3)} >
0.
\end{equation*}%
The intuition is that, since a type-$a_2$ receiver is more optimistic about the state than a type-$a_1$ receiver, he assigns higher prior probability to $\theta_3$ relative to $\theta_2$. He therefore finds a signal that puts more weight on $\theta_3$ relatively more persuasive, while the more pessimistic type-$a_2$ receiver is more persuaded by a signal that puts more weight on~$\theta_2$. 

Moreover, the same perturbation also increases the sender's expected utility for fixed $a_1 ,a_2$. This follows because
\begin{align*}
& (V(a_1,\theta_2)-V(a_2,\theta_2))y_2-(V(a_1,\theta_3)-V(a_2,\theta_3))y_3 \\
=&  \left( \frac{G(a_1|\theta_2)-G(a_2|\theta_2)}{(\theta_2-\theta _{0})g(a_2 |\theta_2)}-\frac{G(a_1|\theta_3)-G(a_2|\theta_3)}{(\theta _3-\theta _{0})g(a_2 |\theta_3)}\right)  \\
>& \frac{1}{(\theta_2-\theta _{0})}\left( \frac{G(a_1|\theta_2)-G(a_2 |\theta_2)}{g(a_2 |\theta_2)}-\frac{G(a_1|\theta_3)-G(a_2|\theta_3)}{g(a_2 |\theta_3)}\right)  \\
=& \frac{1 }{(\theta _2-\theta _{0})}\int_{a_2}^{a_1}\left( \frac{g(t|\theta_2)}{g(a_2|\theta_2)}-\frac{g(t|\theta_3)}{g(a_2 |\theta_3)}\right) \df t\geq 0,
\end{align*}
where the first inequality is by $\theta _{0}<\theta_2<\theta_3$ and the second inequality is by log-submodularity of $g$. Thus, every optimal outcome is single-peaked.

The above chain of inequalities can also be given an intuitive explanation. There are two effects that both benefit the sender. First, if we ignore the effect of $\theta$ on $g(a_2 | \theta)$ (e.g., suppose for the moment that $g(a_2| \theta_2)=g(a_2| \theta_3)$), then the fact that $(\theta-\theta_0 )$ is strictly increasing in $\theta$ implies that $y_2>y_3$. That is, to keep a type-$a_2$ receiver indifferent, the weight on $\theta_2$ that can be moved from $a_2$ to $a_1$ is greater than the weight on $\theta_3$ that moves in the opposite direction. This benefits the sender, as more weight moves from the lower receiver cutoff to the higher cutoff than moves in the opposite direction. This effect explains the first inequality. 

The second inequality comes from the fact that, since $t$ is more likely to be high when $\theta$ is low, $\Pr(t \in [a_2,a_1] | \theta )/\Pr(t=a_2 | \theta)$ is decreasing in $\theta$. That is, the higher is $\theta$, the lower is the relative probability that a decrease in the receiver's cutoff from $a_2$ to $a_1$ is pivotal for the receiver's action. Thus, moving mass on $\theta_2$ to $a_1$ while moving mass on $\theta_3$ to $a_2$ benefits the sender by inducing a higher receiver cutoff at those states where the choice of cutoff is more likely to matter for the receiver's action, and inducing a lower cutoff at states where the choice of cutoff is less likely to matter. This explains the second inequality.

\subsection{Proof of Proposition \ref{p:GL}}
As shown by \citet{KG}, there exists an optimal outcome with a finite support. Suppose the support contains a strictly single-peaked triple $(a_1,\theta_1)$, $(a_2,\theta_2)$, $(a_1,\theta_3)$, with $\theta_1<\theta_2<\theta_3$, $a_1<a_2$, and $\theta_1<a_1<\theta_3$. Notice that $V(a_1,\theta_3)\neq-\infty$ (so $a_1\geq \sigma(\theta_3))$, as otherwise the sender's expected utility would be $-\infty$, which cannot be optimal. Taking into account that $\sigma(\theta)=\theta$ for $\theta\leq \theta_0$ gives $a_1 >\theta_0$. Thus, the first row in $R$ is zero. Consider a perturbation that shifts weights $y_1=(\theta_3-\theta_2)\varepsilon$ and $y_3=(\theta_2-\theta_1)\varepsilon$ on $
\theta_1$ and $\theta_3$ from $a_1$ to $a_2$ and shifts weight $y_2=(\theta_3-\theta_1)\varepsilon$ from $a_2$ to $a_1$, where $\varepsilon$ takes the maximum value such that $y_1\leq \pi(\{(a_1,\theta_1\})$, $y_2\leq \pi(\{(a_2,\theta_2\})$, $y_3\leq \pi(\{(a_1,\theta_3\})$, so that a strictly single-peaked triple is removed. This perturbation holds fixed $a_1$ and $a_2$ and thus does not change the sender's expected utility, since the first row in $R$ is zero. Repeating such perturbations until all strictly single-peaked triples are removed (a finite number of times since $\supp (\pi)$ is finite) yields a single-dipped outcome that is weakly preferred by the sender.

\end{document}